\documentclass[11pt]{article}

\usepackage{graphicx}
\usepackage{changes}
\usepackage[margin=1in]{geometry}
\usepackage[T1]{fontenc}
\usepackage[dvipsnames]{xcolor}
\usepackage[english]{babel}
\usepackage{amsmath,amsfonts,amssymb,amsthm}
\usepackage{mathtools}
\usepackage{authblk}
\usepackage[hidelinks]{hyperref}
\usepackage{accents}
\usepackage{pdflscape}
\usepackage{listings}
\usepackage{chngcntr}
\usepackage{authblk}
\usepackage{bm}
\hypersetup{
    colorlinks=false,
    linkcolor=red,
    filecolor=magenta,
    citecolor=blue,
    urlcolor=blue,
}
\usepackage{natbib}
\usepackage{tikz}
\usetikzlibrary{calc}
\usepackage{pgfplots}
\usepackage[most]{tcolorbox}
\usepackage[toc,page]{appendix}
\usepackage{minitoc}

\pgfplotsset{compat=1.18}

\date{}
\author{Herbert P. Susmann*}
\author{Alec McClean}
\author{Iv\'{a}n D\'{i}az}

\affil[]{\small Division of Biostatistics, Department of Population
  Health, New York University Grossman School of Medicine, New York, USA\\
  *Corresponding author: \href{susmah01@nyu.edu}{susmah01@nyu.edu}
  }

\providecommand{\keywords}[1]{\textbf{\textit{Keywords: }} #1}

\mathtoolsset{showonlyrefs}

\usepackage[skip=\medskipamount]{parskip}

\newtheorem{lemma}{Lemma}
\newtheorem{proposition}{Proposition}

\newtheorem{assumption}{Assumption}
\newtheorem{theorem}{Theorem}
\newtheorem{property}{Property}
\newtheorem{algorithm}{Algorithm}

\theoremstyle{remark}
\newtheorem{remark}{Remark}

\def\exampletext{Example \thetcbcounter} % If English

\NewDocumentEnvironment{example}{ O{} }
{
\colorlet{colexam}{red!55!black} % Global example color
\newtcolorbox[use counter=example]{examplebox}{%
    empty,% Empty previously set parameters
    title={\exampletext: #1},% use \thetcbcounter to access the example counter text
    attach boxed title to top left,
       minipage boxed title,
    boxed title style={empty,size=minimal,toprule=0pt,top=4pt,left=3mm,overlay={}},
    coltitle=colexam,fonttitle=\bfseries,
    before=\par\medskip\noindent,parbox=false,boxsep=0pt,left=3mm,right=0mm,top=2pt,breakable,pad at break=0mm,
       before upper=\csname @totalleftmargin\endcsname0pt, % Use instead of parbox=true. This ensures parskip is inherited by box.
    overlay unbroken={\draw[colexam,line width=.5pt] ([xshift=-0pt]title.north west) -- ([xshift=-0pt]frame.south west); },
    overlay first={\draw[colexam,line width=.5pt] ([xshift=-0pt]title.north west) -- ([xshift=-0pt]frame.south west); },
    overlay middle={\draw[colexam,line width=.5pt] ([xshift=-0pt]frame.north west) -- ([xshift=-0pt]frame.south west); },
    overlay last={\draw[colexam,line width=.5pt] ([xshift=-0pt]frame.north west) -- ([xshift=-0pt]frame.south west); },%
    }
\begin{examplebox}}
{\end{examplebox}\endlist}

\DeclareMathOperator*{\argmin}{\arg\!\min}
\DeclareMathOperator*{\argmax}{\arg\!\max}
\DeclareMathOperator*{\var}{var}

\newcommand{\E}{\mathbb{E}}
\newcommand{\bbP}{\mathbb{P}}

\newcommand{\Var}{\mathsf{Var}}

\newcommand{\I}{\mathbb{I}} 

\newcommand{\ate}{\psi}
\newcommand{\smoothATE}{\ate_{s}}
\newcommand{\smoothL}{L_s}

\newcommand{\smoothU}{U_s}

\newcommand{\lcb}{\left\{}
\newcommand{\rcb}{\right\}}

\newcommand{\model}{\mathcal{P}}

\newcommand{\pscore}{\pi}
\newcommand{\omodel}{\mu}

\newcommand{\smooth}{s}

\newcommand{\condindep}{\mathrel{\text{\scalebox{1.07}{$\perp\mkern-10mu\perp$}}}}

\newcommand{\eif}{\varphi}

\title{Non-overlap Average Treatment Effect Bounds}
\date{} % NB: when submitting to arxiv, better to not use \today because arxiv sometimes recompiles in the future, which causes \today to update

\begin{document}

\maketitle

\begin{abstract}
    \normalsize
    The average treatment effect (ATE), the mean difference in potential outcomes under treatment and control, is a canonical causal effect. \textit{Overlap}, which says that all subjects have non-zero probability of either treatment status, is necessary to identify and estimate the ATE. When overlap fails, the standard solution is to change the estimand, and target a trimmed effect in a subpopulation satisfying overlap. When the outcome is bounded, we demonstrate that this compromise is unnecessary. We derive \textit{non-overlap bounds}: partial identification bounds on the ATE that do not require overlap. The bounds have width proportional to the size of the non-overlap subpopulation, making them informative in common scenarios when overlap violations are limited. Since the bounds are non-smooth functionals, we derive smooth approximations amenable to semiparametric efficiency theory and propose a Targeted Minimum Loss-Based estimator that is $\sqrt{n}$-consistent and asymptotically normal under nonparametric conditions. A multiplier bootstrap procedure yields uniformly valid confidence sets across all non-overlap subpopulation sizes and smoothing parameters, allowing researchers to report the tightest valid interval. Formally, we compare non-overlap confidence intervals to confidence intervals based on point estimation across multiple overlap regimes. We illustrate the method via simulation studies and real-world data applications.
\end{abstract}


\doparttoc % Tell to minitoc to generate a toc for the parts
\faketableofcontents % Run a fake tableofcontents command for the partocs

\renewcommand \thepart{}
\renewcommand \partname{}

\section{Introduction}
A paradigmatic target of counterfactual causal inference is the \textit{average treatment effect} (ATE), defined as the population difference in potential outcomes under assignment to treatment vs control \citep{rosenbaum1983propensity, Robins86}. Identifying the ATE requires three foundational assumptions. The first assumption, \textit{consistency} (also known as the \textit{stable unit treatment value assumption}, or \textit{SUTVA}) requires that each unit's observed outcome is equal to their potential outcome for their observed treatment status. The second assumption, \textit{conditional exchangeability} (also known as \textit{no unmeasured confounding} and \textit{selection on observables}), requires that all common causes of the treatment and outcome are measured. The third assumption, \textit{overlap} (also known as \textit{positivity} or \textit{common support}), requires that all units have a propensity score (the probability of receiving treatment given covariates) that is neither zero nor one. Structural violations of overlap occur when there exists a subpopulation with propensity score that is zero or one, preventing identification of the ATE. In other cases, units may have a small yet non-zero probability of receiving either treatment status. In finite samples, such practical violations of the overlap assumption manifest as strata of covariates with few or zero observed treated (or control) units \citep{zhu2021positivity}. When overlap holds structurally, but is not satisfied in practice, then the ATE is identified, yet statistical estimators may exhibit poor finite-sample performance  \citep{petersen2012positivity}.

A common strategy when faced with overlap violations (either structural or practical) is to abandon the ATE in favor of an alternative estimand that either does not require overlap or requires a weaker form of overlap \citep{petersen2012positivity}. Moving to effects defined by alternative counterfactual interventions is one way to avoid the requirement of overlap; for example, effects defined in terms of incremental propensity score interventions \citep{kennedy2019nonparametric}. Alternatively, the target population can be changed to a subpopulation in which overlap is satisfied, referred to  as the \textit{overlap} or \textit{equipoise} population \citep{greifer2023choosingcausalestimandpropensity}. Methods for estimating the treatment effect within the overlap subpopulation include propensity score trimming \citep{crump2009overlap, mcclean2025propensity} and cardinality matching \citep{visconti2018cardinality}. Overlap balancing weights \citep{li2018balancing} target the overlap population by smoothly down-weighting observations which do not satisfy overlap. The shortcoming of pivoting to an alternative estimand is that the original object of inquiry, the ATE, is abandoned \citep{rizk2025overlap}. Furthermore, popular alternative estimands are typically less interpretable than the ATE. For example, propensity score trimming targets the overlap subpopulation, which can be difficult to interpret particularly in high dimensions; indeed, methods for characterizing the overlap subpopulation is itself an active area of research   \citep{traskin2011trees,karavani2019positivitydetection, wolf2021positivitydetection}.

In this work, we propose a novel approach for identification and estimation of partial identification bounds for the ATE that does not require overlap and does not require a priori knowledge of the overlap population. Our approach makes three contributions: we derive partial identification bounds that do not require overlap, develop $\sqrt{n}$-rate estimators with uniform inference guarantees, and compare our approach theoretically and via simulations to traditional methods based on point estimators of the ATE.

\paragraph{Contribution 1: Identification.} We derive \textit{non-overlap bounds} for the ATE, identified only assuming consistency, conditional exchangeability and outcome boundedness, \emph{but not overlap}. The bounds combine a point estimate of the treatment effect within an overlap subpopulation with worst-case bounds on the complementary non-overlap subpopulation, defined in terms of a propensity score threshold. The bounds are informative whenever the non-overlap subpopulation is small relative to the treatment effect in the overlap subpopulation.

\paragraph{Contribution 2: Estimation and inference.} The non-overlap bounds are non-smooth functionals of the propensity score. We derive smooth approximations of the bounds that are amenable to semiparametric efficiency theory and propose a Targeted Minimum Loss-based Estimator (TMLE) that is $\sqrt{n}$-consistent and asymptotically normal under nonparametric conditions. We develop a multiplier bootstrap approach for constructing uniform confidence sets across a grid of propensity score thresholds, allowing valid inference based on the narrowest interval achieved over the grid. The narrowest interval width decomposes into an identification gap, smooth approximation error, and estimation uncertainty, revealing a tradeoff that, when balanced optimally, can yield intervals tighter than traditional point estimator confidence intervals under realistic overlap regimes. 

\paragraph{Contribution 3: Finite-sample performance.} In simulations designed to exhibit practical overlap violations, non-overlap bounds achieve higher power, shorter uncertainty widths, and better empirical coverage than traditional point estimator confidence intervals. We illustrate the method through an analysis of the six observational datasets, including the causal effect of right heart catheterization on mortality \citep{murphy1990support, connors1996rhc}. 

\paragraph{Prior Work}
One approach to addressing non-identifiability emphasizes the derivation of \textit{partial identification bounds} that bracket the counterfactual parameter of interest under a weaker set of assumptions \citep{richardson2014nonparametric}. For the ATE, \cite{manski1990nonparametricbounds} and \cite{Robins89} independently derived bounds (sometimes referred to as \textit{Manski bounds}) for bounded outcomes that require neither conditional exchangeability nor overlap. However, their bounds are non-informative in the sense that they always include zero, never ruling out the possibility of a null treatment effect. Sharper bounds were subsequently derived under alternative conditions, such as monotonicity in treatment or response \citep{manski1997monotone, manski2000monotonic}. Focusing on the conditional exchangeability assumption, there has been significant interest in the development of bounds on the ATE in the presence of unmeasured confounding under various sensitivity models \citep{cornfield1959cancer, schlesselman1978confounding, rosenbaum1983propensity, lin1998confounding, nabi2024sensitivity, vanderweele2011sensitivity, diaz2013assessing, bonvini2022sensitivity}. For non-overlap, there are far fewer proposals. \cite{lee2021bounding} propose bounds on the ATE based on an assumption that also relaxes consistency. Under the standard consistency assumption, their bounds reduce to non-informative Manski bounds. Closer to our work, \citealt{lu2025addressingpositivityviolationsextending} propose a sensitivity analysis for transported causal effects from a source to a target population. Their approach is conceptually similar to ours, in that they decompose the target population into overlap and non-overlap subpopulations, and bound the effect size in the non-overlap subpopulation in a sensitivity analysis. Similarly, \citealt{zivich2024synthesis} and \citealt{zivich2025accountingmissingdatapublic} divide the population into overlap and non-overlap subpopulations, and find bounds based on assuming a deterministic model in the non-overlap region. Our work differs from these approaches in the application of nonparametric efficiency theory to analyze our proposed bounds, and in providing a formal comparison of non-overlap confidence intervals to traditional methods across several overlap regimes.

Another area of related work derives bounds based on extrapolating the outcome in non-overlap regions under smoothness assumptions or shape constraints placed on the conditional mean outcome functional \citep{armstrong2021optimal}. In the context of off-policy learning, \citealt{khan2024offpolicy} derive bounds under Lipschitz continuity of the outcome model. Interestingly, they briefly describe bounds based on outcome boundedness, which resemble our proposed approach. A key difference in our approach is that we focus on cases where both the outcome conditional mean and propensity score models are estimated, and derive estimators that are double-robust.

A complementary approach to ours is to study the properties of point estimators of the ATE under overlap violations. \citep{ma2020ipw} study the asymptotic distribution of inverse probability weighted (IPW) estimators under various overlap regimes, and \citep{dorn2025weakoverlapdoublyrobust} show that a clipped augmented inverse propensity weighted (AIPW) estimator with a data-dependent threshold achieves asymptotic normality and well-calibrated confidence intervals even in regimes where the ATE is identifiable, but the semiparametric efficiency bound is infinite and $\sqrt{n}$-consistent estimation is not possible. We follow this work in analyzing our approach under the same slowly varying tails model (Assumption~\ref{assumption:slowly-varying-tails}). However, our approach is fundamentally different: rather than targeting a point estimate at a slower-than-$\sqrt{n}$ rate, we derive partial identification bounds that are identifiable even when overlap fails to hold, and whose smooth approximations can be estimated at $\sqrt{n}$-rates regardless of the overlap regime at the cost of a non-vanishing identification gap. 

The remainder of the manuscript proceeds as follows. In Section~\ref{sec:background}, we introduce notation and formally define the ATE in terms of the potential outcomes framework for causal inference. Non-overlap bounds for the ATE, and their smooth approximations, are derived in Section~\ref{sec:bounds}. An estimator based on TMLE and a multiplier bootstrap based method for uniform confidence sets is given in Section~\ref{sec:estimation}, and a theoretical comparison of non-overlap confidence intervals to traditional point estimator confidence intervals across several overlap regimes is in Section~\ref{sec:comparison-wald}. A simulation study is presented in Section~\ref{sec:simulation-study-1}, and an illustrative application studying the effect of right heart catheterization on mortality in Section~\ref{section:application}. We conclude with a discussion in Section~\ref{sec:discussion}. 

\section{Background}
\label{sec:background}

Suppose we observe $n$ i.i.d. copies $Z_1, \dots, Z_n$ of the generic random variable $Z \in \mathcal{Z}$ drawn from the law $\bbP$ falling in the nonparametric statistical model $\model$; that is, $Z \sim \bbP \in \model$. Let $Z = (X, A, Y)$, where $X$ is a vector of covariates, $A \in \{0, 1\}$ a binary treatment indicator, and $Y \in \mathbb{R}$ is an outcome.  Let $Y^0$ and $Y^1$ be the potential outcomes under counterfactual treatment assignments $A = 0$ and $A = 1$, respectively \citep{imbens2004nonparametric}. We define the \textit{propensity score} as $\pscore \equiv \pscore(X) \equiv \bbP(A = 1 \mid X)$ and the \textit{outcome model} as $\omodel_a \equiv \omodel_a(X) \equiv \E \left( Y \mid A = a, X \right)$ for $a \in \{0, 1 \}$, where $\E$ denotes expectation with respect to $\bbP$. We let $\eta = (\pscore, \omodel_0, \omodel_1)$ be the set gathering the nuisance functions. Throughout this paper, unless necessary for clarity, we will omit covariate arguments so that, e.g., $\bbP ( \pi \leq c) \equiv \bbP \{ \pi(X) \leq c \}$.

The Average Treatment Effect (ATE) is defined in terms of potential outcomes as $\psi = \E\left( Y^1 - Y^0 \right)$. Point identification of the ATE requires the following well-known set of assumptions:

\medskip

\begin{assumption}[Consistency]
    \label{assumption:consistency}
    $Y = Y^A$. 
\end{assumption}

\medskip

\begin{assumption}[No unmeasured confounding]
    \label{assumption:conditional-exchangeability}
    $Y^a \condindep A \mid X$. 
\end{assumption}

\medskip

\begin{assumption}[Overlap]
    \label{assumption:weak-positivity}
    It holds $\bbP$-almost everywhere that $0 < \pscore < 1$.
\end{assumption}

\medskip 

\noindent Under Assumptions~\ref{assumption:consistency}, \ref{assumption:conditional-exchangeability}, and \ref{assumption:weak-positivity}, the ATE is identified by the \textit{g-formula} \citep{hernan2020whatif}:
\begin{align}
    \psi = \E \lcb \E \left( Y \mid A = 1, X \right)  - \E \left( Y \mid A = 0, X \right) \rcb \equiv \E \lcb \omodel_1  - \omodel_0 \rcb. 
\end{align}
\begin{remark}[Structural vs practical overlap violations]
    Crucially, the g-formula identification result requires overlap (Assumption~\ref{assumption:weak-positivity}). \emph{Structural} overlap violations occur when $\bbP\{ \pscore = 0 \} > 0$ or $\bbP \{\pscore = 1\} > 0$. In this case, the g-formula does not hold, because $\E \{ \omodel_1 - \omodel_0 \}$ is not defined. \emph{Practical} overlap violations occur when there exists some $X$ such that $\pscore(X) \in (0,1)$, but where $\pscore(X) \approx 0$ or $\pscore(X) \approx 1$. In a specific sample, observations within strata $X$ will be unbalanced, in the sense that most or all subjects will take treatment $A=1$ (if $\pscore \approx 1$) or control $A=0$ (if $\pscore \approx 0)$. As a result, although practical overlap violations do not prevent identification, they can have a profound effect on the finite-sample statistical performance of estimators, which will typically have very large variance. 
\end{remark}

\paragraph{Notation} The indicator function is denoted $\I(a)$, equal to one when the predicate $a$ is true and zero otherwise. We denote convergence in distribution by $\leadsto$ and convergence in probability by $\stackrel{p}{\to}$. For a function $z \mapsto f(z)$, we write $\| f \|_\infty = \sup_z | f(z) |$ to denote the uniform norm, $\| f \| = \sqrt{\int f(z)^2 d\bbP(z)}$ to denote the $L_2(\bbP)$-norm,  $\E(f) = \int f(z) d\bbP$ the expectation of $f$ with respect to $\bbP$, $\bbP(A) = \E \{ \I(A) \}$ to denote the probability of $A$, where $A$ is some event, and $\bbP_n(f) = n^{-1} \sum_{i=1}^n f(Z_i)$ the empirical mean of $f$ with respect to $Z_1, \dots, Z_n$. For simplicity, we may write $\E \equiv \E_{\bbP}$. We write $a \lesssim b$ ($a \gtrsim b$) when $a \leq Cb$ ($a \geq Cb$) for a constant $C$. For $N \geq 1$, we write $[N] = \{1, \dots, N \}$.

\section{Non-overlap bounds}
\label{sec:bounds}
When the overlap assumption (Assumption \ref{assumption:weak-positivity}) does not hold, the ATE is not identifiable. However, if the potential outcomes are bounded it is possible to derive partial identification bounds for the ATE. 

\medskip

\begin{assumption}[Bounded outcome]
    \label{assumption:bounded}
    $Y \in [0, 1]$.
\end{assumption}

The assumption that $Y \in [0, 1]$ is with little loss of generality, in the sense that if $Y \in [a, b]$, then it can be transformed to fall in $[0, 1]$. We consider relaxing this assumption via a sensitivity analysis in Appendix~\ref{appendix:sensitivity-analysis}. Leveraging boundedness, the core of our approach is to split the population into two parts: one where a treatment effect is identifiable because overlap is satisfied, and one where worst-case bounds must be applied. The two parts of the population are divided via a propensity score threshold $c \in \left[ 0, \tfrac{1}{2} \right]$. The following proposition formalizes the result. 

\medskip

\begin{proposition}[Non-overlap bounds] \label{prop:bounds}
    Let $c \in \left[ 0, \tfrac{1}{2}\right]$ denote a propensity score trimming threshold. Under Assumptions~\ref{assumption:consistency}, \ref{assumption:conditional-exchangeability}, and \ref{assumption:bounded}, 
    \begin{align}
    \label{eq:exact-bounds}
        \E\left( Y^1 - Y^0 \right) &\in \Big[ L(c), U(c) \Big], 
    \end{align}
    where
    \begin{align}
        L(c) &= \psi(c) - \bbP(\pscore \geq 1 - c), \\
        U(c) &= \psi(c) + \bbP(\pscore \leq c), \text{ and } \\
        \psi(c) &= \E \left\{ \mu_1 \I (\pi > c)- \mu_0 \I(\pi < 1-c) \right\}
    \end{align}
\end{proposition}

\begin{proof}
    The argument proceeds in three steps. First, 
    \begin{align}
        \E\left( Y^1 \right) &=  \E \lcb Y^1 \I \left( \pscore > c \right) \rcb + \E\left( Y^1 \mid \pscore \leq c \right) \bbP \left( \pscore \leq c \right) \text{ and } \\
        \E\left( Y^0 \right) &= \E\lcb Y^0 \I(\pscore < 1 - c)\rcb + \E\left( Y^0 \mid \pscore \geq 1 - c\right) \bbP(\pscore \geq 1 - c).
    \end{align}
    Therefore, the ATE satisfies 
    \[
    \E(Y^1 - Y^0) =  \E \lcb Y^1 \I \left( \pscore > c \right)   - Y^0 \I(\pscore < 1 - c) \rcb + \E\left( Y^1 \mid \pscore \leq c \right) \bbP \left( \pscore \leq c \right) - \E\left( Y^0 \mid \pscore \geq 1 - c\right) \bbP(\pscore \geq 1 - c).
    \]
    Second, the bounds in the result then follow by imposing the bounds $\{ Y^0, Y^1\} \in [0,1]$ on $\E\left( Y^1 \mid \pscore \leq c \right)$ and $\E\left( Y^0 \mid \pscore \geq 1 - c\right)$. Finally, overlap is not required to identify $\E \left\{ Y^1 \I(\pi > c) - Y^0 \I (\pi < 1-c) \right\}$, and therefore the identification of the central term in the bounds, $\psi(c)$, follows under only Assumptions~\ref{assumption:consistency} and \ref{assumption:conditional-exchangeability} by standard arguments.
\end{proof}

Proposition~\ref{prop:bounds} generalizes canonical worst-case identification bounds under unmeasured confounding \citep{manski1990nonparametricbounds, Robins89}. Setting $c = 0$ recovers the bounds described in \citealt{khan2024offpolicy} in the context of off-policy learning. Our innovation is to allow the size of the non-overlap population to vary based on the threshold $c$. This induces an ``identification gap'', the difference between the upper and lower bounds:
\begin{align}
    U(c) - L(c) = \underset{\text{Identification gap}}{\underbrace{\bbP(\pscore \leq c) + \bbP(\pscore \geq 1 - c)}}.
\end{align}
As formalized in Proposition~\ref{prop:confidence-interval-width-bound}, the identification gap can be balanced against the estimation uncertainty of the treatment effect within the overlap population to yield a narrow bound.

Unlike worst-case bounds under unmeasured confounding, it is plausible that our bounds are informative, in the sense that they exclude zero and allow for sign-identification of the ATE. An important negative result in the sensitivity analysis literature states that without further assumptions, worst-case bounds under unmeasured confounding are uninformative, and must include zero \citep{manski1990nonparametricbounds}. Crucially, however, because we focus on overlap violations rather than unmeasured confounders, \emph{non-overlap bounds can exclude zero} in reasonable scenarios. Simple analysis shows that the bounds exclude zero when $\psi(c) > \bbP(\pscore \geq 1 - c)$ or $\psi(c) < -\bbP(\pscore \leq c)$; that is, when the non-overlap subpopulation is small relative to the identifiable treatment effect. 

\subsection{Smooth valid bounds}
The bounds in \eqref{eq:exact-bounds} are non-smooth because they depend on indicator functions of the propensity score. The non-differentiability of these indicator functions poses challenges for estimation under nonparametric assumptions. Although $\sqrt{n}$-rate estimation and valid inference of the bounds may be possible with a well-specified parametric model for the propensity score or under specific assumptions, such as strong margin conditions on behavior of the propensity score around the threshold $c$, general guarantees are unavailable. We address this issue by constructing smooth approximations of the lower and upper bounds, which yields pathwise differentiable parameters that can be estimated at $\sqrt{n}$-rates under general nonparametric assumptions. Importantly, we show that these smooth bounds \emph{contain the non-smooth bounds}, and therefore, although the smooth bounds may be wider than the original bounds, they still always contain the ATE.

The core of our approach is to replace the indicator functions implicit in \eqref{eq:exact-bounds} with smooth approximations.  We construct two smooth approximations. First, we construct $s_{\ell}(\cdot, c, \gamma): [0,1] \to [0,1]$ as an approximation of $\I(\pi < c)$ (``$\ell$'' for ``$\pi$ \emph{less} than $c$'') and $s_g(\cdot, c, \gamma): [0,1] \to [0,1]$ as an approximation of $\I(\pi > c)$ (``g'' for ``$\pi$ \emph{greater} than $c$''), where $\gamma > 0$ is a smoothing parameter. We provide details on these approximations subsequently. Given these functions, we define the smooth approximation of $\psi(c)$ as
\begin{equation} \label{eq:smooth-ate}
    \psi_s(c, \gamma) = \E \{ \mu_1 s_g(\pi, c, \gamma) - \mu_0 s_{\ell}(\pi, 1-c, \gamma) \}.
\end{equation}
We also approximate the additional offset terms in the bounds of \eqref{eq:exact-bounds}. We approximate $\bbP(\pi \geq 1-c) = 1 - \bbP(\pi < 1-c)$ by $1 - \E \{ s_{\ell}(\pi, 1-c, \gamma) \}$ and $\bbP(\pi \leq c) = 1 - \bbP(\pi > c)$ by $1 - \E \{ s_g(\pi, c, \gamma) \}$.  From these, we construct smooth approximations of the bounds in \eqref{eq:exact-bounds}:
\begin{align}
    \smoothL(c, \gamma) &= \smoothATE(c, \gamma) - \Big[ 1 - \E \{ \smooth_{\ell}(\pscore, 1-c, \gamma) \} \Big] \text{ and } \label{eq:smooth-lower-bound} \\
    \smoothU(c, \gamma) &= \smoothATE(c, \gamma) + \Big[ 1 - \E \{ \smooth_g(\pscore, c, \gamma) \} \Big], \label{eq:smooth-upper-bound}
\end{align}
To guarantee the resulting smooth non-overlap bounds contain the ATE we must construct the smooth approximation functions to satisfy the following property.

\medskip

\begin{property}
    \label{req:smooth-approach} 
    For smooth approximations $\smooth_\ell$ and $\smooth_g$, it holds that $\smooth_{\ell}(x, c, \gamma) \leq \I(x < c)$ and $\smooth_g(x, c, \gamma) \leq \I(x > c)$.
\end{property}

\medskip

Property~\ref{req:smooth-approach} requires that the approximations bound the respective indicator functions \emph{from below}. Not all smooth approximations satisfy this property, but it is not difficult to construct functions that do. For example:

\begin{subequations} \label{eq:example-smooth-approximations}
\begin{align}
    s_{\ell}(x, c, \gamma) &= \begin{cases}
        1, & x \leq c - \gamma, \\
        0, & x \geq c, \\
        1 - \exp\left[1 + \frac{1}{\{ (x - c)/\gamma \}^2 - 1} \right], & \text{otherwise, and}
    \end{cases} \label{eq:example-smooth-approximations-lower} \\
    s_g(x, c, \gamma) &= \begin{cases}
        1, & x \geq c + \gamma, \\
        0, & x \leq c, \\
        1 - \exp\left[1 + \frac{1}{\{ (x - c)/\gamma \}^2 - 1} \right], & \text{otherwise.}
    \end{cases} \label{eq:example-smooth-approximations-upper}
\end{align}
\end{subequations}
We use $s_{\ell}$ and $s_g$ as in \eqref{eq:example-smooth-approximations-lower} and \eqref{eq:example-smooth-approximations-upper} for the rest of the manuscript whenever a specific smooth approximation is needed. We chose this smooth approximation because the region in which they have a non-zero derivative is bounded. Figure~\ref{fig:smooth-approximation} illustrates $s_{\ell}$ and $s_g$ for several decreasing choices of $\gamma$. As $\gamma$ decreases, the smooth approximation gets closer to the relevant indicator function.

\begin{figure}[h!]
    \centering
    \begin{tikzpicture}
        \begin{axis}[
            title={(A) smooth approximations of $\I(x < c)$},
            width={0.45\textwidth},
            height={0.3\textwidth},
            axis x line=bottom,
            axis y line=left,
            xlabel={},
            ylabel={},
            domain={-1:1},
            samples=200,
            ytick={0, 1},
            ytick style={draw=none},
            xtick={0},
            %xtick style={draw=none},
            xticklabels={$c$},
            ymax=1.1,
            ymin=-0.1,
            legend pos=south west
        ]
            \addplot[only marks, mark=*, fill=gray,draw=gray] coordinates {(0, 0)};
            \addplot[only marks, mark=*, fill=white, draw=gray] coordinates {(0, 1)};
            \addplot[gray, thick, domain={-1:0}]{1};
            \addplot[gray, thick, domain={0:1}]{0};
            %\addlegendentry{$\I[x < c]$};
            
            \addplot[blue, dashed]{(x >= 0) * 0 + (x < -1) * 1 + (x >= -1) * (x < 0) * (1 - 1 / exp(-1) * exp(1 / ((x)^2 - 1)};
            \addplot[blue, dashed]{(x >= 0) * 0 + (x < -0.5) * 1 + (x >= -0.5) * (x < 0) * (1 - 1 / exp(-1) * exp(1 / ((x / 0.5)^2 - 1)};
            \addplot[blue, dashed]{(x >= 0) * 0 + (x < -0.25) * 1 + (x >= -0.25) * (x < 0) * (1 - 1 / exp(-1) * exp(1 / ((x / 0.25)^2 - 1)};
            \addplot[blue, dashed]{(x >= 0) * 0 + (x < -0.1) * 1 + (x >= -0.1) * (x < 0) * (1 - 1 / exp(-1) * exp(1 / ((x / 0.1)^2 - 1)};

        \end{axis}
    \end{tikzpicture}
    \begin{tikzpicture}
        \begin{axis}[
            title={(B) smooth approximations of $\I(x > c)$},
            width={0.45\textwidth},
            height={0.3\textwidth},
            axis x line=bottom,
            axis y line=left,
            xlabel={},
            ylabel={},
            domain={-1:1},
            samples=200,
            ytick={0, 1},
            ytick style={draw=none},
            xtick={0},
            %xtick style={draw=none},
            xticklabels={$c$},
            ymax=1.1,
            ymin=-0.1
        ]
            \addplot[only marks, mark=*, fill=gray,draw=gray] coordinates {(0, 0)};
            \addplot[only marks, mark=*, fill=white,draw=gray] coordinates {(0, 1)};
            \addplot[gray, thick, domain={-1:0}]{0};
            \addplot[gray, thick, domain={0:1}]{1};
            
            \addplot[blue, dashed]{(x <= 0) * 0 + (x > 1) * 1 + (x > 0) * (x < 1) * (1 - 1 / exp(-1) * exp(1 / ((x)^2 - 1)};
            \addplot[blue, dashed]{(x <= 0) * 0 + (x > 0.5) * 1 + (x > 0) * (x < 0.5) * (1 - 1 / exp(-1) * exp(1 / ((x / 0.5)^2 - 1)};
            \addplot[blue, dashed]{(x <= 0) * 0 + (x > 0.25) * 1 + (x > 0) * (x < 0.25) * (1 - 1 / exp(-1) * exp(1 / ((x / 0.25)^2 - 1)};
            \addplot[blue, dashed]{(x <= 0) * 0 + (x > 0.1) * 1 + (x > 0) * (x < 0.1) * (1 - 1 / exp(-1) * exp(1 / ((x / 0.1)^2 - 1)};
        \end{axis}
    \end{tikzpicture}
    \caption{Example smooth approximations $s_{\ell}(x, c, \gamma)$ (A) and $s_g(x, c, \gamma)$ (B) as defined in \eqref{eq:example-smooth-approximations-lower} and \eqref{eq:example-smooth-approximations-upper} with smoothness $\gamma \in \{ 1, 0.5, 0.25, 0.1 \}$.}
    \label{fig:smooth-approximation}
\end{figure}

The next result shows that if the smooth approximation functions satisfy Property~\ref{req:smooth-approach}, then they are valid bounds on the ATE.

\medskip

\begin{proposition}[Smooth non-overlap bounds] \label{prop:smooth-bounds}
    Under the conditions of Proposition~\ref{prop:bounds}, suppose $s_{\ell}(x, c, \gamma)$ and $s_g(x, c, \gamma)$ satisfy Property~\ref{req:smooth-approach}. Then,
    \begin{align}
        \E \left( Y^1 - Y^0 \right) \in \Big[ \smoothL(c, \gamma), \smoothU(c, \gamma) \Big],
    \end{align}
    where $\smoothL(c, \gamma)$ and $\smoothU(c, \gamma)$ are defined in \eqref{eq:smooth-lower-bound} and \eqref{eq:smooth-upper-bound}, respectively.
\end{proposition}

Proposition~\ref{prop:smooth-bounds} establishes the key guarantee, that the smooth bounds $\smoothL(c, \gamma)$ and $\smoothU(c, \gamma)$ contain the ATE while requiring only that the approximation functions satisfy Property~\ref{req:smooth-approach}, a condition easily ensured by construction. The proof (included in the appendix) illustrates why Property~\ref{req:smooth-approach} is essential: it ensures that the smooth approximation error $\psi_s(c, \gamma) - \psi(c)$ has the correct sign. Without this property, inequality \eqref{eq:smooth-approx-error} could fail, potentially yielding bounds that exclude the ATE. The use of the smooth approximation introduces another term to the size of the ATE bounds, which is now the sum of the original identification gap and a smooth approximation error term:
\begin{align}
    \label{eq:identification-gap-smooth-approximation}
    \smoothU(c, \gamma) - \smoothL(c, \gamma) &= \overset{\text{Identification gap}}{\overbrace{\bbP(\pscore \leq c) + \bbP(\pscore \geq 1 - c)}} \\
        &\ + \underset{\text{Smooth approximation error}}{\underbrace{\bbP(\pscore > c) - \E\lcb s_g(\pscore, c, \gamma)\rcb +  \bbP(\pscore < 1 - c) - \E\lcb s_\ell(\pscore, 1 - c, \gamma) \rcb}}.
\end{align}

The smooth approximation approach builds naturally on existing work with smooth approximations for non-smooth functionals \citep{van2014targeted}, though two key distinctions merit attention. First, in standard point identification settings, the smooth approximation error typically lacks clear structure: it may be positive or negative without pattern. Consequently, an estimator targeting the smooth approximation does not implicitly target the original non-smooth parameter in any specific way. Instead, researchers must employ data-adaptive smoothing techniques that ensure the approximation error vanishes as sample size grows, thereby enabling valid inference for the non-smooth functional \citep{whitehouse2025inference, levis2025covariate}. Our partial identification framework sidesteps this complication entirely: because our smooth bounds contain the non-smooth bounds by construction, we preserve valid ATE inference even when using fixed smoothing parameters. This guarantee---that smooth bounds maintain validity---has appeared previously with smooth instrumental variable bounds \citep{levis2025covariate}. While \citet{levis2025covariate} develop smooth bounds, they also achieve $\sqrt{n}$-rates for estimating non-smooth bounds under mild margin conditions. In their context, targeting non-smooth bounds directly remains both feasible and preferable from a practical standpoint. By contrast, achieving comparable rates for our non-smooth bounds would require substantially stronger assumptions, making the smooth approximation approach necessary.

The practical implementation involves balancing smooth approximation error against statistical estimation error: smaller $\gamma$ yields tighter bounds but higher estimator variance. This tradeoff is typically favorable in finite samples, as $\gamma$ can be chosen small enough that statistical error dominates approximation error.  Furthermore, our proposed inference framework (developed in Section~\ref{sec:estimation}) allows practitioners to optimize over multiple $(c,\gamma)$ combinations while preserving valid coverage for the ATE.

\subsection{Efficiency theory}
In this section, we derive the efficient influence function for the smooth bounds, which we will use in the next section to construct an estimator. For this, we leverage efficiency theory for low-dimenaional functionals within semiparametric and nonparametric models, developed in a line of research including \cite{begun1983information}, \cite{bickel97}, and \cite{vanderVaart98}, among many others. \cite{kennedy2024doublerobustreview} provides an excellent review oriented to causal inference. 

The core step in the analysis is the derivation of a particular distributional Taylor expansion. For a functional $\theta : \mathcal{P} \to \mathbb{R}$, suppose that $\theta$ is sufficiently smooth so as to admit the following expansion, for any $\bbP_1, \bbP_2 \in \mathcal{P}$:
\begin{align}
    \label{eq:von-mises}
    \theta(\bbP_1) = \theta(\bbP_2) - \E_{\bbP_2}\lcb \eif_\theta(Z; \bbP_1) \rcb + \mathsf{R}(\bbP_1, \bbP_2),
\end{align} 
where $\eif_\theta : \mathcal{Z} \to \mathbb{R}$ is a mean-zero, finite-variance gradient and $\mathsf{R} : \mathcal{P} \times \mathcal{P} \to \mathbb{R}$ is called the \textit{second-order remainder}. In a nonparametric model, there is a unique gradient that characterizes the semiparametric efficiency bound, referred to as the \textit{efficient influence function} (EIF) of $\theta$. Importantly, $\mathsf{R}$ may only depend on $\bbP_1$ and $\bbP_2$ by their products or squares of their differences, which justifies referring to $\mathsf{R}$ as second-order. We refer to \eqref{eq:von-mises} as the \textit{von Mises expansion} of $\theta$ \citep{mises1947asymptotic}, following the terminology used by \cite{fernholz1983vonmises, robins2009quadratic, kennedy2024doublerobustreview}, among others. 

The following theorem establishes that the smooth bounds in Proposition~\ref{prop:smooth-bounds} admit von Mises expansions and characterizes their EIFs.

\medskip

\begin{theorem}[Efficient influence functions]
    \label{thm:eifs} 
    Under the setup of Proposition~\ref{prop:smooth-bounds}, suppose the maps $x \mapsto s_{\ell}(x, c, \gamma)$ and $x\mapsto s_g(x, c, \gamma)$ are twice-differentiable with bounded first and second derivatives. Let $\dot s$ denote the first derivative. Then, the parameters $\smoothATE(c, \gamma)$, $\smoothL(c, \gamma)$, and $\smoothU(c, \gamma)$ admit von-Mises expansions of the form \eqref{eq:von-mises}, with, respectively, the following uncentered efficient influence functions:
    \begin{align}
        \eif_{\smoothATE}(Z, \eta, c, \gamma) &= \omodel_1 \smooth_g(\pscore, c, \gamma) - \omodel_0 \smooth_{\ell}(\pscore, 1 - c, \gamma) \\
                                    &\ + \frac{A}{\pscore} \smooth_g(\pscore, c, \gamma) \left( Y - \omodel_1 \right) - \frac{1- A}{1 - \pscore} \smooth_{\ell}(\pscore, 1 - c, \gamma) \left( Y - \omodel_0 \right) \\
                                    &\ + \lcb \omodel_1 \dot{\smooth}_g(\pscore, c, \gamma) - \omodel_0 \dot{\smooth}_{\ell}(\pscore, 1 - c, \gamma) \rcb \left( A - \pscore \right), \\
        \eif_{L_s}(Z, \eta, c, \gamma)         &= \eif_{\smoothATE}(Z, \eta, c, \gamma) - 1 + \smooth_{\ell}(\pscore, 1 - c, \gamma) + \dot{\smooth}_{\ell}(\pscore, 1 - c, \gamma) \left( A - \pscore \right), \\
        \eif_{U_s}(Z, \eta, c, \gamma)         &= \eif_{\smoothATE}(Z, \eta, c, \gamma) + 1 - \smooth_g(\pscore, c, \gamma) - \dot{\smooth}_g(\pscore, c, \gamma) \left( A - \pscore \right).
    \end{align}
\end{theorem}
Theorem~\ref{thm:eifs} provides the uncentered EIFs for the bounds. We will use them to construct debiased estimators in the next section. Beyond the assumptions of Proposition~\ref{prop:smooth-bounds}, we only require that the smooth approximation functions have bounded first and second derivatives with respect to their first argument, which can be guaranteed by construction, and is satisfied by the examples in Figure~\ref{fig:smooth-approximation}.

The efficient influence functions in Theorem~\ref{thm:eifs} take the typical form of a plug-in plus weighted residuals. Starting with $\smoothATE(c, \gamma)$, the first line of its EIF is the typical plug-in term that we could estimate and then average to obtain a plug-in estimator for $\smoothATE(c, \gamma)$. The second and third lines include weighted residuals that debias the plug-in estimator. The third line is a debiasing term that arises due to the fact that the smoothed trimmed effect depends on the estimated propensity scores via the smoothing functions $s_{\ell}$ and $s_g$. The second and third EIFs combine the EIF for $\smoothATE(c, \gamma)$ with the EIFs for $1 - \E \lcb \smooth_{\ell}(\pi, 1-c, \gamma) \rcb$ and $1 - \E \lcb \smooth_g(\pi, c, \gamma) \rcb $, respectively, and incorporate additional weighted residuals for debiasing estimators for $\E \{ s_{\ell}(\pi, 1-c, \gamma) \}$ and $\E \{ s_g(\pi, c, \gamma) \}$.

\medskip

\begin{remark}
    The EIF $\eif_{\psi_s}$ of the smooth trimmed effect $\psi_s(c, \gamma)$ can be compared to the EIF of the ATE:
    \begin{align}
        \label{eq:eif-ate}
        \eif(Z, \eta) &= \omodel_1 - \omodel_0 + \frac{A}{\pscore} \left( Y - \omodel_{1} \right) - \frac{1 - A}{1 - \pscore} \left( Y - \omodel_0 \right).
    \end{align}
    The key difference is that $\eif_{\psi}$ down-weights observations with extreme propensity scores through the smoothing functions $s_{\ell}$ and $s_g$. This down-weighting occurs both in the plug-in terms and the inverse propensity weighted residuals.
\end{remark}

\medskip

The semiparametric efficiency bounds for the lower and upper smooth non-overlap bounds are defined as the variance of their respective EIFs, which we state in the following proposition building on Theorem~\ref{thm:eifs}. 

\medskip

\begin{proposition}[Efficiency bounds]
    \label{prop:efficiency-bounds}
    Under the conditions of Theorem~\ref{thm:eifs}, the efficiency bounds $\sigma^2_L$ and $\sigma^2_U$ for the parameters $\smoothL(c, \gamma)$ and $\smoothU(c, \gamma)$, respectively, are given by
    \begin{align}
        \sigma^2_L &= \Var\lcb\mu_1 \smooth_g(\pscore, c, \gamma) - \mu_0 \smooth_\ell(\pscore, 1 - c, \gamma) + \smooth_\ell(\pscore, 1 - c, \gamma) \rcb \\
        &+ \E\lcb \frac{\Var(Y \mid A, X)}{\pscore(1 - \pscore)} \times s_g(\pscore, c, \gamma)^2 \times s_\ell(\pscore, 1 - c, \gamma)^2  \rcb \\
        &+ \E\lcb \Var(A \mid X) \times (\mu_1 \dot{s}_g(\pscore, c, \gamma) - \mu_0 \dot{\smooth}_\ell(\pscore, 1 - c, \gamma)  + \dot{s}_\ell(\pscore, 1 - c, \gamma))^2 \rcb \\
        \sigma^2_U &= \Var\lcb\mu_1 \smooth_g(\pscore, c, \gamma) - \mu_0 \smooth_\ell(\pscore, 1 - c, \gamma) - s_g(\pscore, c, \gamma) \rcb \\
        &+ \E\lcb \frac{\Var(Y \mid A, X)}{\pscore(1 - \pscore)} \times s_g(\pscore, c, \gamma)^2 \times s_\ell(\pscore, 1 - c, \gamma)^2  \rcb \\
        &+ \E\lcb \Var(A \mid X) \times (\mu_1 \dot{s}_g(\pscore, c, \gamma) - \mu_0 \dot{\smooth}_\ell(\pscore, 1 - c, \gamma) - \dot{s}_g(\pscore, c, \gamma))^2  \rcb. 
    \end{align}
\end{proposition}

\medskip

\begin{remark}
    A standard alternative approach to achieve $\sqrt{n}$-estimation of the bounds without smoothing is to adopt a margin condition near a fixed propensity score threshold $c$. Specifically, suppose that there exists an $\alpha > 0$ and $K > 0$ such that $\bbP(|\pscore - c| \leq t) \leq Kt^\alpha$ for small $t$. If the propensity score estimator satisfies $\| \hat{\pscore} - \pscore\|_{\infty} = o_{\bbP}(n^{-1/(2\alpha)})$, and estimators of $\mu_0$ and $\mu_1$ converge at rate $n^{-1/4}$ or faster, then $n^{-1/2}$-rate estimation of the non-smooth bounds is possible. We do not pursue this option further because in many practical scenarios there is not a single threshold fixed a-priori; rather, it is of interest to estimate the bounds over many thresholds, and assuming a strong margin condition for every threshold is likely unreasonable.
\end{remark}

\section{Estimation and inference}
\label{sec:estimation}

In this section, we propose debiased estimators for the smooth bounds. With fixed threshold and smoothing parameters, we establish that our estimator is $\sqrt{n}$-consistent and asymptotically Gaussian under nonparametric conditions on the propensity score and outcome regression. Then, we demonstrate that the multiplier bootstrap facilitates asymptotically valid inference across a fixed finite set of parameters, and show how to use the multiplier bootstrap to obtain a narrow confidence interval. We characterize the resulting confidence interval width as a function of the threshold and smoothing parameters, and compare to traditional onfidence intervals centered on a point estimator across several overlap regimes.

A naive plug-in estimator would construct estimators for the nuisance functions in the identification result in Proposition~\ref{prop:smooth-bounds} and plug them into the empirical average to approximate the expectation.
%e.g., a naive plug-in estimator for the lower bound $L_s(c, \gamma)$ would be
%\begin{align}
%\widehat{L}^{\mathsf{plugin}}_s(c, \gamma) = \widehat{\psi}_s(c, \gamma) - \left[ 1 - \bbP_n \{ s_{\ell}(\widehat{\pscore}, 1-c, \gamma) \} \right] 
%\end{align}
%where 
%\begin{align}
%\widehat{\psi}^{\mathsf{plugin}}_s(c, \gamma) = \bbP_n \{ \widehat{\mu}_1 s_g(\widehat{\pscore}, c, \gamma) - \widehat{\mu}_0 s_{\ell}(\widehat{\pscore}, 1-c, \gamma) \}.
%\end{align}
This plug-in estimator will inherit the convergence properties of the estimators for the outcome regression and propensity score. As a result, it can achieve $\sqrt{n}$-convergence with well-specified parametric models for the propensity score and outcome regression or with specific nonparametric assumptions. However, general guarantees are unavailable. In the previous section, we established that the smooth bounds are pathwise differentiable and admit EIFs. Therefore, we use the EIFs to debias the naive plug-in estimator, yielding new estimators that can converge at $\sqrt{n}$-rates under nonparametric assumptions.

\subsection{Targeted learning}
\label{sec:tmle}

A variety of methods exist for constructing nonparametric asymptotically normal and efficient estimators of pathwise differentiable estimands, including estimating equations, one-step estimation, double machine learning, and targeted learning \citep{robins1995efficiency, Tsiatis06, chernozhukov2016double, vanderLaan2003, vanderLaanRose11}. We focus on targeted learning because it is a substitution estimator that automatically yields estimates of the bounds that fall in $[-1, 1]$, unlike other approaches. 

Targeted learning works by iteratively updating initial nuisance estimates $\widehat \eta$ via a fluctuation model. The fluctuation model and loss function are carefully chosen to target the parameter of interest such that the final nuisance estimates $\widehat \eta^*$ solve the empirical EIF estimating equation; e.g., when targeting the smooth lower bound $\smoothL(c, \gamma)$, the final estimates $\widehat \eta^*$ satisfy
\begin{align}
    \bbP_n \lcb \eif_{L}(Z_i, \widehat \eta^*, c, \gamma) \rcb \approx 0.
\end{align}
This careful iteration debiases the resulting estimator so that it has stronger convergence guarantees than the naive plug-in estimator. Nuisance estimation can be combined with sample splitting and cross-fitting to avoid Donsker or other complexity conditions in the analysis \citep{chernozhukov2018double, zheng2011cross, chen2022debiased}. For simplicity of presentation in this section, we assume sample splitting is used, in which nuisance estimates are constructed on a held-out fold and the debiased estimator is constructed on the remaining observations. However, in practice we employ cross-fitting and cycle through the folds to retain full sample efficiency.

For the specific Targeted Minimum Loss-based Estimator (TMLE) we propose, we use the following fluctuation model and loss function, which can be seen as an extension of those in \cite{gruber2010continuoustmle}. For $\epsilon \in \mathbb{R}$ and arbitrary $\mu_0$, $\mu_1$, and $\pi$, define the fluctuation model as
\begin{align}
    \omodel_0(\epsilon) &= \mathrm{logit}^{-1}\left\{ \mathrm{logit}(\mu_0) + \frac{1}{1-\pscore} \smooth_{\ell}(\pscore, 1 - c, \gamma) \times \epsilon \right\}, \\
    \omodel_1(\epsilon) &= \mathrm{logit}^{-1}\left\{ \mathrm{logit}(\mu_1) + \frac{1}{\pscore} \smooth_g(\pscore, c, \gamma) \times \epsilon \right\} \text{, and} \\
    \pscore(\epsilon) &= \mathrm{logit}^{-1}\Big[ \mathrm{logit}(\pscore)  +  \lcb \omodel_1 \dot{\smooth}_g(\pscore, c, \gamma) - \omodel_0 \dot{\smooth}_{\ell}(\pscore, 1 - c, \gamma) + M \rcb \times \epsilon \Big],
\end{align}
where $M = \dot{\smooth}_{\ell}(\pscore, 1 - c, \gamma)$ when targeting a lower bound $\smoothL(c, \gamma)$ and $M = -\dot{\smooth}_g(\pscore, c, \gamma)$ when targeting an upper bound $\smoothU(c, \gamma)$. These fluctuation models satisfy the key property that at $\epsilon = 0$, each fluctuation reduces to $\mu_0$, $\mu_1$, and $\pi$, respectively. We then use the following loss function:
\begin{align}
    \mathcal{L}(\epsilon, Z) =& -A\log \pscore(\epsilon) - (1 - A) \log \{ 1 - \pscore(\epsilon) \} \\
    &- A\Big[ Y\log\omodel_1(\epsilon) + (1 - Y)\log \{ 1 - \omodel_1(\epsilon) \} \Big] \\
    &-(1-A) \Big[ Y\log\omodel_0(\epsilon) + (1 - Y)\log \{ 1-\omodel_0(\epsilon) \} \Big].
\end{align}
This loss function satisfies the key property that the EIF of the target parameter is contained in the linear span of the gradient of the loss function with respect to $\epsilon$ evaluated at $\epsilon = 0$. This ensures that the fluctuation step reduces bias toward the target parameter in the optimal direction, guaranteeing asymptotically efficient estimation. This fluctuation model and loss function can be used for both a bounded continuous outcome ($Y \in [0, 1]$) and a binary outcome $Y \in \{0, 1\}$. 

With the fluctuation model and loss function in hand, we next outline the TMLE. We focus on estimating $\smoothL(c, \gamma)$. An almost identical algorithm is used to estimate $\smoothU(c, \gamma)$, but where $M$ is changed in the fluctuation submodel.  
\medskip

\begin{algorithm}[TMLE] \label{alg:tmle}
Assume access to initial nuisance estimates $\widehat \eta$ learned on a separate, independent, sample.
\begin{enumerate}
    \item Initialization: set $\widehat{\eta}^* = \widehat{\eta}$.
    \item Loss minimization: find $\epsilon^*$ by empirical minimization of $\mathcal{L}$ under the fluctuation model conditional on nuisance parameter estimates $\widehat{\eta}^*$:
    \begin{align}
        \epsilon^* = \argmin_{\epsilon \in \mathbb{R}} \sum_{i=1}^n \mathcal{L}(\epsilon, Z_i)
    \end{align}
    \item Recursive update: set nuisance parameter estimates $\widehat{\eta}^* = \left( \pscore(\epsilon^\ast), \mu_0(\epsilon^\ast), \mu_1(\epsilon^\ast) \right)$ under the fluctuation model conditional on nuisance parameter estimates $\widehat{\eta}^\ast$.
    \item Repeat Steps 2-3 until convergence (i.e. when $\epsilon^* \approx 0$).
    \item Using the final nuisance estimates from step 3, construct debiased plug-in estimators
    \[
    \widehat L_s(c, \gamma) = \widehat \psi_s(c, \gamma) - \left[ 1 - \bbP_n \{ s_{\ell}(\widehat \pi^\ast, 1-c, \gamma) \} \right] 
    \]
    where
    \[
    \widehat \psi_s(c, \gamma) = \bbP_n \{ \widehat \mu_1^\ast s_g(\widehat \pi^\ast, c, \gamma) - \widehat \mu_0^\ast s_{\ell}(\widehat \pi^\ast, 1-c, \gamma) \}.
    \]
    Meanwhile, construct variance estimates $\widehat{\sigma}^2_L$ using the unbiased sample variance of $\varphi_L(Z, \widehat \eta^\ast, c, \gamma)$.
\end{enumerate}
\end{algorithm}

\medskip

The following result establishes convergence of the TMLE to a Gaussian limiting distribution centered on the truth for fixed $c$ and $\gamma$. 

\medskip

\begin{theorem}[Weak convergence, fixed $c, \gamma$] \label{thm:weak-conv}
    Fix $c \in \left(0, \tfrac{1}{2}\right), \gamma \in (0, \infty)$. Assume the nuisance convergence rates $\| \widehat{\pscore} - \pscore \| = o_{\bbP}\left( n^{-1/4} \right)$, $\| \widehat{\omodel}_0 - \omodel_0 \| = o_{\bbP}\left( n^{-1/4} \right)$, and $\|\widehat{\omodel}_1 - \omodel_1 \| = o_{\bbP}\left( n^{-1/4} \right)$, and that $\left|\tfrac{\hat{\sigma}_L(c, \gamma)}{\sigma_L(c, \gamma)} - 1\right| = o_{\bbP}(1)$ and $\left|\tfrac{\hat{\sigma}_U(c, \gamma)}{\sigma_U(c, \gamma)} - 1\right| = o_{\bbP}(1)$. 
    Then
    \begin{align}
        \frac{\sqrt{n}}{\widehat{\sigma}_L(c, \gamma)} \left( \widehat L_s(c, \gamma) - \smoothL(c, \gamma) \right) \leadsto N\left( 0, 1 \right) \quad \text{ and } \quad \frac{\sqrt{n}}{\widehat{\sigma}_U(c, \gamma)} \left( \widehat U_s(c, \gamma) - U_s(c, \gamma) \right) \leadsto N\left( 0, 1 \right).
    \end{align}
\end{theorem}

Theorem~\ref{thm:weak-conv} provides the key convergence guarantee for our TMLE in Algorithm~\ref{alg:tmle}, which relies on two key properties: first, the fluctuation model of Section~\ref{sec:tmle} ensures that the TMLE solves the empirical EIF estimation equation, eliminating first-order bias; second, cross-fitting controls the empirical process remainder, and induces the standard $n^{-1/4}$ rate double-robustness requirement. The theorem demonstrates that if the nuisance functions (the propensity score and outcome regressions) are estimated at $n^{-1/4}$ rates in root mean squared error, then our estimators for the smooth lower and upper bounds converge at $\sqrt{n}$-rates to the true bounds with Gaussian limiting distributions. Moreover, these estimators achieve the nonparametric efficiency bound, which is the variance of each bound's EIF. These nuisance convergence rate assumptions can be satisfied using flexible nonparametric estimators under standard regularity conditions such as smoothness or sparsity \citep{Gyorfietal02}. 

\subsection{Confidence intervals and width analysis}
Theorem~\ref{thm:weak-conv} suggests a straightforward confidence interval for the ATE, as the intersection of $1-\alpha/2$ one-sided intervals for the lower and upper bounds:
\begin{align}
\label{eq:confidence-interval-ate-bounds}
\left[ \widehat L_s(c, \gamma) - q_{1-\alpha/2} \sqrt{\frac{\widehat \sigma_L^2(c, \gamma)}{n}}, \quad \widehat U_s(c, \gamma) + q_{1-\alpha/2} \sqrt{\frac{\widehat \sigma_U^2(c, \gamma)}{n}} \right].
\end{align}
This confidence interval obtains valid coverage for the ATE by Theorem~\ref{thm:weak-conv}, Proposition~\ref{prop:smooth-bounds}, and a union bound. 

\medskip

\begin{remark}
    It may be possible to reduce the critical values used to construct the confidence interval, following the procedure in \citet{imbens2004confidence}. The intuition is that the ATE cannot be close to both the lower and upper bounds when they are well separated, and therefore one can consider the intersection of $1-\alpha$ one-sided intervals.
\end{remark}

\medskip

Studying the relationship between the width of the confidence intervals in \eqref{eq:confidence-interval-ate-bounds} in relation to the threshold $c$ and the smoothness parameter $\gamma$ yields insights on how to choose these parameters. The width of the $(1-\alpha)\times 100\%$ interval is given by 
\begin{align}
    \widehat{W}_s(c, \gamma) = \widehat{U}_s(c, \gamma) - \widehat{L}_s(c, \gamma) + \frac{q_{1-\alpha/2}}{\sqrt{n}}\lcb \sqrt{\widehat{\sigma}^2_L(c,\gamma)} + \sqrt{\widehat{\sigma}^2_U(c,\gamma)}  \rcb
\end{align}
Writing the oracle width as
\begin{align}
    W_s(c, \gamma) = U_s(c, \gamma) - L_s(c, \gamma) + \frac{q_{1-\alpha/2}}{\sqrt{n}}\lcb \sqrt{\sigma^2_L(c,\gamma)} + \sqrt{\sigma^2_U(c,\gamma)}  \rcb,
\end{align}
we have, under the assumptions of Theorem~\ref{thm:weak-conv}, that
\begin{align}
    \label{eq:width-consistency}
    \widehat{W}_s(c, \gamma) &= W_s(c, \gamma) + \left(\widehat{W}_s(c, \gamma) - W_s(c, \gamma) \right) = W_s(c, \gamma) + o_{\bbP}(1).
\end{align}
The following proposition provides a simple upper bound for the estimated confidence interval width that makes explicit the trade-off when choosing $c$ and $\gamma$.

\medskip

\begin{proposition}[Confidence interval width bound]
    \label{prop:confidence-interval-width-bound}
    Fix $c \in (0, \tfrac{1}{2})$ and $\gamma > 0$. By Lemma~\ref{lemma:bounded-derivatives-smooth-approximations}, for the smooth approximations \eqref{eq:example-smooth-approximations-lower} and \eqref{eq:example-smooth-approximations-upper}, there exists a constant $D > 0$ such that, for all $\pscore \in (0, 1)$, $|\dot{s}_\ell(\pscore, c, \gamma)| \leq D/\gamma$ and $|\dot{s}_g(\pscore, c, \gamma)| \leq D/\gamma$. Then, under the assumptions of Theorem~\ref{thm:weak-conv},
    \begin{align}
        \widehat{W}_s(c, \gamma) \leq & \, \bbP(\pscore \leq c + \gamma) + \bbP(\pscore \geq 1 - c - \gamma) \\
        &+ 2\frac{q_{1-\alpha/2}}{\sqrt{n}} \sqrt{1 + \frac{1}{2c} + \frac{9D^2}{4\gamma^2}} \\
        & + o_{\bbP}(1).
    \end{align}
\end{proposition}

Proposition~\ref{prop:confidence-interval-width-bound} makes clear the trade-off involved in choosing $c$ and $\gamma$. Choosing $c$ small decreases the size of the non-overlap region, but increases the efficiency bound (through the $1/2c$ term in the width bound). Choosing $\gamma$ small decreases the smooth approximation error, but also increases the efficiency bound (through the $9D^2 / 4\gamma^2$ term. In the next section, we turn to the problem of finding an optimal tradeoff in order to find a choice of $c$ and $\gamma$ yielding the narrowest interval.

\subsection{Uniform inference} 
In practice, researchers face tradeoffs when choosing a parameter combination: smaller $c$ and $\gamma$ will give tighter bounds (the identification gap and smooth approximation errors will be smaller), but these can be more difficult to estimate (the estimator variance of the treatment effect in the overlap population will increase), leading to a wider combined confidence region for the ATE.
In the previous section, we established asymptotic guarantees for fixed threshold and smoothness parameter combinations. In this section, we show uniform weak convergence over a range of thresholds and smoothness parameters. To apply this result in practice, we develop a practical inference approach based on the multiplier bootstrap that holds uniformly over a finite grid of parameter choices. This uniform inference approach allows researchers to specify a set of plausible parameter combinations and then use \emph{the narrowest confidence interval}. Finally, by allowing the grid size to go to zero, we show that the multiplier bootstrap technique recovers the uniform weak convergence result in the limit.

\medskip

\begin{theorem}[Uniform weak convergence]
    \label{thm:weak-convergence-uniform}
    Let $\mathcal{C} = \left[c_1, c_2 \right]$ for $0 < c_1 < c_2 < \tfrac{1}{2}$ and $\Gamma = [\gamma_1, \gamma_2]$ for $0 < \gamma_1 < \gamma_2 < \infty$. Let $\widehat{\sigma}^2_L(c, \gamma)$ and $\widehat{\sigma}^2_U(c, \gamma)$ denote the estimators of the efficient influence function variances 
    \begin{align}
        \sigma^2_L(c, \gamma) = \E\left[ \lcb \eif_L(Z, \eta, c, \gamma) - \smoothL(c, \gamma) \rcb^2 \right] \quad \text{ and } \quad \sigma^2_U(c, \gamma) = \E\left[ \lcb \eif_U(Z, \eta, c, \gamma) - \smoothU(c, \gamma) \rcb^2 \right],
    \end{align} respectively.
    Assume that
    \begin{enumerate}
        \item \label{assumption:differentiable-1} For all $x \in [0, 1]$ and $\gamma \in \Gamma$, the maps $c \mapsto \smooth_{\ell}(x, c, \gamma)$, $c \mapsto \dot{\smooth}_{\ell}(x, c, \gamma)$, $c \mapsto \smooth_g(x, c, \gamma)$, and $c \mapsto \dot{\smooth}_g(x, c, \gamma)$ are differentiable with bounded first derivative.
        \item \label{assumption:differentiable-2} For all $x \in [0, 1]$ and $c \in \mathcal{C}$, the maps
        $\gamma \mapsto \smooth_{\ell}(x, c, \gamma)$, $\gamma \mapsto \dot{\smooth}_{\ell}(x, c, \gamma)$, $\gamma \mapsto \smooth_{g}(x, c, \gamma)$, and $\gamma \mapsto \dot{\smooth}_g(x, c, \gamma)$ are differentiable with bounded first derivative.
        \item \label{assumption:rates} $\| \widehat{\pscore} - \pscore \| = o_{\bbP}\left( n^{-1/4} \right)$, $\| \widehat{\omodel}_0 - \omodel_0 \| = o_{\bbP}\left( n^{-1/4} \right)$, and $\|\widehat{\omodel}_1 - \omodel_1 \| = o_{\bbP}\left( n^{-1/4} \right)$.
        \item \label{assumption:sup-consistency} $\sup_{c \in \mathcal{C}, \gamma \in \Gamma} \left| \frac{\hat{\sigma}_L(c, \gamma)}{\sigma_L(c, \gamma)} - 1 \right| = o_{\bbP}(1)$ and $\sup_{c \in \mathcal{C}, \gamma \in \Gamma} \left| \frac{\hat{\sigma}_U(c, \gamma)}{\sigma_U(c, \gamma)} - 1 \right| = o_{\bbP}(1)$. 
    \end{enumerate}
    Then,
    \begin{align}
        \frac{\widehat{L}_s(c, \gamma) - \smoothL(c, \gamma)}{\widehat{\sigma}_L(c, \gamma) \slash \sqrt{n}} \leadsto \mathbb{G}_L(c, \gamma) \quad \text{ and } \quad \frac{\widehat{U}_s(c, \gamma) - U_s(c, \gamma)}{\widehat{\sigma}_U(c, \gamma) \slash \sqrt{n}} \leadsto \mathbb{G}_U(c, \gamma)
    \end{align}
    in $\ell^\infty(\mathcal{C} \times \Gamma)$, where $\mathbb{G}_L(\cdot, \cdot)$ and $\mathbb{G}_U(\cdot, \cdot)$ are mean-zero Gaussian processes with covariances given by 
    \begin{align}
        \E\lcb \mathbb{G}_L(c_1, \gamma_1) \mathbb{G}_L(c_2, \gamma_2) \rcb &= \E\lcb \widetilde{\eif}_L(Z, \eta, c_1, \gamma_1) \widetilde{\eif}_L(Z, \eta, c_2, \gamma_2) \rcb \quad \text{ and } \quad  \\
        \E\lcb \mathbb{G}_U(c_1, \gamma_1) \mathbb{G}_U(c_2, \gamma_2) \rcb &= \E\lcb \widetilde{\eif}_U(Z, \eta, c_1, \gamma_1) \widetilde{\eif}_U(Z, \eta, c_2, \gamma_2) \rcb,
    \end{align}
    and where
    \begin{align}
        \widetilde{\eif}_L(Z, \eta, c, \gamma) &= \lcb \eif_L(Z, \eta, c, \gamma) - \smoothL(c, \gamma) \rcb / \sigma_L(c, \gamma) \\
        \widetilde{\eif}_U(Z, \eta, c, \gamma) &= \lcb \eif_U(Z, \eta, c, \gamma) - \smoothU(c, \gamma) \rcb / \sigma_U(c, \gamma).
    \end{align}
\end{theorem}
The first two assumptions of the theorem must be verified for the specific choice of smooth approximations; we do so for the smooth approximations \eqref{eq:example-smooth-approximations-lower} and  \eqref{eq:example-smooth-approximations-upper} in Lemma~\ref{lemma:bounded-derivatives-smooth-approximations}. The rate conditions are the same as those required in Theorem~\ref{thm:weak-conv} for pointwise weak convergence. The final assumption is stronger than is required for the pointwise case, requiring that the variance estimators are uniformly consistent over $\mathcal{C} \times \Gamma$.

In the following algorithm, we assume as input a set of finite parameters: $\mathcal{C}_K = \{ (c_1, \gamma_1), \dots, (c_K, \gamma_K) \}$. For brevity, we'll omit $(c_k, \gamma_k)$ arguments, and instead use $k$ subscripts. For example, we will let $\widehat L_k \equiv \widehat L_s(c_k, \gamma_k)$ and $\widehat \sigma_{L,k} \equiv \widehat \sigma_{L}(c_k, \gamma_k)$. We use a multiplier bootstrap to simulate the joint distribution of all bound estimators, then calibrate our confidence intervals using the most extreme realization across all parameter combinations. This ensures that even after selecting the narrowest interval, we have appropriate coverage.

\medskip

\begin{algorithm} \label{alg:bootstrap} Given $K$ threshold and smooth parameter combinations $\{ (c_k, \gamma_k) \}_{k=1}^{K}$:
    \begin{enumerate}
        \item For each $k \in [K]$, construct efficient influence function estimates as in the prior section: \( \left\{ \varphi_{L,k}(Z_i, \widehat \eta),  \varphi_{U,k}(Z_, \widehat \eta) \right\}_{i=1}^{n} \). Also construct point estimates and standard error estimators for the upper and lower bounds: $\left\{ \widehat L_k, \widehat U_k \right\}$ and $\left\{ \widehat \sigma_{L,k}, \widehat \sigma_{U,k} \right\}$.
        \item For $B$ bootstrap samples, draw i.i.d.\ multipliers $\{\xi_i^{(b)}\}_{i=1}^n$ with $\E (\xi) =0$, $\E ( \xi^2) =1$, and form the  studentized residuals for each index $k$: 
        \[
        T_{L,k}^{(b)}\ =\ \frac{1}{\sqrt n}\sum_{i=1}^n \xi_i^{(b)}\left\{ 
        \frac{ \varphi_{L,k}(Z_i, \widehat \eta) - \widehat L_k }{\widehat \sigma_{L,k}} \right\},
        \qquad
        T_{U,k}^{(b)}\ =\ \frac{1}{\sqrt n}\sum_{i=1}^n \xi_i^{(b)}
        \left\{ \frac{ \varphi_{U,k}(Z_i, \widehat \eta) - \widehat U_k }{\widehat\sigma_{U,k}} \right\}.
        \]
        \item Compute the bootstrap max–max statistic
        \[
        \widehat M^{(b)}\ :=\ \max_{1\le k\le K}\ \max\Big( \, T_{L,k}^{(b)}\,,\, -\, T_{U,k}^{(b)}\,\Big).
        \]
        Let $\widehat q_{1-\alpha}$ be the empirical $(1-\alpha)$ quantile of $\{ \widehat M^{(b)}\}_{b=1}^B$.
        \item Construct a uniform confidence set across all parameter sets as 
        \[
        \left\{ \left[ \widehat L_k - \widehat{q}_{1-\alpha} \frac{\widehat{\sigma}_{L,k}}{\sqrt{n}}, \widehat{U}_k + \widehat{q}_{1-\alpha} \frac{\widehat{\sigma}_{U,k}}{\sqrt{n}} \right] \right\}_{k=1}^{K}.
        \]
        \item Construct the narrowest confidence interval for the ATE as
        \[
        \widehat{CI} = \left[ \max_{k} \lcb \widehat L_k - \widehat q_{1-\alpha} \frac{\widehat{\sigma}_{L,k}}{\sqrt{n}} \rcb, \min_k \lcb \widehat U_k + \widehat q_{1-\alpha} \frac{\widehat \sigma_{U,k}}{\sqrt{n}} \rcb \right].
        \]
    \end{enumerate}
\end{algorithm}
Algorithm~\ref{alg:bootstrap} combines the TMLE in Algorithm~\ref{alg:tmle} with the multiplier bootstrap. It outputs two quantities: a uniform confidence set across all parameters, and the narrowest CI for the ATE. A key feature of our approach is the construction of the max-max statistic using only one side of each studentized residual: $T_{L,k}$ for lower bounds and $-T_{U,k}$ for upper bounds. This asymmetric construction reflects the different types of estimation errors we seek to control. For the lower bounds, the relevant error event is over-estimation (when $\widehat L_k$ exceeds the true $L_k$), which corresponds to positive deviations in $T_{L,k}$. For the upper bounds, the critical error is under-estimation (when $\widehat U_k$ falls below the true $U_k$), corresponding to negative deviations in $T_{U,k}$. By incorporating only the relevant tail of each distribution, this construction yields one-sided simultaneous confidence bands that cover the set $[L_{\max}, U_{\min}]$ with the desired coverage probability.

The next result establishes that the uniform set and narrowest CI facilitate asymptotically valid inference. 

\medskip 

\begin{theorem}[Asymptotic validity of narrowest bootstrap interval] \label{thm:bootstrap}
    Suppose the conditions of Theorem~\ref{thm:weak-conv} hold and a uniform confidence set and narrow confidence interval are constructed according to Algorithm~\ref{alg:bootstrap}. Further suppose $\widehat \sigma_{L,k} \stackrel{p}{\to} \sigma_{L,k}$ and $\widehat \sigma_{U,k} \stackrel{p}{\to} \sigma_{U,k}$ for all $k \in [K]$. Then, 
    \[
    \bbP \left\{ \forall k \in [K],\; \E(Y^1 - Y^0) \in \left[ \widehat L_k - \widehat q_{1-\alpha} \frac{\widehat{\sigma}_{L,k}}{\sqrt{n}}, \widehat U_k + \widehat{q}_{1-\alpha} \frac{\widehat{\sigma}_{U,k}}{\sqrt{n}} \right]  \right\} \geq 1 - \alpha + o_{\bbP}(1).
    \]
    Moreover, 
    \[
    \bbP \left\{ \E(Y^1 - Y^0) \in \widehat{CI} \right\} \geq 1 - \alpha + o_{\bbP}(1). 
    \]
\end{theorem}

Theorem~\ref{thm:bootstrap} establishes the key inferential guarantees, which because of the finite-grid approximation only require the pointwise conditions of Theorem~\ref{thm:weak-conv}. First, it shows that the uniform confidence set attains valid coverage across all bounds.  As a result, the narrowest CI for the ATE, constructed using the maximum lower CI bound and minimum upper CI bound, also maintains valid coverage. The crucial caveat is that this comes at a cost: the critical value $\widehat q_{1-\alpha}$ will be larger than its pointwise counterpart. In practice, we have seen roughly 10\% increases (e.g., from $1.96$ to approximately $2.1$) with $K = 162$. This tradeoff is valuable when researchers are uncertain about the optimal balance between smooth approximation error and estimator variance, and this balance varies substantially across the parameter space. Conversely, when prior analyses or diagnostic tools provide clear guidance toward a single parameter combination, the pointwise estimators from the prior section are preferable. In the next section, we characterize the relationship between the non-overlap bound confidence interval width and point estimator confidence interval width across different overlap regimes, providing theoretical guidance on when the uniform inference approach is most beneficial.

A natural question is whether the finite-grid bootstrap procedure of Algorithm~\ref{alg:bootstrap} approximates the theoretically optimal uniform inference as the grid becomes dense. The following proposition confirms this, showing that the bootstrap critical value converges to the quantile of the supremum of the limiting Gaussian process as the grid mesh goes to zero.

\medskip

\begin{proposition}[Bootstrap recovery of uniform inference]
    \label{prop:bootstrap-uniform}
    Suppose the conditions of Theorem~\ref{thm:weak-convergence-uniform} hold. Let $\{ \mathcal{C}_K \}_{K=1}^\infty$ be a sequence of finite grids $\mathcal{C}_K = \{ c_k, \gamma_k \}_{k=1}^K \subset \mathcal{C} \times \Gamma$ with mesh size $\delta_K = \max_{(c, \gamma) \in \mathcal{C} \times \Gamma} \min_{1 \leq k \leq K} \lcb |c - c_k| + |\gamma - \gamma_k| \rcb \to 0$ as $K \to \infty$. Let $\widehat{q}_{1-\alpha}^{K}$ denote the bootstrap critical value from Algorithm~\ref{alg:bootstrap} applied to the grid $\mathcal{C}_K$, and let $q_{1-\alpha}^\infty$ denote the $(1-\alpha)$ quantile of $\sup_{(c, \gamma) \in \mathcal{C} \times \Gamma} \max\left( \mathbb{G}_L(c, \gamma), -\mathbb{G}_U(c, \gamma)\right)$, where $\mathbb{G}_L$ and $\mathbb{G}_U$ are the limiting Gaussian processes of Theorem~\ref{thm:weak-convergence-uniform}. Assume $B \to \infty$. Then, as $n, K \to \infty$,
    \begin{align}
        \widehat{q}_{1-\alpha}^{K} = q_{1-\alpha}^\infty + o_{\bbP}(1).
    \end{align}
\end{proposition}

\section{Comparison of non-overlap bounds to standard methods across overlap regimes}
\label{sec:comparison-wald}
When the ATE is not point identifiable, due to structural overlap violations, then non-overlap bounds remain valid even when point estimation is impossible.
When the ATE is point identified, however, it is of interest to compare our methods to traditional methods based on point estimation. In this section, we compare the confidence intervals from non-overlap bounds to traditional confidence intervals constructed using the empirical variance of the EIF and centered around an asymptotically normal and efficient estimator of the ATE. We organize the comparison by considering three overlap regimes in which the ATE is identified: strong overlap ($\pscore$ bounded away from 0 and 1), somewhat weak overlap (propensity score tails decay polynomially), and very weak overlap. We find conditions under which the non-overlap confidence intervals are asymptotically equivalent to point estimator confidence intervals under strong overlap (Theorem~\ref{thm:equiv-tmle}), remain valid even when $\sqrt{n}$-inference for point estimators breaks down entirely under very weak overlap (Theorem~\ref{thm:very-weak-overlap}), and, interestingly, can yield shorter intervals in finite samples near the phase transition between somewhat weak and very weak overlap (Theorem~\ref{thm:somewhat-weak-overlap}). 

\subsection{Strong overlap}

The strong overlap assumption represents a best-case regime for traditional doubly-robust estimators, and is typically invoked as a simple condition to ensure that the efficiency bound is finite.

\medskip

\begin{assumption}[Strong overlap]
    \label{assumption:strong-overlap}
    There exists a $\delta > 0$ such that $\bbP\lcb \delta \leq \pscore(X) \leq 1 - \delta \rcb = 1$. 
\end{assumption}

We establish in the following theorem that under strong overlap, the non-overlap confidence intervals formed via Algorithm~\ref{alg:bootstrap} may be asymptotically equivalent to point estimator confidence intervals.

\medskip

\begin{theorem}[Asymptotic equivalence of non-overlap bounds and point estimator confidence intervals under strong overlap]
\label{thm:equiv-tmle}
Let $\widehat\psi$ denote an asymptotically normal and efficient estimator of the ATE based on the efficient influence function in \eqref{eq:eif-ate}, constructed using the same nuisance estimates and cross-fitting folds as the non-overlap bounds of Algorithm \ref{alg:bootstrap}. Let $\widehat\sigma_\psi$ denote the standard error estimator of the estimated EIF. Denote the endpoints of $\widehat{CI}$ constructed using the smooth approximations \eqref{eq:example-smooth-approximations-lower} and \eqref{eq:example-smooth-approximations-upper}  as
\begin{align}
    \underline{C}_n = \max_{1 \le k \le K}
    \lcb \widehat{L}_k - \widehat{q}_{1-\alpha} \frac{\widehat{\sigma}_{L,k}}{\sqrt n} \rcb,
    \qquad
    \overline{C}_n = \min_{1 \le k \le K}
    \lcb
    \widehat{U}_k + \widehat{q}_{1 - \alpha} \frac{\widehat{\sigma}_{U,k}}{\sqrt n}
    \rcb,
\end{align}
and the endpoints of the standard EIF-based point estimator confidence interval $\widehat{CI}_{\psi}$ as
\begin{align}
    \underline{C}_{\psi,n} = \widehat{\psi} - z_{1 - \alpha / 2}\frac{\widehat{\sigma}_\psi}{\sqrt{n}},
    \qquad
    \overline{C}_{\psi,n} = \widehat{\psi} + z_{1-\alpha/2}\frac{\widehat\sigma_\psi}{\sqrt{n}}.
\end{align}
Assume strong overlap (Assumption~\ref{assumption:strong-overlap}) and that:
\begin{enumerate}
    \item The propensity score estimator is uniformly consistent:
    $\|\widehat\pi-\pi\|_\infty=o_{\bbP}(1)$.
    \item The number of bootstrap draws satisfies $B\to\infty$.
\end{enumerate}
Then:
\begin{enumerate}
    \item \label{consequence:bounded-threshold-smoothness} If the finite collection of threshold--smoothness pairs $\{(c_k,\gamma_k)\}_{k=1}^K$ satisfies, for some $\varepsilon>0$, $\max_{1\le k\le K}(c_k + \gamma_k)\le \delta - \varepsilon$, then
    \begin{align}
        \sqrt n\,(\underline{C}_n - \underline{C}_{\psi,n}) = o_{\bbP}(1),
        \qquad
        \sqrt n\,(\overline{C}_n-\overline C_{\psi, n}) = o_{\bbP}(1).
        \label{eq:endpoint-convergence}
    \end{align}
    \item \label{consequence:bounded-threshold-smoothness-weak} If there exists at least one $k \in \{1, \dots, K \}$ such that for some $\varepsilon>0$, $c_k + \gamma_k \le \delta - \varepsilon$, then
    \begin{align}
        \sqrt n\,(\underline{C}_n - \underline{C}_{\psi,n}) \leq o_{\bbP}(1),
        \qquad
        \sqrt n\,(\overline{C}_n-\overline C_{\psi, n}) \geq o_{\bbP}(1).
        \label{eq:endpoint-convergence-weak}
    \end{align}
\end{enumerate}
\end{theorem}

Theorem~\ref{thm:equiv-tmle} provides welcome reassurance that, when strong overlap holds, there is no (asymptotic) penalty for uncertainty quantification using non-overlap bounds compared to a traditional point estimator confidence interval, as long as one uses threshold-smoothness pairs that are bounded by the overlap constant. Part \ref{consequence:bounded-threshold-smoothness} shows that if all threshold--smoothness pairs satisfy $c_k + \gamma_k \leq \delta - \epsilon$, then the non-overlap confidence interval is asymptotically exactly equivalent to a point estimator confidence interval. This is intuitive: in this regime, every estimated propensity score falls well within the overlap region, and the bounds collapse to the standard doubly-robust point estimate. Part~\ref{consequence:bounded-threshold-smoothness-weak} relaxes this to require that \textit{at least one} pair falls below the strong overlap threshold. Here, the non-overlap confidence interval is asymptotically at least as wide as the point estimator confidence interval, with any additional width driven by inflation of the bootstrap critical value from pairs with $c_k + \gamma_k > \delta$. Since the strong overlap threshold $\delta$ is typically unknown in practice, Part~\ref{consequence:bounded-threshold-smoothness-weak} ensures that including conservatively large thresholds in the grid preserves asymptotic validity at the cost of a modest width inflation from the bootstrap procedure. 

\subsection{Weak overlap regimes}
Next, we study a version of the slowly varying tails model of \citealt{dorn2025weakoverlapdoublyrobust} (who generalized the assumption used in \citealt{ma2020ipw}), which allows us to study more nuanced overlap violations:

\medskip

\begin{assumption}[Slowly varying tails]
    \label{assumption:slowly-varying-tails}
    Assume that every $\bbP \in \model$ satisfies, for some $\sigma_{\min} > 0$, $C > 0$, and $\gamma_0 > 0$:
    \begin{enumerate}
        \item Non-trivial outcome model: $\Var(Y \mid A, X) \geq \sigma^2_{\min}$. 
        \item Propensity score tails: for all $t \in [0, 1]$, $\bbP(\pscore \leq t) \leq C t^{\gamma_0 - 1}$ and $\bbP(1 - \pscore \leq t) \leq C t^{\gamma_0 - 1}$.
    \end{enumerate}
\end{assumption}
The first part of the assumption is required to rule out trivial cases in which the outcome is perfectly predictable. The second part controls the tail behavior of the propensity scores (here, we extend the original model to consider both tails). Note that the moment boundedness conditions in \citealt{dorn2025sensitivity}'s model (their Assumption 1) are already satisfied in our setting by outcome boundedness (our Assumption~\ref{assumption:bounded}). Following \citealt{dorn2025sensitivity}, we refer to $\gamma_0 > 2$ as \textit{somewhat weak overlap} and $\gamma_0 < 2$ as \textit{very weak overlap}, with a key phase transition in the behavior of the efficiency bound when $\gamma_0 = 2$. Figure~\ref{fig:overlap-regimes} visually illustrates strong overlap, somewhat weak, and very weak propensity score densities. In very weak overlap regimes, the subject of the following theorem, there is enough mass in the tails of the propensity score distribution to prevent asymptotically normal and efficient estimation of the ATE using standard doubly-robust estimators. 

\medskip

\begin{figure}[ht]
    \centering
    \begin{tikzpicture}
        \begin{axis}[
            name=panelA,
            title={(A) Strong overlap},
            width={0.3\textwidth},
            height={0.25\textwidth},
            axis x line = bottom,
            axis y line = left,
            xlabel={$\pi$},
            ylabel={Density},
            domain={0:1},
            samples=200,
            xtick={0, 0.2, 0.8, 1},
            xticklabels={$0$,$\delta$,$1{-}\delta$,$1$},
            ytick=\empty,
            ymin=0,
            ymax=5,
            xmin=0,
            xmax=1,
            clip=false,
        ]
            \addplot[blue,thick,domain={0.2:0.8}]{1260*x^4*(1-x)^4};
            \addplot[blue,thick,domain={0:0.2}]{0};
            \addplot[blue,thick,domain={0.8:1}]{0};
        \end{axis}

        \begin{axis}[
            name=panelB,
            title={(B) Somewhat weak overlap},
            at={($(panelA.south)+(4cm,0cm)$)},
            width={0.3\textwidth},
            height={0.25\textwidth},
            axis x line=bottom,
            axis y line=left,
            xlabel={$\pi$},
            ylabel={Density},
            domain={0.001:0.999},
            samples=200,
            xtick={0, 1},
            xticklabels={$0$,$1$},
            ytick=\empty,
            ymin=0,
            ymax=2.5,
            xmin=0,
            xmax=1,
            clip=false,
        ]
            \addplot[blue,thick]{6*x*(1-x)};
        \end{axis}

        \begin{axis}[
            name=panelC,
            title={(C) Very weak overlap},
            at={($(panelB.south)+(4cm,0)$)},
            width={0.3\textwidth},
            height={0.25\textwidth},
            axis x line = bottom,
            axis y line = left,
            xlabel={$\pi$},
            ylabel={Density},
            domain={0.01:0.99},
            samples=200,
            xtick={0, 1},
            xticklabels={$0$,$1$},
            ytick=\empty,
            ymin=0,
            ymax=4,
            xmin=0,
            xmax=1,
            clip=false,
        ]
            \addplot[blue,thick]{1 / (3.14159 * sqrt(x*(1-x)))};
        \end{axis}
    \end{tikzpicture}
    \caption{Illustration of propensity score densities under the three overlap regimes considered in Section~\ref{sec:comparison-wald}. (A)~Strong overlap: the propensity score density is supported on $[\delta, 1 - \delta]$ (Assumption~\ref{assumption:strong-overlap}. (B)~Somewhat weak overlap: the density extends to $0$ and $1$ with polynomial tails satisfying $\gamma_0 > 2$ (Assumption~\ref{assumption:slowly-varying-tails}). (C)~Very weak overlap: the density has polynomial tails satisfying $\gamma_0 < 2$ (Assumption~\ref{assumption:slowly-varying-tails}). }
    \label{fig:overlap-regimes}
\end{figure}

\medskip

\begin{theorem}[Efficiency bound comparison under very weak overlap]
    \label{thm:very-weak-overlap}
    Suppose Assumption~\ref{assumption:slowly-varying-tails} holds for $\gamma_0 \in (1, 2)$. Assume further that there exists a $\bbP \in \model$ and $C' > 0$ such that $\bbP(\pscore \leq t) \geq C' t^{\gamma_0 - 1}$ and $\bbP(\pscore \geq 1 - t) \geq C' (1 - t)^{\gamma_0 - 1}$ for all $t \in (0, 1]$. Then
    \begin{enumerate}
        \item The efficiency bound for the ATE is infinite.
        \item Assume conditions 1 and 2 of Theorem~\ref{thm:weak-convergence-uniform} (bounded first derivatives of smooth approximations). Then the efficiency bounds for the smooth non-overlap bounds are finite: for any $c \in (0, \tfrac{1}{2})$ and $\gamma > 0$, there exists a $D > 0$ such that $\sigma^2_L(c, \gamma) < D$ and $\sigma^2_U(c, \gamma) < D$.
    \end{enumerate}
\end{theorem}
\begin{proof}
    Part 1 is proved by a simple extension of \citealt[Proposition 1]{dorn2025sensitivity}, who proved the same result considering the lower tail of the propensity score distribution. Part 2 holds by applying the boundedness of the derivatives of the smooth approximations to the definition of the efficiency bounds (Proposition~\ref{prop:efficiency-bounds}). 
\end{proof}
Theorem~\ref{thm:very-weak-overlap} establishes an important contrast: under very weak overlap, no regular estimator of the ATE can achieve $\sqrt{n}$-convergence, yet the smooth non-overlap bounds can be estimated at $\sqrt{n}$-rates with asymptotically valid confidence intervals. Non-overlap bounds therefore provide a viable path to inference on the ATE in settings where traditional approaches fail.

The previous results establish the extreme cases: non-overlap bounds confidence intervals match those of point estimators  under strong overlap, and remain valid even when $\sqrt{n}$-rate point estimation is not possible under very weak overlap. The natural next question is how the approaches compare in the intermediate regime where the ATE efficiency bound is finite, but large ($\gamma_0 > 2$). Interestingly, building on the width decomposition of Proposition~\ref{prop:confidence-interval-width-bound}, we find that the best performing method involves a delicate tradeoff between sample size and the proximity to phase transition ($\gamma_0 = 2$).

\medskip

\begin{theorem}[Efficiency bound comparison under somewhat weak overlap]
    \label{thm:somewhat-weak-overlap}
    Suppose Assumption~\ref{assumption:slowly-varying-tails} holds for some $\sigma^2_{\min} > 0$, $C > 0$, and $\gamma_0 > 2$. Assume further that there exists a constant $C' \leq C$ such that, for all $\bbP \in \model$ and all $t \in [0, 1]$, $\bbP(\pscore \leq t) \geq C' t^{\gamma_0 - 1}$ and $\bbP(1 - \pscore \leq t) \geq C' t^{\gamma_0 - 1}$.  Then
    \begin{enumerate}
        \item The width of the point estimator confidence interval, given by $\overline{C}_{\psi, n} - \underbar{C}_{\psi,n}$, scales as $O\left(n^{-1/2}\right)$, with a constant proportional to $(\gamma_0 - 2)^{-1/2}$. 
        \item The width of the non-overlap confidence interval $\widehat{W}_s(c, \gamma)$, optimized over $c$, scales as $O\left(n^{-(\gamma_0 - 1) / (2 \gamma_0 - 1)}\right)$. 
        \item For any $n$, there exists a $\gamma_0(n) > 2$ such that, for any $\gamma_0 \in (2, \gamma_0(n))$, the width of the non-overlap confidence interval $\widehat{W}_s(c, \gamma)$, optimized over $c$, is shorter than the point estimator confidence interval. 
    \end{enumerate}
\end{theorem}
The implication of Theorem~\ref{thm:somewhat-weak-overlap} is that, in the somewhat weak overlap setting, the width of EIF-based point estimator confidence intervals converges at a faster rate ($O(n^{-1/2})$) than the optimal non-overlap confidence intervals ($O(n^{-(\gamma_0 - 1)/(2\gamma_0 - 1)})$) as $n \to \infty$ (parts 1 and 2). However, for any finite sample size $n$, there exists a range of $\gamma_0$ values near the phase transition at $\gamma_0 = 2$ for which the non-overlap bounds yield shorter intervals than the point estimator intervals (part 3). This reflects the fact that, near the phase transition, the constant in the point estimator interval diverges as $(\gamma_0 - 2)^{-1/2}$, while the non-overlap bounds remain well-behaved. We investigate this behavior empirically in the simulation studies to follow.

\section{Simulation Studies}
\label{sec:simulation-study-1}
The simulation studies compare the non-overlap bounds to a doubly-robust estimator of the ATE in settings designed to exhibit a range of overlap regimes. As a benchmark, we applied a doubly-robust estimator based on the uncentered EIF of the ATE, given by
\begin{align}
    \label{eq:one-step}
    \widehat{\psi}^{\mathsf{dr}} = \frac{1}{n} \sum_{i=1}^n \widehat{\phi}(Z_i),
\end{align}
where $\widehat{\phi}$ is an estimate of the uncentered EIF. Asymptotically valid $95\%$ confidence intervals are formed as $\widehat{\psi}^{\mathsf{dr}} \pm q_{0.975} \sqrt{\var(\widehat{\phi}) \slash n}$, where $q_{0.975}$ is the 97.5\% quantile of the standard normal distribution and $\var(\widehat{\phi}) $ is the empirical variance of the estimated uncentered EIF. The width of the one-step 95\% confidence interval is $W^{\mathsf{dr}} = \min(2, 2 \times q_{0.975} \times \sqrt{\var(\widehat{\phi}) \slash n})$. We truncated the width to have a maximum of $2$, indicating the interval covered the entire parameter space. We rejected the null hypothesis of that the ATE equals zero if the estimated 95\% confidence interval excluded zero.

In both simulation studies, the non-overlap ATE bounds were estimated with a logarithmic grid of propensity score thresholds from $c = 10^{-4}$ to $c = 0.05$. A uniform 95\% confidence set over the thresholds was calculated using the proposed multiplier bootstrap method with $1000$ bootstrap draws. For comparison with the one-step estimator, the uncertainty interval width for the non-overlap ATE bounds at each $\gamma$ was taken to be the tightest uniform 95\% confidence interval over the tested propensity score thresholds, also truncated to have maximum of 2: $W^{\mathsf{bounds}}_\gamma = \min(2, \min_{c} \widehat{\smoothU}(c, \gamma) - \max_{c} \widehat{\smoothL}(c, \gamma))$. We rejected the null hypothesis of a null treatment effect if any of the uniform 95\% confidence intervals over the propensity score thresholds excluded zero.

\subsection{Simulation Study 1}
First, we compared the finite-sample performance of the one-step estimator and non-overlap bounds with respect to a data-generating process (DGP) designed to exhibit extreme practical overlap violations. Simulated datasets are composed of $n$ i.i.d. copies of $Z = (X_1, X_2, A, Y)$ drawn from the joint law characterized by
\begin{align}
    \label{eq:simulation-dgp}
    X_1 &\sim \mathrm{Uniform}(-1, 1), \\
    X_2 &\sim \mathrm{Categorical}(\{ -1, 0, 1 \}, \{ 0.05, 0.9, 0.05 \}), \\
    A \mid X_1, X_2 &\sim \mathrm{Bernoulli}\left( \mathrm{logit}^{-1}\left( X_1 + \alpha \times X_2 \right)  \right), \\
    Y \mid A, X_1, X_2 &\sim \mathrm{Bernoulli}\left( \mathrm{logit}^{-1}\left( 0.5 + X_1 + A) \right) \right),
\end{align}
where $\alpha$ controls the severity of non-overlap. For the main simulation study we set $\alpha = 5$ to induce extreme positivity violations. In Appendix~\ref{section:additional-results}, we present results for simulations with $\alpha =1$ to evaluate the non-overlap bounds in a setting where overlap is satisfied. The true ATE under the DGP above is $\psi \approx 0.227$.  

The sampling distribution of doubly-robust estimators of the ATE can be highly skewed under DGPs exhibiting extremely variable propensity scores \citep{robins2007comment}. For the one-step estimator \eqref{eq:one-step}, the reason this occurs can be easily seen by analyzing the term involving inverse propensity weights of the uncentered EIF. Define the observed inverse propensity weight as
\begin{align}
    \label{eq:propensity-weight}
    r_i = \left( \frac{A_i}{\pscore(X_i)} - \frac{1 - A_i}{1 - \pscore(X_i)}\right),
\end{align}
which enters into the one-step estimator by weighting the residual, $Y_i - \mu_A(X_i)$ (see \eqref{eq:one-step}). The residual will be non-zero assuming that the conditional variance of $Y_i$ given $A_i$ and $X_i$ is non-zero. For the inverse propensity weight $r_i$ to be large, there must be an observation with both $A_i = 1$ and $\pi(X_i)$ small (or $A_i = 0$ and $1 - \pi(X_i)$ small). By definition, this is a rare event: it is unlikely to observe $A = 1$ when $\pi(X)$ is small (and vice-versa for the opposite case). In the event that this occurs, then $r_i$ is large, leading to high variance in the estimator and skewness of the sampling distribution. In addition, the empirical variance of the estimated EIF will be large, leading to wide and uninformative confidence intervals.

The skewness of the sampling distributions of doubly-robust point estimators due to overlap violations complicates characterizing their behavior in simulation studies. \cite{robins2007comment} made this point in their response to \cite{kang2007demystifying}, arguing that ``one thousand replications are not enough to capture the tail behavior of highly skewed sampling distributions, and as such cannot produce reliable Monte Carlo estimates of bias, much less of variance.'' One approach to reduce the Monte Carlo error would be to simply increase the number of replications; we do so by running $5000$ replications. In addition, we address the issue of highly skewed sampling distributions in two other ways. First, we conditionally simulate datasets that exhibit non-overlap. Specifically, we draw datasets from \eqref{eq:simulation-dgp} via rejection sampling such that $\max_i |r_i| > 100$. That is, datasets of size $n$ are drawn repeatedly until a dataset is found that has at least one observation with $A_i = 1$ and $\pi(X_i) < 0.01$ or $A_i = 0$ and $\pi(X_i) > 0.99$ (which each imply $|r_i| > 100$). This targets the investigation to the datasets with extreme non-overlap that cause particular problems for traditional estimators. This helps us to characterize the ``worst-case" performance of traditional estimators. For context, we also perform a simulation study that is identical except for using unconditional draws from \eqref{eq:simulation-dgp}, without rejection sampling. This allows us to characterize the ``average'' performance of the estimators with respect to our DGP. The second way we address skewed sampling distributions under non-overlap is by evaluating our methods with additional metrics that capture the full sampling distribution. For the 95\% uncertainty interval widths of the one-step and non-overlap bounds, we calculate their means, 90\% quantiles, and standard deviation across simulations to capture both the central tendency and spread of the sampling distributions. We also report traditional metrics including the empirical coverage of the 95\% confidence intervals and the power of the hypothesis tests, defined as the proportion of simulations in which each method rejected the null hypothesis of a null treatment effect.

We generated $5000$ simulation datasets for every sample size $n \in \{ 100, 250, 500, 1000 \}$ via rejection sampling (such that simulation datasets have $\max_i |r_i| > 100$) and non-rejection sampling methods described above. In all cases, nuisance parameters for both the benchmark one-step estimator and the non-overlap bounds were estimated using well-specified generalized linear models and $5$-fold sample splitting. The non-overlap bounds were estimated for each of the smoothness parameters $\gamma \in \{ 10^{-3}, 10^{-2}, 10^{-1} \}$.

We emphasize that we did not compare the non-overlap bounds to methods targeting alternative estimands, such as overlap weighting \citep{li2018balancing} or propensity score trimming \citep{crump2009overlap, mcclean2025propensity} precisely because our inferential goal is to estimate the \textit{population} ATE, and not an alternative estimand. Because propensity score trimming, for example, does not target the ATE, it is not a relevant comparator to non-overlap bounds. For comparisons of the finite-sample performance of various weighting and trimming estimators, we refer to the substantial literature on the subject \citep{sturmer2010propensity, lee2011trimming, busso2014weighting, li2018propensity, zhou2020overlap, sturmer2021trimming, benmichael2023balancing}. 

\paragraph{Results} The results of the simulation study are shown in Table~\ref{tab:simulation-study-1-results}, focusing on uncertainty interval widths, and Table~\ref{tab:simulation-study-1-results-power}, containing empirical coverage and power results. Each table is separated into separate sections for (A) the simulations generated conditional on $\max_i|r_i| > 100$ and (B) simulations generated unconditionally. The estimators exhibit significantly different performance in the conditionally simulated simulations, which we emphasize represents the ``worst-case'' behavior of traditional doubly-robust estimators by targeting the skewed region of their sampling distributions. As expected, in this regime the 95\% confidence interval widths for the doubly-robust one-step estimator are wide on average, and are highly variable. Coverage is conservative, reaching near 100\% coverage. Power to reject the null hypothesis is low; at the smallest sample size $n = 100$, the null hypothesis is rejected in only $3\%$ of simulations. In comparison, the non-overlap bounds are narrower and are less variable at all sample sizes and choice of tuning parameter $\gamma$. The non-overlap bounds also exhibit higher power; for example achieving approximately 50\% power at the smallest sample size of $n=100$. 

In the unconditional regime, representing the ``average-case'' performance of the estimators under the simulation DGP, the non-overlap bounds have smaller mean width and are less variable than the doubly-robust uncertainty intervals. On the other hand, the doubly-robust intervals exhibit better empirical coverage and power in this scenario. The non-overlap bounds nevertheless achieve higher power for some sample sizes. That the doubly-robust one-step estimator performs better in this regime compared to the previous conditional regime can be explained by the skewness of its sampling distribution. By marginalizing over the entire sampling distribution, poor behavior when extreme finite-sample overlap violations occur (i.e. $\max_i |r_i|$ is large) is de-emphasized. The skewness also explains the relationships that the mean and median confidence intervals widths have with increasing sample size: the mean width \textit{increases} with sample size, while the median \textit{decreases}, because larger simulated datasets have higher probability of including an observation with a large $r_i$ value. 

The non-overlap bounds in both simulation regimes were slightly sensitive to the choice of tuning parameter $\gamma$. As expected by the construction of the bounds, the mean interval width slightly increased and empirical coverage became more conservative as the smoothness parameter increased.

\begin{table}
    \centering
    \begin{tabular}{|llrrrrrrrr|}
    \hline
    & & \multicolumn{8}{c|}{95\% UI Width} \\
    & & \multicolumn{2}{c}{Mean} & \multicolumn{2}{c}{Median} & \multicolumn{2}{c}{90\% Quantile} & \multicolumn{2}{c|}{Std. Dev.} \\
    $N$ & $\gamma$ & Bounds & DR & Bounds & DR & Bounds & DR & Bounds & DR \\
    \hline
    \multicolumn{10}{|l|}{\textit{(A) Simulations conditional on $\max_i |r_i| > 100$}} \\
    100 & $10^{-3}$ & \textbf{0.46} & 1.90 & \textbf{0.45} & 2.00 & \textbf{0.52} & 2.00 & \textbf{0.05} & 0.37\\
     & $10^{-2}$ & \textbf{0.46} & 1.90 & \textbf{0.45} & 2.00 & \textbf{0.52} & 2.00 & \textbf{0.05} & 0.37\\
     & $10^{-1}$ & \textbf{0.49} & 1.90 & \textbf{0.49} & 2.00 & \textbf{0.57} & 2.00 & \textbf{0.06} & 0.37\\
    250 & $10^{-3}$ & \textbf{0.29} & 1.83 & \textbf{0.29} & 2.00 & \textbf{0.32} & 2.00 & \textbf{0.03} & 0.50\\
     & $10^{-2}$ & \textbf{0.29} & 1.83 & \textbf{0.29} & 2.00 & \textbf{0.33} & 2.00 & \textbf{0.04} & 0.50\\
     & $10^{-1}$ & \textbf{0.37} & 1.83 & \textbf{0.37} & 2.00 & \textbf{0.40} & 2.00 & \textbf{0.03} & 0.50\\
    500 & $10^{-3}$ & \textbf{0.22} & 1.65 & \textbf{0.21} & 2.00 & \textbf{0.28} & 2.00 & \textbf{0.04} & 0.68\\
     & $10^{-2}$ & \textbf{0.23} & 1.65 & \textbf{0.22} & 2.00 & \textbf{0.27} & 2.00 & \textbf{0.04} & 0.68\\
     & $10^{-1}$ & \textbf{0.29} & 1.65 & \textbf{0.29} & 2.00 & \textbf{0.31} & 2.00 & \textbf{0.02} & 0.68\\
    1000 & $10^{-3}$ & \textbf{0.17} & 1.35 & \textbf{0.15} & 2.00 & \textbf{0.23} & 2.00 & \textbf{0.05} & 0.84\\
     & $10^{-2}$ & \textbf{0.20} & 1.35 & \textbf{0.20} & 2.00 & \textbf{0.23} & 2.00 & \textbf{0.03} & 0.84\\
     & $10^{-1}$ & \textbf{0.23} & 1.35 & \textbf{0.23} & 2.00 & \textbf{0.25} & 2.00 & \textbf{0.01} & 0.84\\
    \multicolumn{10}{|l|}{\textit{(B) Unconditional simulations}} \\
    100 & $10^{-3}$ & \textbf{0.51} & 0.52 & 0.51 & \textbf{0.41} & 0.59 & \textbf{0.51} & \textbf{0.06} & 0.40\\
     & $10^{-2}$ & \textbf{0.51} & 0.52 & 0.51 & \textbf{0.41} & 0.58 & \textbf{0.51} & \textbf{0.06} & 0.40\\
     & $10^{-1}$ & \textbf{0.51} & 0.52 & 0.51 & \textbf{0.41} & 0.58 & \textbf{0.51} & \textbf{0.05} & 0.40\\
    250 & $10^{-3}$ & \textbf{0.34} & 0.54 & 0.34 & \textbf{0.25} & \textbf{0.37} & 2.00 & \textbf{0.03} & 0.66\\
     & $10^{-2}$ & \textbf{0.34} & 0.54 & 0.34 & \textbf{0.25} & \textbf{0.37} & 2.00 & \textbf{0.03} & 0.66\\
     & $10^{-1}$ & \textbf{0.35} & 0.54 & 0.35 & \textbf{0.25} & \textbf{0.38} & 2.00 & \textbf{0.03} & 0.66\\
    500 & $10^{-3}$ & \textbf{0.25} & 0.67 & 0.26 & \textbf{0.17} & \textbf{0.29} & 2.00 & \textbf{0.03} & 0.81\\
     & $10^{-2}$ & \textbf{0.26} & 0.67 & 0.26 & \textbf{0.17} & \textbf{0.29} & 2.00 & \textbf{0.03} & 0.81\\
     & $10^{-1}$ & \textbf{0.28} & 0.67 & 0.28 & \textbf{0.17} & \textbf{0.30} & 2.00 & \textbf{0.02} & 0.81\\
    1000 & $10^{-3}$ & \textbf{0.19} & 0.90 & 0.20 & \textbf{0.19} & \textbf{0.23} & 2.00 & \textbf{0.04} & 0.91\\
     & $10^{-2}$ & \textbf{0.21} & 0.90 & 0.21 & \textbf{0.19} & \textbf{0.23} & 2.00 & \textbf{0.02} & 0.91\\
     & $10^{-1}$ & \textbf{0.23} & 0.90 & 0.23 & \textbf{0.19} & \textbf{0.24} & 2.00 & \textbf{0.01} & 0.91\\
    \hline
    \end{tabular}
    \caption{Uncertainty width results from Simulation Study 1 (Section~\ref{sec:simulation-study-1}) for (A) simulations conditional on $\max_i |r_i| > 100$, with $r_i$ the inverse propensity weight defined in \eqref{eq:propensity-weight}, and (B) simulations generated unconditionally. Results compare the 95\% uncertainty intervals of a doubly-robust one-step estimator (DR) and the tightest uniform 95\% non-overlap bounds (Bounds) across a grid of propensity score thresholds and smoothness parameters in terms of the mean, median, 90\%, quantile, and standard deviation of the  uncertainty interval width.}
    \label{tab:simulation-study-1-results} 
\end{table}

\begin{table}
    \centering
    \begin{tabular}{|llrrrr|}
    \hline
    & & \multicolumn{2}{c}{95\% Coverage} & \multicolumn{2}{c|}{Power} \\
    $N$ & $\gamma$ & Bounds & DR & Bounds & DR \\
    \hline
    \multicolumn{6}{|l|}{\textit{(A) Simulations conditional on $\max_i |r_i| > 100$}} \\
    100 & $10^{-3}$ & \textbf{96.7\%} & 99.8\% & \textbf{0.49} & 0.03\\
     & $10^{-2}$ & \textbf{96.7\%} & 99.8\% & \textbf{0.49} & 0.03\\
     & $10^{-1}$ & \textbf{98.2\%} & 99.8\% & \textbf{0.43} & 0.03\\
    250 & $10^{-3}$ & \textbf{96.8\%} & 99.4\% & \textbf{0.89} & 0.07\\
     & $10^{-2}$ & \textbf{97.0\%} & 99.4\% & \textbf{0.89} & 0.07\\
     & $10^{-1}$ & 99.6\% & \textbf{99.4\%} & \textbf{0.71} & 0.07\\
    500 & $10^{-3}$ & \textbf{97.4\%} & 99.3\% & \textbf{0.99} & 0.15\\
     & $10^{-2}$ & \textbf{98.2\%} & 99.3\% & \textbf{0.99} & 0.15\\
     & $10^{-1}$ & 99.9\% & \textbf{99.3\%} & \textbf{0.95} & 0.15\\
    1000 & $10^{-3}$ & \textbf{97.6\%} & 98.8\% & \textbf{1.00} & 0.31\\
     & $10^{-2}$ & 99.4\% & \textbf{98.8\%} & \textbf{1.00} & 0.31\\
     & $10^{-1}$ & 100.0\% & \textbf{98.8\%} & \textbf{1.00} & 0.31\\
    \multicolumn{6}{|l|}{\textit{(B) Unconditional simulations}} \\
    100 & $10^{-3}$ & 98.6\% & \textbf{94.1\%} & 0.34 & \textbf{0.51}\\
     & $10^{-2}$ & 98.6\% & \textbf{94.1\%} & 0.34 & \textbf{0.51}\\
     & $10^{-1}$ & 98.7\% & \textbf{94.1\%} & 0.35 & \textbf{0.51}\\
    250 & $10^{-3}$ & 98.8\% & \textbf{94.1\%} & 0.77 & \textbf{0.78}\\
     & $10^{-2}$ & 98.9\% & \textbf{94.1\%} & 0.78 & 0.78\\
     & $10^{-1}$ & 99.4\% & \textbf{94.1\%} & 0.75 & \textbf{0.78}\\
    500 & $10^{-3}$ & 99.1\% & \textbf{95.1\%} & \textbf{0.98} & 0.72\\
     & $10^{-2}$ & 99.4\% & \textbf{95.1\%} & \textbf{0.98} & 0.72\\
     & $10^{-1}$ & 99.8\% & \textbf{95.1\%} & \textbf{0.96} & 0.72\\
    1000 & $10^{-3}$ & 98.4\% & \textbf{95.9\%} & \textbf{1.00} & 0.57\\
     & $10^{-2}$ & 99.8\% & \textbf{95.9\%} & \textbf{1.00} & 0.57\\
     & $10^{-1}$ & 99.9\% & \textbf{95.9\%} & \textbf{1.00} & 0.57\\
    \hline
    \end{tabular}
    \caption{Empirical coverage and power results from Simulation Study 1 (Section~\ref{sec:simulation-study-1}) for (A) simulations conditional on $\max_i |r_i| > 100$, with $r_i$ the inverse propensity weight defined in \eqref{eq:propensity-weight}, and (B) simulations generated unconditionally. Results compare the 95\% uncertainty intervals of a doubly-robust one-step estimator (DR) and the tightest uniform 95\% non-overlap bounds (Bounds) across a grid of propensity score thresholds and smoothness parameters in terms of empirical 95\% coverage and power.}
    \label{tab:simulation-study-1-results-power} 
\end{table}

\subsection{Simulation Study 2}
Simulation Study 2 is designed to investigate the theoretical results derived in Section~\ref{sec:comparison-wald} comparing non-overlap bounds and traditional EIF-based point estimator confidence intervals in overlap regimes characterized by the slowly-varying tails assumption (Assumption~\ref{assumption:slowly-varying-tails}). To ensure that the propensity score distribution satisfied the slowly-varying tails assumption, we drew $A$ from a symmetric Beta distribution with shape parameters $a = \gamma_0 - 1$ and $b = \gamma_0 - 1$. The full simulation data-generating process was characterized by
\begin{align}
    U &\sim N(0, 1), \\
    X_1 &\sim N(U, 0.25^2), \\
    X_2 &\sim N(U^2, 0.25^2), \\
    X_3 &\sim N(0, 1), \\
    A &\sim \mathrm{Bernoulli}\left( F_{\beta}^{-1}(\Phi(U), \gamma_0 - 1, \gamma_0 - 1) \right), \\
    Y &\sim \mathrm{Bernoulli}\left( \mathrm{logit}^{-1}(-0.5 + 0.5U + A) \right), 
\end{align}
where $U$ is a latent confounder, $X_1, X_2, X_3$ are observed covariates, $F_{\beta}^{-1}(x, a, b)$ denotes the inverse Beta cumulative density function evaluated at $x$ with shape parameters $a$ and $b$, and $\Phi$ indicates the standard normal cumulative density function. Because $\Phi(U) \sim \mathrm{Uniform}(0,1)$ when $U \sim N(0, 1)$, applying the inverse Beta CDF yields propensity scores that are marginally distributed as $\mathrm{Beta}(\gamma_0 - 1, \gamma_0 - 1)$ while depending smoothly on the latent confounder $U$, ensuring both the desired tail behavior for the propensity score distribution and the existence of causal confounding. The third covariate, $X_3$, is a noise covariate independent of all other variables. A key property of this DGP is that the true ATE does not depend on $\gamma_0$, because $\gamma_0$ only controls the propensity score distribution and not the potential outcomes. Therefore, differences in estimator performance across $\gamma_0$ are attributable solely to the overlap regime. 

We generated $5000$ simulation datasets for every combination of sample sizes $n \in \{ 100, 250, 500, 1000, 2500, 5000 \}$ and $\gamma_0 \in \{ 1.5, 1.75, 2, 2.25, 2.5, 3 \}$. Because the observed covariates are nonlinear functions of the latent confounder $U$, we chose to estimate nuisance parameters using random forests (via the \texttt{SL.ranger} SuperLearner library) to evaluate the non-overlap bounds under realistic nonparametric conditions. Based on the results of Simulation Study 1, which showed performance was not highly sensitive to the smoothness parameter $\gamma$, we chose to set $\gamma = 10^{-2}$.

\paragraph{Results}
Figure~\ref{fig:simulation-study-2-results} illustrates the mean uncertainty interval widths for each sample size $n$ and smoothly-varying tails parameter $\gamma_0$. In the very weak overlap regime ($\gamma_0 < 2$), the mean point estimator confidence interval width fails to converge as sample size increases, consistent with the fact that the efficiency bound in this regime is unbounded (Theorem~\ref{thm:very-weak-overlap}). With larger $n$, the probability of drawing an observation with an extreme propensity score ($P(A = a | X)$) increases, inflating the width of the estimated CI. Non-overlap bounds, on the other hand, are able to adaptively place these observations into a non-overlap subpopulation. In the somewhat weak overlap regime, the non-overlap CI width may outperform traditional CIs in finite samples, particularly for $\gamma_0$ near the phase transition at $\gamma_0 = 2$ (Theorem~\ref{thm:somewhat-weak-overlap}). We see this verified in the simulation results: as $\gamma_0$ increases away from 2, the gap between the traditional CIs and non-overlap CIs decreases. Appendix~\ref{section:additional-results} includes corresponding results for 95\% empirical coverage and power to reject the null hypothesis of a null treatment effect, showing that non-overlap CIs are conservative in terms of coverage, yet the associated test exhibits higher power than a traditional test based on the doubly-robust estimator. 

\begin{figure}
    \centering
    \includegraphics[width=1\linewidth]{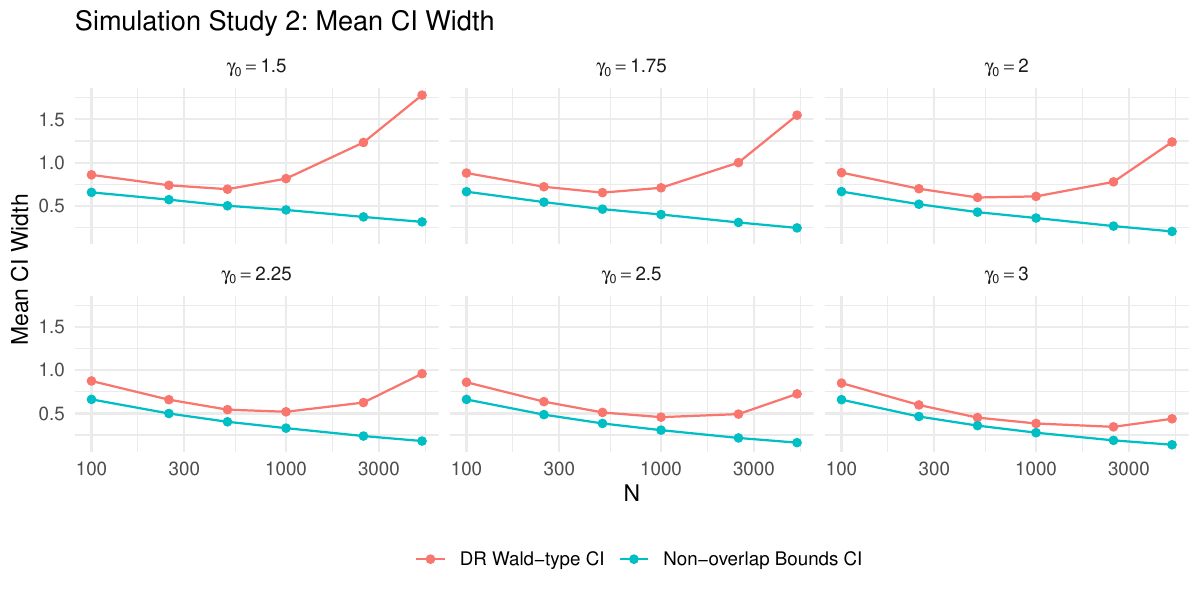}
    \caption{Simulation Study 2. Comparison of uncertainty intervals from a 95\% point estimator confidence interval (CI) centered on a doubly-robust one-step point estimate (red) and the union 95\% confidence interval formed via estimated non-overlap bounds (blue). The $x$-axis is the sample size and the $y$-axis the mean CI width. Each panel shows a different setting of $\gamma_0$, the smoothly-varying tails parameter (Assumption~\ref{assumption:slowly-varying-tails}) controlling the overlap regime, with $\gamma_0 = 2$ separating very weak overlap ($\gamma_0 < 2$) from somewhat weak overlap ($\gamma_0 > 2$). }
    \label{fig:simulation-study-2-results}
\end{figure}

\section{Application}
\label{section:application}
To illustrate the utility of our methods we reanalyzed data from six observational datasets spanning a range of overlap conditions. A systematic comparison across all six datasets (Appendix~\ref{appendix:additional-applications}) demonstrates that non-overlap bounds exhibit a ``no-regret'' property in practice: in datasets where overlap is well-satisfied, the non-overlap bounds are similar in width to traditional point estimator confidence intervals, consistent with Theorem~\ref{thm:equiv-tmle}, while in datasets with severe overlap violations, the non-overlap bounds remain informative even when the point estimator intervals span the entire parameter space. Here, we present a detailed analysis of one dataset that exemplifies the latter regime: an observational study of mortality following treatment with right heart catheterization (RHC) among patients in five medical centers in the USA \citep{murphy1990support,connors1996rhc}. Formally, the treatment is defined as whether RHC was applied within 24 hours of hospital admission (with $A = 1$ indicating treatment and $A = 0$ non-treatment). The outcome is a binary indicator of survival at 30 days post admission. The covariates comprise $72$ variables. The dataset includes $n = 5735$ patients, with $2184$ in the treatment group and $3551$ in the non-treatment group. Due to significant differences in the covariate distribution between groups, these data have been reanalyzed numerous times to illustrate causal inference methods designed to address practical overlap violations \citep{hirano2001propensityrhc, crump2009overlap, traskin2011trees, rosenbaum2012optimal, li2018balancing, rothe2017robustoverlapci, lee2021bounding, ma2024testingoverlap}. For our analysis, we use the version of the dataset publicly available in the \texttt{R} package \texttt{ATbounds} \citep{lee2021atboundscran}.

We calculated the non-overlap bounds on a logarithmic grid from $c = 10^{-6}$ to $c = 10^{-1}$. The smoothness tuning parameter was set to each of the values $\gamma \in \{ 0.001, 0.01 \}$. A uniform 95\% confidence set over the grid of thresholds and both tuning parameter values was calculated using the multiplier bootstrap method with $1000$ bootstrap draws. Nuisance parameters were estimated using cross-fitted generalized linear models with $5$ folds. Replication material in the form of an \texttt{R} script is available in Appendix~\ref{section:application-code}.

The estimated uniform non-overlap 95\% bounds for $\gamma = 0.01$ are shown in Figure~\ref{fig:rhc-example} (a comparable plot for $\gamma = 0.001$ is included in the appendix). The uniform non-overlap bounds exclude zero for a subset of the propensity score thresholds, implying statistically significant evidence of a non-zero ATE at the 5\% level. The tightest uniform 95\% uncertainty interval for the ATE was $(-0.10, -0.03)$. In comparison, a doubly-robust one-step point estimate of the ATE has 95\% confidence intervals that are non-informative (covering the entire parameter space $[-1, 1]$) due to the presence of observations with large inverse propensity weights. Figure~\ref{fig:rhc-example-pscore} displays the estimated propensity score distribution and the size of the non-overlap subpopulation for a range of propensity score thresholds. The structure of the propensity score distribution explains the success of the non-overlap bounds in this example: while the propensity score distribution includes values near $0$ and $1$, indicating overlap violations, the extreme cases are concentrated in a very small non-overlap subpopulation.

The estimated non-overlap bounds for the ATE are consistent with previous point estimates of alternative treatment effects defined via weighting or exclusion to target subpopulations satisfying overlap. The estimates reported in \citep{crump2009overlap} for the treatment effect among the population with propensity score estimates in $[0.1, 0.9]$ fall within the non-overlap bounds, as do the estimated overlap treatment effect reported in \cite{li2018propensity}. As opposed to their results, however, our results allow for inference on the ATE. 

\begin{figure}
    \centering
    \includegraphics[width=0.8\linewidth]{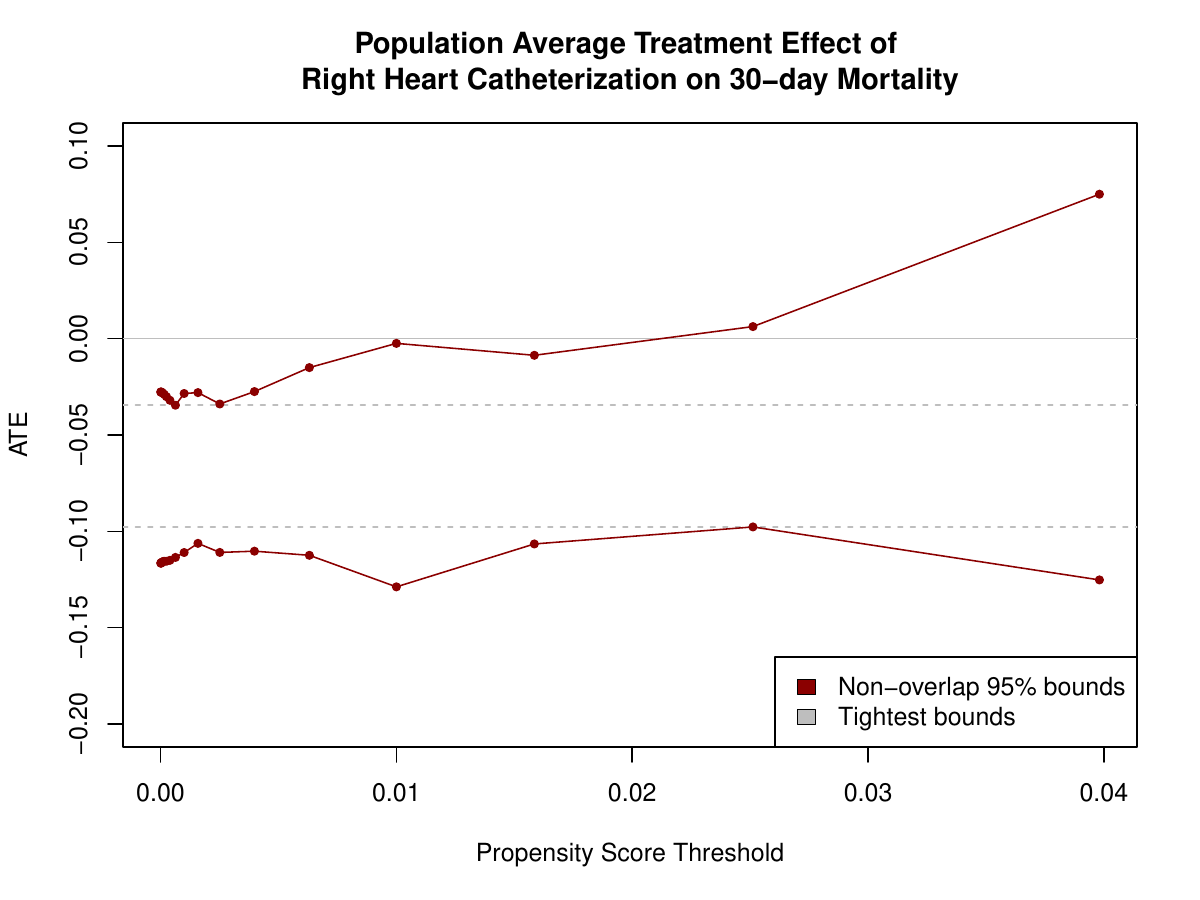}
    \caption{Uniform 95\% non-overlap bounds (for $\gamma = 0.01$) on the average treatment effect (ATE) of right heart catheterization on survival. The points illustrate the lower and upper bounds with respect to a logarithmic grid of propensity score thresholds. The lines between points are solely to guide the eye. The horizontal dotted lines indicate the tightest valid 95\% uncertainty interval that may be formed from the non-overlap bounds. }
    \label{fig:rhc-example}
\end{figure}

\begin{figure}
    \centering
    \includegraphics[width=0.9\linewidth]{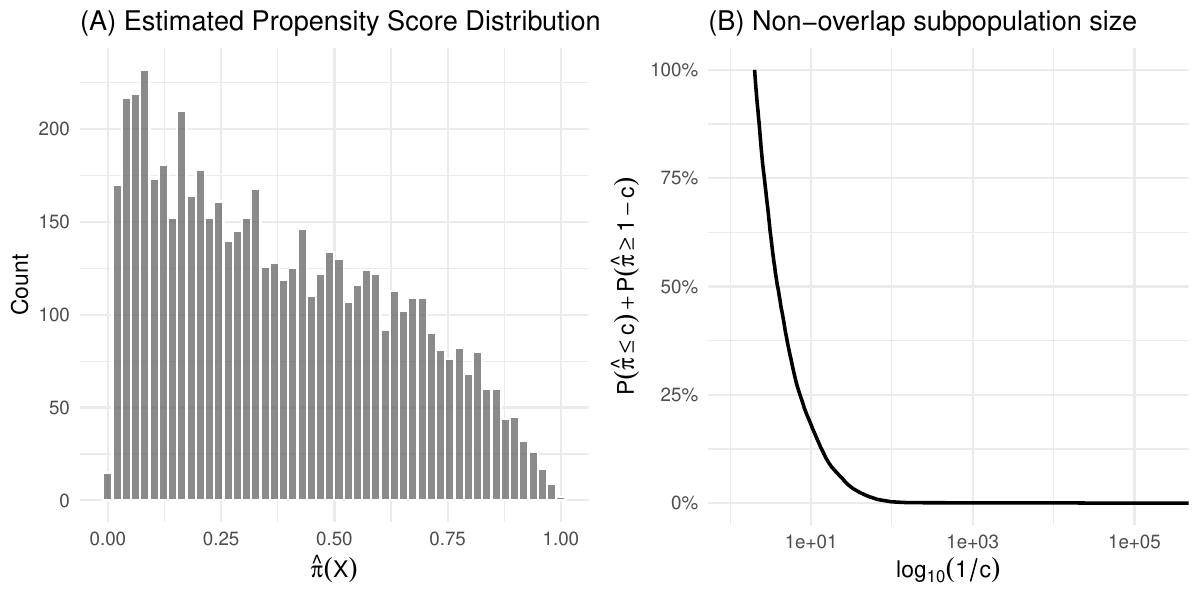}
    \caption{Propensity score distribution and non-overlap subpopulation size from the Right Heart Catheterization case study. (A) Estimated distribution of propensity scores. (B) Estimated size of non-overlap subpopulation ($\bbP(\hat{\pscore} \leq c) + \bbP(\hat{\pscore} \geq c)$) versus propensity score threshold $c$ (shown as $\log_{10}(1/c)$).}
    \label{fig:rhc-example-pscore}
\end{figure}

\section{Discussion}
\label{sec:discussion}
In practice, our results have implications for Average Treatment Effect (ATE) estimation when there are structural or practical violations of the overlap assumption. Under structural overlap violations, our method facilitates valid statistical inference for the ATE even when it is not point identified. This is in contrast to alternative approaches, which achieve point identification or precise estimation by changing the target of interest and thus altering the scientific question of interest \citep{petersen2014roadmap, dang2023roadmap, rizk2025overlap}. Therefore, if the estimand of interest is the ATE and not a weighted or trimmed alternative, then our methods provide a novel solution to conduct inference on the target of interest. From this view, alternative weighted or trimmed estimands can be part of a \textit{secondary} subgroup analysis, and reported in addition to non-overlap bounds, which provide important context about what is known about the ATE. 

The performance of the non-overlap bounds in finite samples with practical non-overlap is surprising in that the bounds have provably better properties than classical point esteimator confidence intervals in some overlap regimes. Our theoretical analysis in Section~\ref{sec:comparison-wald} formalizes this comparison. Under very weak overlap ($\gamma_0 < 2$), the ATE efficiency bound is infinite, precluding $\sqrt{n}$-estimation by any regular estimator, yet the non-overlap bounds can be estimated at $\sqrt{n}$-rates with valid confidence intervals (Theorem~\ref{thm:very-weak-overlap}). In the regime of somewhat weak overlap ($\gamma_0 > 2$), Theorem~\ref{thm:somewhat-weak-overlap} shows that the point estimator confidence interval converges at a faster asymptotic rate, yet the non-overlap bounds can yield shorter confidence intervals in finite samples near the phase transition. 

Lest we be charged with offering a free lunch, we emphasize that the use of non-overlap bounds involves tradeoffs. First, in exchange for overlap we must control the expected values of the potential outcomes in the non-overlap region. We assumed outcome boundedness as a straightforward way to achieve this, although there are alternatives such as assuming $\E ( Y^1 \mid \pi \leq c )$ and $\E( Y^0 \mid \pi \geq 1 - c)$ are bounded (for example, a sufficient condition for this is if $Y^1$ is non-negative and there exists $b > 0$ such that $\bbP(Y^1 > t \mid \pi \leq c) < 1/t^2$ for all $t > b$, implying $\E(Y^1 \mid \pi \leq c) \leq b + 1/b$). The second trade-off is that non-overlap bounds sacrifice point identification in exchange for partial identification bounds. Finally, non-overlap bounds are informative only in specific scenarios. The width of the non-overlap bounds is a function of the size of the non-overlap subpopulation, and the bounds will exclude zero only if the size of that population (representing ``noise'') is smaller than the treatment effect in the overlap subpopulation (representing ``signal''). The core insight of our work is that, in our experience as applied researchers, in many real-world situations the non-overlap population is small, meaning that the non-overlap bounds may still yield sufficiently narrow bounds to be useful (see Appendix~\ref{appendix:additional-applications} for an analysis of five additional datasets). 

Practical use of the non-overlap bounds requires setting two additional tuning parameters beyond those required for traditional doubly-robust estimators, in the form of the propensity score threshold $c$ and smooth approximation parameter $\gamma$. Proposition~\ref{prop:confidence-interval-width-bound} shows that the width of the non-overlap confidence interval decomposes into an identification gap (controlled by $c$), smooth approximation error (controlled by $\gamma$), and estimation uncertainty (scaling as $c^{-1/2}$ and $\gamma^{-1}$). In the simulation study, the performance of non-overlap bounds was not highly sensitive to $\gamma$, suggesting that in practice a small value can be chosen, such as $\gamma = 10^{-2}$ or $10^{-3}$. Following \cite{branson2024causaleffectestimationpropensity}, who apply a similar type of smooth approximation in a different context, we advise testing multiple values of $\gamma$ in a sensitivity analysis. The strength of the uniform non-overlap bounds formed via the multiplier bootstrap is that they are valid over all specified combinations of threshold and smoothness parameter, allowing for valid inference even when multiple smoothness parameters are used. In addition, an interesting avenue for future work is to consider uniform inference methods valid for both non-overlap bounds and confidence intervals for point estimates of the ATE. Indeed, when the threshold $c = 0$ and the smoothness is set to $\gamma = 0$, then the non-overlap bounds are equivalent (in practice) to a classical point estimator confidence interval around a TMLE point estimate of the ATE. A combined inference approach would allow for fully pre-specified analysis plans that are both robust to overlap violations while yielding point estimates when overlap is satisfied. 

The non-overlap bounds we propose in Proposition~\ref{prop:bounds} divide the population into subpopulations by thresholding the propensity score, which has an immediate and intuitive connection to the definition of overlap in Assumption~\ref{assumption:weak-positivity}. However, other decompositions of the population may have advantageous properties. First, an immediate extension would be to threshold the inverse propensity score by fixing $c \geq 2$ and considering subpopulations for which $\pscore^{-1} < c$ (or $(1 - \pscore)^{-1} < c$). For identification, this is equivalent to thresholding the propensity score, but for estimation allows for empirical Riesz loss estimation of the inverse propensity score \citep{chernozhukov2021automatic}. Another option thresholds the weight $\bbP(A)\bbP(X) \slash \bbP(A, X)$, which can be interpreted in terms of the mutual information between $A$ and $X$ (defined as $-\log \E\lcb \bbP(A)\bbP(X) \slash \bbP(A, X) \rcb$). Careful analysis of the efficiency bound for the ATE shows that large inverse propensity scores contribute to the efficiency bound proportional to the conditional variance of the outcome, suggesting thresholds related to $\Var(Y \mid A, X) \slash (\pscore(X)(1 - \pscore(X)))$ may yield estimators with better finite-sample properties. Finally, the problem could be framed in terms of finding a maximal subpopulation that achieves covariate balance subject to an entropy or variance constraint on the balancing weights, with worst-case bounds applied to the complement subpopulation. Investigating these possibilities and studying their implications for identification and estimation is a rich area for future work.

In addition, while we focused on deriving non-overlap bounds for the ATE, as a canonical example, we expect that our approach is relevant to a larger class of causal estimands. The two key ingredients necessary are a bounded outcome and overlap violations within a small proportion of the population relative to the magnitude of the treatment effect in the rest of the population. We suspect that continuous treatment effects \citep{schomaker2024continuous}, longitudinal treatment effects \citep{mcclean2025propensity}, and transport treatment effects (\citealt{zivich2024simulation, zivich2024synthesis}) may be particularly amenable to this style of data-adaptive non-overlap bounds.

\subsection*{Acknowledgements}
The computational requirements for this work were supported in part by the NYU Langone High Performance Computing (HPC) Core's resources and personnel. During the preparation of this work the authors used generative AI tools to suggest edits for structure, grammar, and clarity and to suggest edits to code for the data applications and simulation studies. After using these tools/services the authors reviewed and edited the content as necessary and take full responsibility for the content of the publication.

\subsection*{Data Availability}
Replication materials for the simulation study and data application are available at \url{https://github.com/herbps10/robust_ate_bounds_paper}. An \texttt{R} package implementing the proposed estimators is available at \url{https://github.com/herbps10/effectbounds}. 

\bibliography{references}

\clearpage

\begin{appendix}
\addcontentsline{toc}{section}{Appendix} % Add the appendix text to the document TOC

\part{Appendix}
\setcounter{table}{0} 
\setcounter{figure}{0} 
\renewcommand{\thetable}{A\arabic{table}}
\renewcommand{\thefigure}{A\arabic{figure}}

\parttoc

\section{Proofs}
%\subsection{Proof of Proposition~\ref{prop:smooth-bounds}}
%The result follows from the smoothed approximation error bounds \eqref{eq:approximation-error-bounds}:
%\begin{align}
%    \smoothU(c) &= \ate(c) + \bbP(\pscore \leq c) \\
%                   &= \smoothATE(c) + \bbP(\pscore \leq c) + \ate(c) - \smoothATE(c) \\
%                   &\leq \smoothATE(c) + \bbP(\pscore \leq c) + \bbP(\pscore > c) - \E \{ \smooth_g(\pscore, c, \gamma) \}  \\
%                   &= \smoothATE(c) + \Big[  1 - \E[\smooth_g(\pscore, c, \gamma)] \Big] \equiv \smoothU(c) \text{ and } \\
%    L(c) &= \ate(c) - \bbP(\pscore \geq 1 - c) \\
%                   &= \smoothATE(c) - \bbP(\pscore \geq 1 - c) + \psi(c) - \smoothATE(c))  \\
%                   &\geq \smoothATE(c) - \bbP(\pscore \geq 1 - c) + \E \{ \smooth_{\ell}(\pscore, 1 - c, \gamma) \} - \bbP(\pscore < 1-c)  \\
%                   &= \smoothATE(c) - \Big[ 1 - \E \{ \smooth_{\ell}(\pscore, 1 - c, \gamma) \}  \Big] \equiv \smoothL(c).
%\end{align}

\subsection{Proof of Proposition~\ref{prop:smooth-bounds}}
\begin{proof}
    By Proposition~\ref{prop:bounds}, $\E(Y^1 - Y^0) \in \left[L(c), U(c) \right]$. Therefore, it is only necessary to show $L(c) \geq L_s(c, \gamma)$ and $U(c) \leq U_s(c, \gamma)$.  Notice first that the smooth approximation error $\psi(c) - \psi_s(c, \gamma)$ satisfies
    \begin{align}
        \psi(c) - \smoothATE(c, \gamma) = \E \left[ \mu_1(X) \{ \I(\pscore > c) - \smooth_g(\pscore, c, \gamma) \} - \mu_0(X) \{ \I (\pscore < 1-c) - \smooth_{\ell}(\pscore, 1 - c, \gamma) \} \right].
    \end{align}
    The boundedness of the outcome and Property~\ref{req:smooth-approach} imply
    \begin{equation} \label{eq:smooth-approx-error}
        \E\lcb \smooth_{\ell}(\pscore, 1 - c, \gamma) \rcb - \bbP\left( \pscore < 1-c \right) \leq \psi(c) - \smoothATE(c, \gamma) \leq \bbP\left( \pscore > c\right) - \E \lcb \smooth_g(\pscore, c, \gamma) \rcb.
    \end{equation}
    Applying these inequalities yields the result. For the lower bound,
    \begin{align}
    L(c) &= \ate(c) - \bbP(\pscore \geq 1 - c) \\
                   &= \smoothATE(c, \gamma) - \bbP(\pscore \geq 1 - c) + \psi(c) - \smoothATE(c, \gamma)  \\
                   &\geq \smoothATE(c, \gamma) - \bbP(\pscore \geq 1 - c) + \E \{ \smooth_{\ell}(\pscore, 1 - c, \gamma) \} - \bbP(\pscore < 1-c)  \\
                   &= \smoothATE(c, \gamma) - \Big[ 1 - \E \{ \smooth_{\ell}(\pscore, 1 - c, \gamma) \}  \Big] \equiv \smoothL(c, \gamma).
    \end{align}
    For the upper bound,
    \begin{align}
        U(c) &= \ate(c) + \bbP(\pscore \leq c) \\
        &= \smoothATE(c, \gamma) + \bbP(\pscore \leq c) + \ate(c) - \smoothATE(c, \gamma) \\
        &\leq \smoothATE(c, \gamma) + \bbP(\pscore \leq c) + \bbP(\pscore > c) - \E \{ \smooth_g(\pscore, c, \gamma) \} \\
        &= \smoothATE(c, \gamma) + \Big[  1 - \E \{ \smooth_g(\pscore, c, \gamma) \} \Big] \equiv \smoothU(c, \gamma).
    \end{align}
\end{proof}

\subsection{Derivations of efficient influence functions}
The following lemmas derive the efficient influence functions of the parameters $\ate(c)$, $\smoothL$, and $\smoothU$. The proofs follow a similar structure. First, we use heuristic techniques to derive putative EIFs; in particular, we draw on several ``tricks" from \cite{kennedy2024doublerobustreview}, which we refer to for further explanation; in our proofs, we annotate where each trick is applied. We then show that the parameter satisfies the von-Mises expansion \eqref{eq:von-mises} by deriving the form of the second-order remainder term. By showing that the remainder term is second-order, and that the putative EIF is bounded, an appeal to \citep[Lemma 2]{kennedy2023density} establishes that the putative EIF is the EIF. 

\begin{lemma}[Efficient influence function of $\smoothATE(c, \gamma)$]
    \label{lemma:eif-ate-c}
    Fix $c \in \left[0, \tfrac{1}{2}\right]$. The parameter $\smoothATE(c, \gamma)$ is pathwise differentiable with uncentered efficient influence function given by
    \begin{align}
        \eif_{\smoothATE}(Z) =& \frac{A}{\pscore} \smooth_g(\pscore, c, \gamma) \lcb Y - \omodel_1 \rcb - \frac{1- A}{1 - \pscore} \smooth_{\ell}(\pscore, 1 - c, \gamma) \lcb Y - \omodel_0 \rcb \\
                                    &+ \lcb \omodel_1 \dot{\smooth}_g(\pscore, c, \gamma) - \omodel_0 \dot{\smooth}_{\ell}(\pscore, 1 - c, \gamma) \rcb \lcb A - \pscore \rcb  \\
                                    &+ \omodel_1 \smooth_g(\pscore, c, \gamma) - \omodel_0 \smooth_{\ell}(\pscore, 1 - c, \gamma).
    \end{align}
\end{lemma}
\begin{proof}
    Write $\psi(c) = \psi_1(c) - \psi_2(c)$, with $\psi_1(c) = \E\lcb \omodel_1(X) \smooth_g(\pscore(X), c, \gamma) \rcb$ and $\psi_2(c)= \E\lcb \omodel_0(X) \smooth_{\ell}(\pscore(X), 1 - c, \gamma)\rcb$. We derive the EIF $\eif_{\psi_1}$ of $\psi_1(c)$; the EIF $\eif_{\psi_2}$ of $\psi_2(c)$ follows from similar arguments. The EIF of $\psi(c)$ is then simply $\eif_{\psi} = \eif_{\psi_1} - \eif_{\psi_2}$. 
    
    We begin by proposing a putative EIF of $\psi_1$:
    \begin{align}
        \mathbb{IF}\left( \psi_1(c) \right) =& \mathbb{IF}\lcb \sum_x \omodel_1(x) \smooth_g(\pscore(x), c, \gamma) p(x) \rcb \hspace{2em} \text{ (Trick 1)} \\
        =& \sum_x \mathbb{IF}\lcb \omodel_1(x) \rcb \smooth_g(\pscore(x), c, \gamma) p(x) \hspace{2em} \text{ (Trick 2a)} \\
        &+ \omodel_1(x) \mathbb{IF}\lcb \smooth_g(\pscore(x), c, \gamma)  \rcb p(x) \\
        &+ \omodel_1(x) \smooth_g(\pscore(x), c, \gamma) \mathbb{IF}\lcb p(x) \rcb \\
        =& \sum_x \frac{\I(X = x, A = 1)}{p(1, x)} \left( Y - \omodel_1(x) \right) \smooth_g(\pscore(x), c, \gamma) p(x) \hspace{2em} \text{(Trick 3)} \\
        &+ \mu_1(x) \dot{\smooth}_g(\pscore(x), c, \gamma) \frac{\I[X = x]}{p(x)} \left( A - \pscore(x) \right) p(x) \hspace{2em} \text{(Trick 2b and Trick 3)} \\
        & + \omodel_1(x) \smooth_g(\pscore(x), c, \gamma) \left( \I(X = x) - p(x) \right) \\
        =& \frac{A}{\pscore(X)} \smooth_g(\pscore(X), c, \gamma) \left( Y - \omodel_1(X) \right) \\
        &+ \omodel_1(X) \dot{\smooth}_g(\pscore(X), c, \gamma) \left( A - \pscore(X) \right) \\
        &+ \omodel_1(X) \smooth_g(\pscore(X), c, \gamma)  - \psi_1(c).
    \end{align}
    Next, we derive the second-order remainder for $\psi_1$ induced by the putative EIF. Let $\bbP, \hat{\bbP} \in \model$. We denote expectations, nuisance parameters, and the EIF of the parameter with respect to $\hat{\bbP}$ as $\hat{\E}$, $\hat{\omodel}$, $\hat{\pscore}$, and $\hat{\varphi}$, respectively. We will also write $s(\pscore) \equiv \smooth_g(\pscore, c, \gamma)$.  Arguments are suppressed throughout.
    \begin{align}
        \mathsf{R}_{\psi_1}(\bbP, \hat{\bbP}) &= - \hat{\E}\lcb \eif_{\psi_1}(Z) \rcb +  \hat{\ate} \\
        &=  - \hat{\E} \lcb \frac{A}{\pscore} s(\pscore) \left( Y - \omodel_1 \right) + \omodel_1 \cdot \dot{\smooth}(\pscore)(\hat{\pscore} - \pscore) + \omodel_1 \cdot s(\pscore) \rcb + \hat{\E}\lcb \hat{\omodel}_1 \cdot s(\hat{\pscore}) \rcb \\
        &= \hat{\E} \lcb  - s(\pscore) \frac{\hat{\pscore}}{\pscore}  \left( \hat{\omodel}_1 - \omodel_1 \right) - \omodel_1 \cdot \dot{\smooth}(\pscore)(\hat{\pscore} - \pscore) - \omodel_1 \cdot s(\pscore) + \hat{\omodel}_1 \cdot s(\hat{\pscore}) \rcb \\
        &= \hat{\E} \lcb \frac{s(\pscore)}{\pscore} \left( \pscore - \hat{\pscore} \right) \left( \hat{\omodel}_1 - \omodel_1 \right)  - \omodel_1 \cdot \dot{\smooth}(\pscore)(\hat{\pscore} - \pscore) + \hat{\omodel}_1 \cdot s(\hat{\pscore}) - \hat{\omodel}_1 \cdot s(\pscore)  \rcb
    \end{align}
    Applying a second-order Taylor expansion of $\hat{\omodel} \cdot s(\pscore)$ around $s(\hat{\pscore})$, the latter three terms inside the expectation simplify:  
    \begin{align}
        & - \omodel_1 \cdot \dot{\smooth}(\pscore)(\hat{\pscore} - \pscore) + \hat{\omodel}_1 \cdot \smooth(\hat{\pscore}) - \hat{\omodel}_1 \cdot s(\pscore) \\
        =& - \omodel_1 \cdot \dot{\smooth}(\pscore)(\hat{\pscore} - \pscore) - \hat{\omodel}_1 \cdot \dot{\smooth}(\hat{\pscore})(\pscore - \hat{\pscore}) - \frac{1}{2}\hat{\omodel}_1 \cdot \ddot{\smooth}(\hat{\pscore}_1)(\pscore - \hat{\pscore})^2 + o\left(\left( \pscore - \hat{\pscore} \right)^2 \right) \\
        =& \left( \hat{\omodel}_1\cdot\dot{\smooth}(\hat{\pscore}) - \omodel_1 \cdot \dot{\smooth}(\pscore) \right)\left( \hat{\pscore} - \pscore \right) - \frac{1}{2}\hat{\omodel}_1\cdot\ddot{\smooth}(\hat{\pscore})(\pscore - \hat{\pscore})^2 + o\left(\left( \pscore - \hat{\pscore} \right)^2 \right),
    \end{align}
    where $\ddot{\smooth}$ is the second derivative of $x \mapsto s(x)$. 
    Therefore
    \begin{align}
        \mathsf{R}_{\psi_1}(\bbP, \hat{\bbP}) &=  \hat{\E} \lcb \frac{s(\pscore)}{\pscore} \left( \pscore - \hat{\pscore} \right) \left( \hat{\omodel}_1 - \omodel_1\right) + \left( \hat{\omodel}_1 \cdot \dot{\smooth}(\hat{\pscore}) - \omodel_1 \cdot \dot{\smooth}(\pscore) \right)\left( \hat{\pscore} - \pscore \right) - \frac{1}{2}\hat{\omodel}_1 \cdot \ddot{\smooth}(\hat{\pscore})(\pscore - \hat{\pscore})^2 + o\left(\left( \pscore - \hat{\pscore} \right)^2 \right)\rcb. 
    \end{align}
    The remainder term $\mathsf{R}_{\psi_1}$ is indeed second-order, depending only on products squares of differences in the nuisance parameters. Finally, we argue that the variance of the EIF is bounded. This follows from the fact that outcome is bounded, the smooth approximation function $s_{\psi}$ is bounded, and the derivative of the smooth approximation is bounded by definition. Therefore by \cite[Lemma 2]{kennedy2023density}, $\eif_{\psi_1}$ is the EIF of $\psi_1$. 
\end{proof}

\begin{lemma}[Efficient influence function of $\smoothU(c, \gamma)$]
    \label{lemma:eif-upper}
    Fix $c \in \left[0, \tfrac{1}{2} \right]$. The parameter $\smoothU(c, \gamma)$ is pathwise differentiable with uncentered efficient influence function given by
    \begin{align}
        \eif_{\smoothU}(Z) =& \eif_{\smoothATE}(Z) + 1 - \smooth_g(\pscore, c, \gamma) - \dot{\smooth}_g(\pscore, c, \gamma) \lcb A - \pscore \rcb.
    \end{align}
\end{lemma}
\begin{proof}
    The proof follows a similar structure as the proof of Lemma~\ref{lemma:eif-ate-c}, and we adopt the same notational conveniences. Within only this proof, let $\theta(c) = \E\lcb \smooth_g(\pscore, c, \gamma) \rcb$, so that $\smoothU(c) = \smoothATE(c) + \Big[ 1 - \theta(c) \Big]$. We derive the EIF of $\theta(c)$, from which the EIF for $\smoothL(c)$ follows by combining with the EIF of $\smoothATE(c)$ (Lemma~\ref{lemma:eif-ate-c}). For convenience we may write $s(\pscore) \equiv \smooth_g(\pscore, c, \gamma)$.  A putative EIF of $\theta(c)$ is
    \begin{align}
        \mathbb{IF}\left(\theta(c)\right) &= \mathbb{IF}\lcb \sum_x s(\pscore(x), c, \gamma) p(x) \rcb \hspace{2em}  \text{Trick 1}  \\
        &= \sum_x \mathbb{IF}\lcb s(\pscore(x) > c) \rcb p(x) +  s(\pscore(x), c, \gamma) \mathbb{IF}\lcb p(x) \rcb \hspace{2em} \text{Trick 2a} \\
        &= \sum_x \dot{\smooth}_g(\pscore(x), c, \gamma) \frac{\I(X = x)}{p(x)} \left( A - \pscore(x) \right) p(x) + s(\pscore(x) > c) \left( \I(X = x) - p(x) \right) \\
        &= \dot{\smooth}_g(\pscore(X), c, \gamma) \left( A - \pscore(X) \right) + \smooth_g(\pscore(X), c, \gamma) - \theta(c). 
    \end{align}
    The second-order term is
    \begin{align}
        \mathsf{R}_{\theta}(\bbP, \hat{\bbP}) &= -\hat{\E}\lcb \eif_{\theta} (Z) \rcb + \hat{\theta} \\
        &= -\hat{\E}\lcb \dot{\smooth}(\pscore)\left(A - \pscore \right) + s(\pscore) \rcb + \hat{\E}\lcb s(\hat{\pscore}) \rcb \\
        &= \hat{\E} \lcb -\dot{\smooth}(\pscore)(\hat{\pscore} - \pscore) - s(\pscore) + s(\hat{\pscore}) \rcb
    \end{align}
    Expanding $s(\pscore)$ about $s(\hat{\pscore})$ yields:
    \begin{align}
        &= \hat{\E} \lcb -\dot{\smooth}(\pscore)(\hat{\pscore} - \pscore) + \dot{\smooth}(\hat{\pscore})(\hat{\pscore} - \pscore) - \frac{1}{2}\ddot{\smooth}(\hat{\pscore} - \pscore)^2 - o\left( \left( \hat{\pscore} - \pscore \right)^2 \right) \rcb \\
        &=\hat{\E} \lcb \left( \dot{\smooth}(\hat{\pscore}) - \dot{\smooth}(\pscore) \right) \left( \hat{\pscore} - \pscore \right) - \frac{1}{2} \ddot{\smooth}(\hat{\pscore} - \pscore)^2 - o\left( \left( \hat{\pscore} - \pscore \right)^2 \right) \rcb. 
    \end{align}
    This shows that the remainder term is indeed second-order. The variance of $\eif_{\theta}$ is bounded, as both $s$ and $\dot{\smooth}$ are bounded by construction. Therefore $\eif_{\theta}$ is the EIF of $\theta(c)$ by \citealt[Lemma 2]{kennedy2023density}. The EIF of $\smoothU$ follows by combining EIFs for $\ate(c)$ and $\theta(c)$. 
\end{proof}

\begin{lemma}[Efficient influence function of $\smoothL(c, \gamma)$]
    \label{lemma:eif-lower}
    Fix $c \in \left[0, \tfrac{1}{2} \right]$. The parameter $\smoothL(c, \gamma)$ is pathwise differentiable with uncentered efficient influence function given by
    \begin{align}
        \eif_{\smoothL}(Z)         =& \eif_{\smoothATE}(Z) -1 + \smooth_{\ell}(\pscore, 1 - c, \gamma) + \dot{\smooth}_\ell(\pscore, 1 - c, \gamma) \lcb A - \pscore \rcb. 
    \end{align}
\end{lemma}
\begin{proof}
The proof follows from minor modifications to the proof of Lemma~\ref{lemma:eif-upper}.
\end{proof}

\subsection{Proof of Theorem~\ref{thm:eifs}}
The proof follows from Lemmas~\ref{lemma:eif-ate-c}, \ref{lemma:eif-lower}, and \ref{lemma:eif-upper}

\subsection{Proof of Proposition~\ref{prop:efficiency-bounds}}
\begin{proof}
For the smooth upper bound, the key of the derivation is to write
\begin{align}
    \Var(\eif_{\smoothU}) &= \E\lcb \left[ \psi_{\smoothU} - \E\left(\eif_{\smoothU}\right) \right]^2 \rcb \\
    &= \E\lcb \left[ \omodel_1 \smooth_g(\pscore, c, \gamma) - \omodel_0 \smooth_{\ell}(\pscore, 1 - c, \gamma) - \smooth_g(\pscore, c, \gamma) -\smoothU \right]^2 \rcb \\
    &\ + \E\lcb \left[ \frac{A}{\pscore} \smooth_g(\pscore, c, \gamma) \left( Y - \omodel_1 \right) - \frac{1- A}{1 - \pscore} \smooth_{\ell}(\pscore, 1 - c, \gamma) \left( Y - \omodel_0 \right) \right]^2 \rcb \\
    &\ + \E\lcb \left[ \omodel_1 \dot{\smooth}_g(\pscore, c, \gamma) - \omodel_0 \dot{\smooth}_{\ell}(\pscore, 1 - c, \gamma) \left( A - \pscore \right) \right]^2 \rcb,
\end{align}
where the equality follows because the cross-terms are zero, because the residual terms --- $Y - \E[Y\mid A, X]$ and $A - \pscore$ --- are mean zero conditional on $X$, and therefore the cross-terms have mean zero by the tower rule. Applying the tower rule again to the second and third lines and simplifying yields the result. 
\end{proof}

\subsection{Proof of Theorem~\ref{thm:weak-conv}}
The weak convergence of TMLE estimators based on cross-fitted nuisance estimators is well-established; see for example \citep[Theorem 3]{diaz2023nonparametric} and its proof, which can be easily adapted to this setting. The key steps are first to see that, by construction of the fluctuation models and the choice of loss function, the fluctuated estimates $\hat{\eta}^*$ solve the empirical EIF estimating equation, implying null first-order bias of the associated debiased plug-in estimators. Second, the use of cross-fitting for the nuisance estimators allows control of the empirical process remainder term.

\subsection{Proof of Proposition~\ref{prop:confidence-interval-width-bound}}
\begin{proof}
By Theorem~\ref{thm:weak-conv} and the consistency argument \eqref{eq:width-consistency}, it suffices to bound $W_s(c, \gamma)$. 
First, consider the smooth identification gap:
\begin{align}
    \smoothU(c, \gamma) - \smoothL(c, \gamma) = \left[1 - \E\lcb \smooth_g(\pscore, c, \gamma) \rcb\right] + \left[ 1 - \E\lcb \smooth_{\ell}(\pscore, 1 - c, \gamma) \rcb \right]. 
\end{align}

By construction of the smooth approximations \eqref{eq:example-smooth-approximations-lower} and \eqref{eq:example-smooth-approximations-upper}, recall that
\begin{align}
    \bbP(\pscore > c + \gamma) &\leq \E\lcb s_g(\pscore, c, \gamma \rcb \leq \bbP(\pscore > c) \\
    \text{and } \bbP(\pscore < 1 - c - \gamma) &\leq \E\lcb s_\ell(\pscore, 1 - c, \gamma \rcb \leq \bbP(\pscore < 1 - c).
\end{align}
Therefore,
\begin{align}
    \bbP(\pscore \leq c)  &\leq 1 - \E\lcb s_g(\pscore, c, \gamma \rcb \leq \bbP(\pscore \leq c + \gamma) \\
    \text{and } \bbP(\pscore \geq 1 - c) &\leq 1 - \E\lcb s_\ell(\pscore, 1 - c, \gamma \rcb \leq \bbP(\pscore \geq 1 - c - \gamma).
\end{align}
The smooth identification gap is therefore bounded by
\begin{align}
    \left[ 1 - \E\lcb s_g(\pscore, c, \gamma \rcb \right] + \left[ 1 - \E\lcb s_\ell(\pscore, 1 - c, \gamma \rcb \right] \leq \bbP(\pscore \leq c + \gamma) + \bbP(\pscore \geq 1 - c - \gamma).
\end{align}

Next, we derive bounds on $\sigma^2_{U}(c, \gamma)$; a similar argument holds for $\sigma^2_L(c, \gamma)$. Recall that 
\begin{align}
    \sigma^2_U &= \Var\lcb\mu_1 \smooth_g(\pscore, c, \gamma) - \mu_0 \smooth_\ell(\pscore, 1 - c, \gamma)) - s_g(\pscore, c, \gamma) \rcb \\
    &+ \E\lcb \frac{\Var(Y \mid A, X)}{\pscore(1 - \pscore)} \times s_g(\pscore, c, \gamma)^2 \times s_\ell(\pscore, 1 - c, \gamma)^2  \rcb \\
    &+ \E\lcb \Var(A \mid X) \times (\mu_1 \dot{s}_g(\pscore, c, \gamma) - \mu_0 \dot{\smooth}_\ell(\pscore, 1 - c, \gamma) - \dot{s}_g(\pscore, c, \gamma))^2  \rcb.
\end{align}
We work through the above display line by line. 
First, note that multiple quantities in the above display are bounded:
\begin{itemize}
    \item $A \in \{0, 1\}$ and $Y \in [0, 1]$ implies $\Var(A | X) \in [0, 1/4]$ and $\Var(Y \mid A, X) \in [0, 1/4]$.
    \item Similarly, $Y \in [0, 1]$ implies $\mu_1 \in [0, 1]$ and $\mu_0 \in [0, 1]$. 
    \item The smooth approximations satisfy, for all $\pscore \in [0, 1]$, $s_\ell(\pscore, c, \gamma) \in [0, 1]$ and $s_g(\pscore, c, \gamma) \in [0, 1]$.
    \item By assumption, for all $\pscore \in [0, 1]$, $|\dot{s}_\ell(\pscore, c, \gamma)| < D$ and $|\dot{s}_g(\pscore, c, \gamma)| < D$.
\end{itemize}
Applying these bounds, we have that
\begin{align}
     \mu_1 \smooth_g(\pscore, c, \gamma) - \mu_0 \smooth_\ell(\pscore, 1 - c, \gamma)) - s_g(\pscore, c, \gamma) \in [-1, 1],
\end{align}
implying that
\begin{align}
     \Var\lcb\mu_1 \smooth_g(\pscore, c, \gamma) - \mu_0 \smooth_\ell(\pscore, 1 - c, \gamma)) - s_g(\pscore, c, \gamma) \rcb \leq 1.
\end{align}
Next, note that by Property~\ref{prop:bounds}, $s_g(\pscore, c, \gamma) \leq \I(x > c)$, so $s_g(\pscore, c, \gamma)^2 \leq s_g(\pscore, c, \gamma) \leq \I(\pscore > c)$ (likewise for $s_\ell$). Therefore 
\begin{align}
    \frac{\smooth_g(\pscore, c, \gamma)^2 \times s_\ell(\pscore, 1 - c, \gamma)^2}{\pi(1-\pi)} \leq \frac{\I(\pscore > c)\I(\pscore < 1 - c)}{\pscore(1-\pscore)} \leq \frac{1}{c(1-c)}. 
\end{align}
Because $c \in (0, \frac{1}{2})$, $c(1-c) \geq c/2$. Combined with the bound on the conditional variance of $Y$, this yields
\begin{align}
    \E\lcb \frac{\Var(Y \mid A, X)}{\pscore(1 - \pscore)} \times s_g(\pscore, c, \gamma)^2 \times s_\ell(\pscore, 1 - c, \gamma)^2  \rcb \leq \frac{1}{4c(1 - c)} \leq \frac{1}{2c}.
\end{align}
Finally, 
\begin{align}
    \E\lcb \Var(A \mid X) \times (\mu_1 \dot{s}_g(\pscore, c, \gamma) - \mu_0 \dot{\smooth}_\ell(\pscore, 1 - c, \gamma) - \dot{s}_g(\pscore, c, \gamma))^2  \rcb \leq \frac{9D^2}{4\gamma^2},
\end{align}
because
\begin{align}
    |\mu_1 \dot{s}_g(\pscore, c, \gamma) - \mu_0 \dot{\smooth}_\ell(\pscore, 1 - c, \gamma) - \dot{s}_g(\pscore, c, \gamma)| \leq |\mu_1| |\dot{s}_g(\pscore, c, \gamma)| + |\mu_0| |\dot{\smooth}_\ell(\pscore, 1 - c, \gamma)| + |\dot{s}_g(\pscore, c, \gamma)| \leq 3D. 
\end{align}
Summing the inequalities yields the stated result.
\end{proof}

\subsection{Lemma: bounded derivatives of $\smooth_{\ell}$ and $\smooth_g$}
\begin{lemma}
\label{lemma:bounded-derivatives-smooth-approximations}
Let $\mathcal{C} = [c_1, c_2]$ for $0 < c_1 < c_2 < \tfrac{1}{2}$ and $\Gamma = [\gamma_1, \gamma_2]$ for $0 < \gamma_1 < \gamma_2 < \infty$. Then, for $\smooth_{\ell}$ and $\smooth_g$ as given in \eqref{eq:example-smooth-approximations-lower} and \eqref{eq:example-smooth-approximations-upper}:
\begin{enumerate}
    \item For all $c \in \mathcal{C}$ and $\gamma \in \Gamma$, the maps $x \mapsto \smooth_\ell(x, c, \gamma)$ and $x \mapsto \smooth_g(x, c, \gamma)$ are twice differentiable with first derivatives $\dot{\smooth}_{\ell}$ and $\dot{\smooth}_g$ and second derivatives $\ddot{\smooth}_{\ell}$ and $\ddot{\smooth}_g$. Moreover, $\dot{\smooth}_{\ell}$, $\dot{\smooth}_g$, $\ddot{\smooth}_{\ell}$, and $\ddot{\smooth}_g$ are bounded. 
    \item For all $x \in [0, 1]$ and $\gamma \in \Gamma$, the maps $c \mapsto \smooth_{\ell}(x, c, \gamma)$, $c \mapsto \dot{\smooth}_{\ell}(x, c, \gamma)$, $c \mapsto \smooth_{g}(x, c, \gamma)$, and $c \mapsto \dot{\smooth}_{g}(x, c, \gamma)$ are differentiable with bounded first derivative.
    \item For all $x \in [0, 1]$ and $c \in \mathcal{C}$, the maps $\gamma \mapsto \smooth_{\ell}(x, c, \gamma)$, $\gamma \mapsto \dot{\smooth}_{\ell}(x, c, \gamma)$, $\gamma \mapsto \smooth_{g}(x, c, \gamma)$, and $\gamma \mapsto \dot{\smooth}_{g}(x, c, \gamma)$ are differentiable with bounded first derivative.
\end{enumerate}
\end{lemma}
\begin{proof}
We demonstrate the results focusing on $\smooth_{g}$, as the derivations are similar for $\smooth_{\ell}$. Recall that
\begin{align}
    \smooth_{g}(x, c, \gamma) = \begin{cases}
        1, & x \geq c + \gamma, \\
        0, & x \leq c, \\
        1 - \exp\left[1 + \frac{1}{\lcb (x - c) \slash \gamma \rcb^2 - 1} \right], &\text{otherwise.}
    \end{cases}
\end{align}
As a preliminary, let $f : \mathbb{R} \to [0, 1]$ be characterized by
\begin{align}
    f(x) = \begin{cases}
        1, x \geq 1, \\
        0, x \leq 0, \\
        1 - \exp\left[ 1 + \frac{1}{x^2 - 1} \right], \text{ otherwise.}
    \end{cases}
\end{align}
See that $\smooth_g(x, c, \gamma) = f((x - c) / \gamma)$. The first and second derivatives of $f$ are given by
\begin{align}
    \dot{f}(x) &= \begin{cases}
        (1 - f(x)) \frac{2x}{(x^2 - 1)^2}, & 0 < x < 1, \\
        0, & \text{otherwise.}
    \end{cases}, \\
    \ddot{f}(x) &= \begin{cases}
        (1 - f(x)) \frac{-6x^4 + 2}{(x^2 - 1)^4}, & 0 < x < 1, \\
        0, & \text{otherwise.}
    \end{cases}
\end{align}
It can be shown via (somewhat tedious) calculations that both $\dot{f}$ and $\ddot{f}$ are bounded. Now we address each part of the lemma in turn.
\begin{enumerate}
    \item By the chain rule, $x \mapsto \smooth_g(x, c, \gamma)$ is twice differentiable with first and second derivatives given by
    \begin{align}
        \dot{\smooth}_g(x, c, \gamma) &= \frac{1}{\gamma} \dot{f}((x - c) / \gamma), \\
        \ddot{\smooth}_g(x, c, \gamma) &= \frac{1}{\gamma^2} \ddot{f}((x - c)/\gamma).
    \end{align}
    Therefore, because $\dot{f}$ and $\ddot{f}$ and $\gamma > \gamma_1 > 0$ by assumption, $\ddot{\smooth}_g$ and $\ddot{\smooth}_g$ are bounded as a function of $x$. 
    \item Again by the chain rule, $c \mapsto \smooth_g(x, c, \gamma)$ is differentiable with derivative $c \mapsto \tfrac{-1}{\gamma} \dot{f}((x - c) \gamma)$, which is bounded by the same argument as for part 1. Furthermore, the function $c \mapsto \dot{s}_g(x, c, \gamma)$ is differentiable with respect to $c$ with derivative $c \mapsto -\tfrac{1}{\gamma^2} \ddot{f}((x - c) / \gamma)$, which is bounded by a similar argument.
    \item Again by the chain rule, $\gamma \mapsto \smooth_g(x, c, \gamma)$ is differentiable with derivative $\gamma \mapsto -\tfrac{x - c}{\gamma^2} \dot{f}((x - c)/\gamma)$, which is bounded by a similar argument as above. Furthermore, the function $\gamma \mapsto \dot{s}_g(x, c, \gamma)$ is differentiable with respect to $\gamma$ with derivative $\gamma \mapsto -\tfrac{x - c}{\gamma^3}\ddot{f}((x - c) / \gamma)$, which is again bounded. 
\end{enumerate}
\end{proof}

\subsection{Proof of Theorem~\ref{thm:weak-convergence-uniform}}
We focus on proving the results for the upper bound parameter $\smoothU$, as similar arguments can be used to derive results for the lower bound parameter. The proof follows similar logic as that of \citep[Theorem 3]{kennedy2019nonparametric}

Let $\|f\|_{\mathcal{C},\Gamma} = \sup_{c\in\mathcal{C},\gamma\in\Gamma}|f(c, \gamma)|$ be the supremum norm over $\mathcal{C}\times\Gamma$. First, we show that the empirical process $U_n(c,\gamma) = \mathbb{G}_{n}\lcb \widetilde{\varphi}_U(Z, \eta, c, \gamma) \rcb \leadsto \mathbb{G}\lcb \widetilde{\varphi}_U(Z, \eta, c, \gamma) \rcb$  in $\ell^\infty(\mathcal{C}\times\Gamma)$. Then, by Theorem~\ref{thm:weak-conv}, we have that 
\begin{align}
    \| \widehat{U}_n(c, \gamma) - U_n(c, \gamma) \|_{\mathcal{C},\mathcal{D}} = o_{\bbP}(1),
\end{align}
where $\widehat{U}_n(c, \gamma) = \sqrt{n}\lcb \widehat{\smoothU}(c, \gamma) - \smoothU(c, \gamma) \rcb \slash \hat{\sigma}(c, \gamma)$, which completes the argument. 

Let $\mathcal{F}_L = \{ \eif_L(Z, \eta, c, \gamma) : c \in \mathcal{C}, \gamma \in \Gamma \}$ and $\mathcal{F}_U = \{ \eif_U(Z, \eta, c, \gamma) : c \in \mathcal{C}, \gamma \in \Gamma \}$ be the function classes of the efficient influence functions for the lower and upper bound parameters, respectively, in terms of $c \in \mathcal{C}$ and $\gamma \in \Gamma$. Theorem~\ref{thm:weak-conv} establishes weak convergence of $\widehat{\smoothL}(c, \gamma)$ and $\widehat{\smoothU}(c, \gamma)$ to a limiting Gaussian distribution for a fixed $c \in \mathcal{C}$ and $\gamma \in \Gamma$. To show uniform weak convergence over all $c \in \mathcal{C}$ and $\gamma \in \Gamma$ we show that $\mathcal{F}_L$ and $\mathcal{F}_U$ are both Lipschitz continuous in $c$ and $\gamma$, which implies that the function classes have finite bracketing integrals and are therefore Donsker \citep[Theorem 2.5.6]{vanderVaart&Wellner96}.

We show that the function classes are Lipschitz by showing that they have bounded derivatives. The derivative of $\varphi_{\ate}$ with respect to $c$ is given by
\begin{align}
    \frac{\partial}{\partial c} \eif_{\ate}(Z, \eta, c, \gamma) &= \omodel_1 \frac{\partial}{\partial c} \smooth_g(\pscore, c, \gamma) - \omodel_0 \frac{\partial}{\partial c} \smooth_{\ell}(\pscore, 1 - c, \gamma) \\
                                    &\ + \frac{A}{\pscore} \frac{\partial}{\partial c} \smooth_g(\pscore, c, \gamma) \left( Y - \omodel_1 \right) - \frac{1- A}{1 - \pscore} \frac{\partial}{\partial c} \smooth_{\ell}(\pscore, 1 - c, \gamma) \left( Y - \omodel_0 \right) \\
                                    &\ + \lcb \omodel_1 \frac{\partial}{\partial c}\dot{\smooth}_g(\pscore, c, \gamma) - \omodel_0 \frac{\partial}{\partial c} \dot{\smooth}_{\ell}(\pscore, 1 - c, \gamma) \rcb \left( A - \pscore \right), 
\end{align}
By assumption, the derivatives of $\smooth_{\ell}$, $\smooth_g$, $\dot{\smooth}_{\ell}$, and $\dot{\smooth}_g$ with respect to $c$ are bounded. Therefore, $\eif_{\ate}(Z, \eta, c, \gamma)$ is Lipschitz continuous with respect to $c$ because the sum of Lipschitz continuous functions is Lipschitz. For $\eif_U$, see that
\begin{align}
    \frac{\partial}{\partial c} \eif_U(Z, \eta, c, \gamma) &= \frac{\partial}{\partial c}  \eif_{\smoothATE}(Z, \eta, c, \gamma) - \frac{\partial}{\partial c}\smooth_{g}(\pscore, 1 - c, \gamma) - \frac{\partial}{\partial c}\dot{\smooth}_{\ell}(\pscore, 1 - c, \gamma) \left( A - \pscore \right).
\end{align}
Again, by assumption the derivatives of $\smooth_{\ell}$ and $\dot{\smooth}_{\ell}$ are bounded. This implies that $\eif_U$ is Lipschitz in $c$; an identical argument shows that $\eif_U$ is Lipschitz in $\gamma$, and can be easily applied to show the same properties hold for $\eif_U$.

\subsection{Proof of Theorem~\ref{thm:bootstrap}}

\begin{proof}
    For this proof we use the same notation as in Algorithm~\ref{alg:bootstrap}.

    \noindent \emph{\textbf{Joint weak convergence:}} \\ 
    By Theorem~\ref{thm:weak-conv}, the estimators satisfy a joint linear expansion; i.e., 
    \[
    \sqrt{n} \begin{pmatrix} \widehat L_1 - L_1 \\ \vdots \\ \widehat L_K - L_K \\ \widehat U_1 - U_1 \\ \vdots \\ \widehat U_K - U_K \end{pmatrix} = n^{-1/2} \sum_{i=1}^{n} \begin{pmatrix} \varphi_{L,1}(Z_i, \eta) - L_1 \\ \vdots \\ \varphi_{L,K}(Z_i, \eta) - L_K \\ \varphi_{U,1}(Z_i, \eta) - U_1  \\ \vdots \\ \varphi_{U,K}(Z_i, \eta) - U_K \end{pmatrix} + \begin{pmatrix} o_{\bbP}(1) \\ \vdots \\ o_{\bbP}(1) \end{pmatrix}.
    \]
    The multivariate CLT yields joint weak convergence to a multivariate Gaussian with covariance matrix $\E \{ \varphi(Z) \varphi(Z)^T \}$ where \\ $\varphi(Z) = \Big( \varphi_{L,1}(Z_i, \eta) - L_1, \dots, \varphi_{L,K}(Z_i, \eta) - L_K, \varphi_{U,1}(Z_i, \eta) - U_1, \dots, \varphi_{U,K}(Z_i, \eta) - U_K \Big)^T$.

    \begin{remark}
        From here, one could use the multivariate Gaussian limit to conduct inference by sampling from $N(0, \E(\varphi \varphi^T) )$. We instead use the multiplier bootstrap.
    \end{remark}

    \noindent \emph{\textbf{Event equivalence:}} \\
     Let $T_{L,k} = \sqrt{n} \left( \frac{\widehat L_k - L_k}{\widehat \sigma_{L,k}} \right)$ and $T_{U,k} = \sqrt{n} \left( \frac{U_k - \widehat U_k}{\widehat \sigma_{U,k}} \right)$. For any fixed $t\ge0$ and each $k$,
    \[
    L_k \geq \widehat L_k - t \tfrac{\widehat\sigma_{L,k}}{\sqrt n} \iff T_{L,k} \leq t,
    \qquad
    U_k \leq \widehat U_k + t \tfrac{\widehat\sigma_{U,k}}{\sqrt n} \iff T_{U,k} \leq t.
    \]
    The inequalities on the left-hand side of each equivalence are reversed because $T_{L,k}$ and $T_{U,k}$ are reversed. Next, notice that the following two events are equivalent:
    \[
    \Big\{\forall k:\ L_k \geq \widehat L_k -t \tfrac{\widehat\sigma_{L,k}}{\sqrt n}\ \ \text{and}\ \
    U_k \le \widehat U_k + t \tfrac{\widehat\sigma_{U,k}}{\sqrt n}\Big\}
    \]
    and
    \[
    \Big\{\,\max_{k}\max\{T_{L,k},T_{U,k} \}\le t\,\Big\}. 
    \]
    Moreover, notice both events also imply a third event. Let $k_L = \argmax_k \widehat L_k - t \widehat \sigma_{L,k} / \sqrt{n}$ and $k_U = \argmin_k \widehat U_k + t \widehat \sigma_{U,k} / \sqrt{n}$. Then, both events above imply
    \[
    \left\{ L_{k_L} \geq \max_k \widehat L_k - t \tfrac{\widehat \sigma_{L,k}}{\sqrt{n}}, U_{k_U} \leq \min_k \widehat U + t \tfrac{\widehat \sigma_{U,k}}{\sqrt{n}} \right\},
    \]
    By Proposition~\ref{prop:smooth-bounds}, $\E(Y^1 - Y^0) \in[L_k,U_k]$ for every $k$; indeed, $\E (Y^1 - Y^0) \in [L_j, U_k]$ for every $k,j$. Therefore, the previous event implies 
    \[
    \left\{ \E(Y^1 - Y^0) \in \left[ \max_k \widehat L_k - t \tfrac{\widehat \sigma_{L,k}}{\sqrt{n}}, \min_k \widehat U_k + t \tfrac{\widehat \sigma_{U,k}}{\sqrt{n}} \right] \right\},
    \]
    Therefore, 
    \begin{align*}
        \bbP \Big\{\max_{k}\max\{T_{L,k}, T_{U,k}\}\le t \Big\} &\leq \bbP \left\{ \forall\ k:\ \E(Y^1-Y^0) \in \Big[ \widehat L_k - t \tfrac{\widehat\sigma_{L,k}}{\sqrt n},\ \widehat U_k + t \tfrac{\widehat\sigma_{U,k}}{\sqrt n} \Big] \right\} \text{ and } \\
        \bbP \Big\{\max_{k}\max\{T_{L,k}, T_{U,k}\}\le t \Big\} &\leq \bbP \left\{ \E(Y^1 - Y^0) \in \left[ \max_k \widehat L_k - t \tfrac{\widehat \sigma_{L,k}}{\sqrt{n}}, \min_k \widehat U + t \tfrac{\widehat \sigma_{U,k}}{\sqrt{n}} \right] \right\} 
    \end{align*}     
    and so a lower bound on the left-hand probability implies a lower bound on the right-hand probabilities.

    \noindent \emph{\textbf{Multiplier bootstrap approximation:}} \\
    Let $\widehat M\ :=\ \max_{k}\max\{T_{L,k}, T_{U,k} \}$. Using the notation from above, the left-hand probability is $\bbP \left( \widehat M \leq t \right)$. We want to show $\bbP \left( \widehat M \leq \widehat q_{1-\alpha} \right) \geq 1- \alpha + o_{\bbP}(1)$.

    By the joint weak convergence result above and the continuous mapping theorem for the max map, \( \widehat M \leadsto M \). Hence,
    \[
    \sup_{t \in \mathbb{R}} \left|  \bbP \left( \widehat M \leq t \right) - \bbP \left( M \leq t \right) \right| = o(1).
    \]
    Meanwhile, by standard analysis of the multiplier bootstrap and the continuity of the max-max transformation with a fixed set of thresholds \citep[Theorem 2.9.6]{vanderVaart&Wellner96},
    \[
    \sup_{t \in \mathbb{R}} \left| \bbP \left( \widehat M^{(b)} \leq t \mid \mathbf{Z}_n \right) - \bbP \left( \widehat M \leq t \right) \right| = o_{\bbP}(1), 
    \]
    where $\mathbf{Z}_n$ denotes the data, so the first probability is conditional on the data.

    Next, let $\widetilde q_{1-\alpha}$ denote the exact bootstrap $1-\alpha$ quantile, so that $\bbP \left( \widehat M^{(b)} \leq \widetilde q_{1-\alpha} \mid \mathbf{Z}_n \right) = 1-\alpha$. By the previous display, 
    \[
    \bbP \left( \widehat M \leq \widetilde q_{1-\alpha} \right) = 1-\alpha + o_{\bbP}(1).
    \]
    Finally, we address the Monte Carlo error. Assuming $B \to \infty$ as $n \to \infty$,
    \[
    \widehat q_{1-\alpha} \stackrel{p}{\to} \widetilde q_{1-\alpha}.
    \]
    Therefore,
    \[
    \bbP \left( \widehat M \leq \widehat q_{1-\alpha} \right) = 1-\alpha + o_{\bbP}(1).
    \]
    The results follow from the definition of $\widehat M$ and the event equivalence argument above.

\end{proof}

\subsection{Proof of Proposition~\ref{prop:bootstrap-uniform}}
\begin{proof}
    We sketch the proof here, as it requires combining previous results with standard arguments for the multiplier bootstrap on Donsker classes. First, by the standard multiplier bootstrap CLT for Donsker classes \citep[Theorem 2.9.6]{vanderVaart&Wellner96}, the conditional law of $(T_L^{(b)}(c, \gamma), T_U^{(b)}(c, \gamma))$ converges to the law of $(\mathbb{G}_L, \mathbb{G}_U)$ in $\ell^\infty(\mathcal{C}\times \Gamma)$, where $\mathbb{G}_L$ and $\mathbb{G}_U$ are the limiting Gaussian processes from Theorem~\ref{thm:weak-convergence-uniform}. Next, by the continuous mapping theorem, the conditional distribution of $M_{\infty}^{(b)} = \sup_{(c, \gamma)} \max(T_L^{(b)}(c,\gamma), -T_U^{(b)}(c,\gamma))$ converges to $\sup_{(c,\gamma)} \max(\mathbb{G}_L(c,\gamma), - \mathbb{G}_U(c, \gamma))$. Finally, the Donsker property of the EIF function class (see Proof of Theorem~\ref{thm:weak-convergence-uniform}) implies asymptotic equicontinuity of the bootstrap processes, which implies that $|\widehat{M}^{(b)} - M_\infty^{(b)} | \xrightarrow{p} 0$. Combining these results, we have that the distribution of $\widehat{M}^{(b)}$ converges to the distribution of $\sup_{(c, \gamma)} \max(\mathbb{G}_L(c, \gamma), -\mathbb{G}_U(c, \gamma))$ as $n, K \to \infty$ with $\delta_K \to 0$. The bootstrap critical values then converge: $\widehat{q}_{1-\alpha}^{(K)} \xrightarrow{p} q_{1-\alpha}^\infty$, where the error from using $B$ bootstrap draws is $o_\bbP(1)$ when $B\to\infty$. 
\end{proof}

\subsection{Proof of Theorem \ref{thm:equiv-tmle}}
\begin{proof}
We begin by proving Part 1 of the theorem, which is the simpler case. Let
\[
E_n =
\left\{
\|\widehat\pscore-\pscore\|_\infty \le \varepsilon
\right\}.
\]
Note that, by assumption, $\mathbb{P}(E_n)\to 1$.  In the event $E_n$, for all $k\in[K]$ and $i\in[n]$, we have:
\begin{align}
    \widehat\pscore(X_i)
    &\ge
    \pscore(X_i)-\varepsilon
    \ge
    \delta-\varepsilon
    \ge
    c_k+\gamma_k.
\end{align}
Likewise,
\begin{align}
    \widehat\pscore(X_i)
    &\le
    \pscore(X_i)+\varepsilon
    \le
    1-\delta+\varepsilon
    \le
    1-c_k-\gamma_k.
\end{align}
Therefore, by \eqref{eq:example-smooth-approximations-lower} and \eqref{eq:example-smooth-approximations-upper},
\begin{align}
    s_g\{\widehat\pscore(X_i),c_k,\gamma_k\}=1,
    \qquad
    \dot s_g\{\widehat\pi(X_i),c_k,\gamma_k\}=0,
    \label{eq:sg-one-2}
\end{align}
and
\begin{align}
    s_\ell\{\widehat\pi(X_i),1-c_k,\gamma_k\}=1,
    \qquad
    \dot s_\ell\{\widehat\pi(X_i),1-c_k,\gamma_k\}=0
    \label{eq:sl-one-2}
\end{align}
for every $i$ and $k$, on the event $E_n$.

Substituting \eqref{eq:sg-one-2}--\eqref{eq:sl-one-2} into the efficient influence functions in Theorem~\ref{thm:eifs} shows that, on $E_n$,
\[
\varphi_L(\,\cdot\,;\eta,c_k,\gamma_k)
=
\varphi_U(\,\cdot\,;\eta,c_k,\gamma_k)
=
\varphi_\psi(\,\cdot\,;\eta)
\qquad\text{for all }k.
\]
This implies that the corresponding estimators and standard error estimators coincide across all $k$ and coincide with the standard ATE estimator:
\[
\widehat L_k=\widehat U_k=\widehat\psi,
\qquad
\widehat\sigma_{L,k}=\widehat\sigma_{U,k}=\widehat\sigma_\psi
\qquad\text{for all }k,
\]
on $E_n$. Hence
\begin{align}
    \underline C_n
    &=
    \widehat\psi-\widehat q_{1-\alpha}\frac{\widehat\sigma_\psi}{\sqrt n},
    \label{eq:lower-collapse}
    \\
    \overline C_n
    &=
    \widehat\psi+\widehat q_{1-\alpha}\frac{\widehat\sigma_\psi}{\sqrt n},
    \label{eq:upper-collapse}
\end{align}
on $E_n$.

We will now show that the multiplier-bootstrap critical value converges to $z_{1-\alpha/2}$. On $E_n$, the bootstrap studentized statistics in Algorithm~\ref{alg:bootstrap} satisfy
\[
T_{L,k}^{(b)}=T_{U,k}^{(b)}\equiv S_n^{(b)}
\qquad\text{for all }k,
\]
where
\[
S_n^{(b)}
=
\frac{1}{\sqrt n}
\sum_{i=1}^n
\xi_i^{(b)}
\left\{
\frac{\widehat\varphi_\psi(Z_i)-\widehat\psi}{\widehat\sigma_\psi}
\right\}.
\]
Therefore the bootstrap max--max statistic reduces to
\[
\widehat M^{(b)}
=
\max_{1\le k\le K}\max(T_{L,k}^{(b)},-T_{U,k}^{(b)})
=
|S_n^{(b)}|.
\]
Condition on the observed data \(Z_1,\dots,Z_n\), and define
\[
R_{ni}
=
\frac{\widehat\varphi_\psi(Z_i)-\widehat\psi}{\widehat\sigma_\psi},
\qquad
a_{ni}
=
\frac{R_{ni}}{\sqrt n},
\qquad
S_n^{(b)}=\sum_{i=1}^n a_{ni}\xi_i^{(b)}.
\]
We now verify a conditional scalar triangular-array CLT. Conditional on the data, the variables $Y_{ni}^{(b)}:=a_{ni}\xi_i^{(b)}$
are independent, mean-zero, and satisfy
\[
\Var\!\left(Y_{ni}^{(b)}\mid Z_1,\dots,Z_n\right)=a_{ni}^2.
\]
Hence
\[
\Var\!\left(S_n^{(b)}\mid Z_1,\dots,Z_n\right)=\sum_{i=1}^n a_{ni}^2.
\]
Moreover,
\[
\sum_{i=1}^n a_{ni}^2
=
\frac{1}{\widehat\sigma_\psi^2}
\frac{1}{n}\sum_{i=1}^n
\left\{\widehat\varphi_\psi(Z_i)-\widehat\psi\right\}^2.
\]
Thus, if \(\widehat\sigma_\psi^2\) is the empirical second moment, then
\(\sum_i a_{ni}^2=1\) exactly; if it is the unbiased sample variance, then
\(\sum_i a_{ni}^2=(n-1)/n\to 1\). Hence
\[
\sum_{i=1}^n a_{ni}^2 \to 1.
\]
Also, it can be verified that
\[
\max_{1\le i\le n}|a_{ni}|=o_{\mathbb P}(1).
\]
Therefore, for every fixed \(\eta>0\),
\begin{align*}
&\sum_{i=1}^n
\mathbb E\!\left[
Y_{ni}^{(b)\,2}\mathbf 1\{|Y_{ni}^{(b)}|>\eta\}
\mid Z_1,\dots,Z_n
\right] \\
&\qquad=
\sum_{i=1}^n
a_{ni}^2
\mathbb E\!\left[
\xi_i^2\mathbf 1\{|\xi_i|>\eta/|a_{ni}|\}
\right] \\
&\qquad\le
\left(\sum_{i=1}^n a_{ni}^2\right)
\sup_{1\le i\le n}
\mathbb E\!\left[
\xi_i^2\mathbf 1\{|\xi_i|>\eta/|a_{ni}|\}
\right]
\to 0,
\end{align*}
because \(\max_i |a_{ni}|\to 0\) and \(\mathbb E(\xi_i^2)<\infty\). Hence the conditional Lindeberg condition holds.

By the Lindeberg--Feller theorem, conditional on the data,
\[
S_n^{(b)}=\sum_{i=1}^n a_{ni}\xi_i^{(b)} \leadsto N(0,1).
\]
Equivalently,
\[
\sup_{t\in\mathbb R}
\left|
\mathbb P\!\left(
S_n^{(b)}\le t \mid Z_1,\dots,Z_n
\right)-\Phi(t)
\right|
\to 0.
\]
Since \(\widehat M^{(b)}=|S_n^{(b)}|\), the continuous mapping theorem yields
\[
\sup_{t\in\mathbb R}
\left|
\mathbb P\!\left(
\widehat M^{(b)}\le t \mid Z_1,\dots,Z_n
\right)
-
\mathbb P\{|N(0,1)|\le t\}
\right|
\to 0.
\]
Therefore, it follows that the exact conditional bootstrap quantile $\widetilde q_{1-\alpha}$ satisfies
\[
\widetilde q_{1-\alpha}\stackrel{p}{\to} z_{1-\alpha/2}.
\]
Subtracting the lower endpoint of $\widehat{CI}_\psi$ from \eqref{eq:lower-collapse}, conditional on $E_n$,
\begin{align}
    \underline C_n-\underline C_{\psi,n}
    &=
    \left\{
    \widehat\psi-\widehat q_{1-\alpha}\frac{\widehat\sigma_\psi}{\sqrt n}
    \right\}
    -
    \left\{
    \widehat\psi-z_{1-\alpha/2}\frac{\widehat\sigma_\psi}{\sqrt n}
    \right\}
    \notag\\
    &=
    \left(
    z_{1-\alpha/2}-\widehat q_{1-\alpha}
    \right)
    \frac{\widehat\sigma_\psi}{\sqrt n}.
\end{align}
Thus
\begin{align}
    \sqrt n\,(\underline C_n-\underline C_{\psi,n})
    =
    \left(
    z_{1-\alpha/2}-\widehat q_{1-\alpha}
    \right)\widehat\sigma_\psi
    \stackrel{p}{\to} 0.
\end{align}
Since $\mathbb P(E_n)\to 1$, the same convergence holds unconditionally. A similar argument shows the result for the upper endpoint of the confidence interval.

Next, we consider Part 2 of the theorem. Divide the threshold-smoothness pairs $\{ (c_k, \gamma_k ) \}_{k=1}^K$ into two disjoint sets: $\mathcal{K}_{\mathsf{in}} = \{ (c_k, \gamma_k) : c_k + \gamma_k \leq \delta - \epsilon \}$ and $\mathcal{K}_{\mathsf{out}} = \{ (c_k, \gamma_k) : c_k + \gamma_k > \delta - \epsilon \}$. By assumption, $\mathcal{K}_{\mathsf{in}}$ is non-empty; that is, there is at least one threshold-smoothness pair that falls below the strong overlap threshold $\delta$. 

Let $\Delta_k = \psi - L_s(c_k, \gamma_k)$. For $k \in \mathcal{K}_{\mathsf{in}}$, $\Delta_k = 0$, and for  $k \in \mathcal{K}_{\mathsf{out}}$, $\Delta_k > 0$. Let $\Delta_{\min} = \min_{k \in \mathcal{K}_{\mathsf{out}}} \Delta_k > 0$. Conditioning on the event $E_n$, for large enough $n$ and with $k \in \mathcal{K}_{\mathsf{out}}$,
\begin{align}
    \hat{L}_k - \hat{q}_{1-\alpha} \frac{\hat{\sigma}_{L,k}}{\sqrt{n}} = L_s(c_k, \gamma_k) + O_P(n^{-1/2}) = \psi - \Delta_k + O_P(n^{-1/2}). 
\end{align}
For $k \in \mathcal{K}_{\mathsf{in}}$, on the other hand,
\begin{align}
    \hat{L}_k - \hat{q}_{1-\alpha} \frac{\hat{\sigma}_{L,k}}{\sqrt{n}} = L_s(c_k, \gamma_k) + O_P(n^{-1/2}) = \psi + O_P(n^{-1/2}). 
\end{align}
For all $k \in \mathcal{K}_{\mathsf{out}}$, $\Delta_k > 0$ and does not decrease with sample size. Therefore, the $k$ that maximizes the lower bound will fall in $\mathcal{K}_{\mathsf{in}}$. An analogous argument holds for the upper bound.

Next, see that the bootstrap max--max statistic satisfies
\begin{align} 
    \widehat M^{(b)} &= \max_{1\le k\le K}\max\left(T_{L,k}^{(b)},-T_{U,k}^{(b)}\right) \\
    &\geq \max_{k \in \mathcal{K}_{\mathsf{in}}} \max\left(T_{L,k}^{(b)},-T_{U,k}^{(b)}\right) \\
    &= |S_n^{(b)}|.
\end{align}
Following the argument of Part 1, this implies that $\hat{q}_{1-\alpha} \geq z_{1-\alpha/2} + o_\bbP(1)$.

These two results taken together imply that the narrowest CI will be formed using point estimates from $\mathcal{K}_{\mathsf{in}}$, which have no identification gap, but the width of the CIs may be slightly larger than the Wald CI because $\hat{q}_{1-\alpha} \geq z_{1-\alpha/2} + o_{\bbP}(1)$. This establishes Part 2 of the theorem.

\end{proof}

\subsection{Proof of Theorem~\ref{thm:somewhat-weak-overlap}}
\begin{proof}
    First, we analyze the width of the point estimator confidence interval in this setting. The ATE efficiency bound is given by
    \begin{align}
        \Var\lcb \eif(Z, \eta) \rcb &= \Var(\mu_1 - \mu_0) + \E\lcb \frac{\Var(Y \mid A, X)}{\pi (1 - \pi)} \rcb.
    \end{align}
    By the boundedness of the outcome, the variance terms are upper bounded by $\frac{1}{4}$. By assumption, the conditional variance of the outcome is lower bounded by $\sigma^2_{\min}$. Taken together, the ATE efficiency bound is bounded by
    \begin{align}
        \E\lcb \frac{\sigma^2_{\min}}{\pscore(1-\pscore)} \rcb \leq \E\lcb \eif(Z, \eta)^2 \rcb \leq \frac{1}{4} + \frac{1}{4} \E\lcb \frac{1}{\pscore(1 - \pscore)} \rcb.
    \end{align}
    An adaption of the the proof of \citealt[Proposition 1]{dorn2025sensitivity} yields an upper bound on $\E\lcb 1 \slash (\pscore(1 - \pscore)) \rcb$:
    \begin{align}
        \E\lcb \frac{1}{\pi(1 - \pi)} \rcb  &= \int_0^{\infty} \bbP\left(\frac{1}{\pscore (1 - \pscore)} > t\right) dt \\
        &= \int_0^{\infty} \bbP\left(\pscore (1 - \pscore) < \tfrac{1}{t} \right) dt  \\
        &= 1 + \int_1^{\infty} \bbP\left(\pscore (1 - \pscore) < \tfrac{1}{t} \right) dt \\
        &\leq 1 + \int_1^{\infty} \bbP\left(\pscore < \tfrac{1}{t} \right) dt + \int_1^\infty \bbP\left(1 - \pscore < \tfrac{1}{t}\right) dt \\
        &\leq 1 + 2 C \int_1^\infty t^{1-\gamma_0} dt \\
        &= 1 + \frac{2C}{\gamma_0 - 2}.
    \end{align}
    Next, a lower bound on $\E\lcb 1 / (\pscore(1 - \pscore)) \rcb$ is given by
    \begin{align}
        \E\lcb \frac{1}{\pscore(1 - \pscore)} \rcb &= \int_0^\infty \bbP\left(\pscore(1- \pscore) < \tfrac{1}{t} \right)dt \\
        &\geq \int_2^\infty \bbP\left(\pscore(1- \pscore) < \tfrac{1}{t} \right)dt \\
        &\geq \int_2^\infty \bbP\left(\pscore < \tfrac{1}{t}\right) dt \\
        &\geq \int_2^\infty C' t^{1 - \gamma_0} dt \\
        &= \frac{C' 2^{2 - \gamma_0}}{\gamma_0 - 2}.
    \end{align}
    Note that a tighter lower bound could be established by using both directions of the assumed tail bounds, but doing so would not change the rate of the lower bound in $\gamma_0$. 
    Therefore, the ATE efficiency bound is bounded by
    \begin{align}
        \sigma^2_{\min} \lcb \frac{C' 2^{2 - \gamma_0}}{\gamma_0 - 2}\rcb \leq \Var\lcb \eif(Z, \eta) \rcb \leq \frac{1}{4} \lcb 2 + \frac{2C}{\gamma_0 - 2} \rcb. 
    \end{align}
    The width of a EIF-based point estimator confidence interval, up to an estimation error term converging to zero in probability, is therefore bounded by
    \begin{align}
        \frac{2q_{1-\alpha/2} \sigma_{\min}}{\sqrt{n}} \times \sqrt{\frac{C' 2^{2 - \gamma_0}}{\gamma_0 - 2}} + o_{\bbP}(1)  \leq \overline{C}_{\psi,n} - \underline{C}_{\psi, n} \leq \frac{q_{1-\alpha/2}}{\sqrt{n}} \times \sqrt{2 + \frac{2C}{\gamma_0 - 2}} + o_{\bbP}(1),
    \end{align}
    which establishes part 1 of the theorem.

    Next, consider the width of the confidence interval for the non-overlap bounds. By Proposition~\ref{prop:confidence-interval-width-bound},
    \begin{align}
        \widehat{W}_s(c, \gamma) &\leq \bbP(\pscore \leq c + \gamma) + \bbP(\pscore \geq 1 - c - \gamma) + 2\frac{q_{1-\alpha/2}}{\sqrt{n}} \sqrt{1 + \frac{1}{2c} + \frac{9D^2}{4\gamma^2}} + o_{\bbP}(1) \\
        &\leq 2C(c+\gamma)^{\gamma_0 - 1} + 2\frac{q_{1-\alpha/2}}{\sqrt{n}} \sqrt{1 + \frac{1}{2c} + \frac{9D^2}{4\gamma^2}} + o_{\bbP}(1).
    \end{align}

    For fixed (small) $\gamma$, we have
    \begin{align}
        \widehat{W}_s(c, \gamma) \lesssim c^{\gamma_0 - 1} + \frac{1}{\sqrt{nc}} + o_{\bbP}(1).
    \end{align}
    Balancing the two terms leads to
    \begin{align}
        c^{\gamma_0 - 1} \sim \frac{1}{\sqrt{nc}} \implies c^{\gamma_0 - 1/2} \sim \frac{1}{\sqrt{n}},
    \end{align}
    which suggests the optimal choice of $c$ is $c^* \sim n^{-1/(2\gamma_0 - 1)}$. Under this choice, the confidence interval width scales as
    \begin{align}
        \widehat{W}_s(c^*, \gamma) \sim n^{-(\gamma_0 - 1) / (2\gamma_0 - 1)},
    \end{align}
    which establishes part 2 of the theorem.
    
    Finally, for part 3, the non-overlap bounds CI will be shorter than the Wald CI if the lower bound on the Wald width exceeds the upper bound on the non-overlap width:
    \begin{align}
        \frac{2q_{1-\alpha/2} \sigma_{\min}}{\sqrt{n}} \times \sqrt{\frac{C'2^{2 - \gamma_0}}{\gamma_0 - 2}} >  2C(c+\gamma)^{\gamma_0 - 1} + 2\frac{q_{1-\alpha/2}}{\sqrt{n}} \sqrt{1 + \frac{1}{2c} + \frac{9D^2}{4\gamma^2}} + o_{\bbP}(1).
    \end{align}

    As $\gamma_0 \to 2^+$, and holding $n$ fixed, the left hand side diverges, while the right-hand side is bounded (because the optimal $c$ as $\gamma_0 \to 2^+$ is $c^* \to n^{-1/3}$, which is bounded from zero when $n$ is fixed, and the $o_\bbP(1)$ terms on both sides are bounded in $\gamma_0$ for fixed $n$). This implies that, for any $n$, there exists a $\gamma_0(n) > 2$ such that the above inequality is satisfied. 
\end{proof}

\section{Additional Simulation Results}
\label{section:additional-results}

\subsection{Simulation Study 1}
Tables~\ref{tab:additional-simulation-study-1-results} and \ref{tab:additional-simulation-study-1-results-power} mirror Tables~\ref{tab:simulation-study-1-results} and \ref{tab:simulation-study-1-results-power} from the main text, but with the simulation DGP overlap parameter $\alpha = 1$ rather than $\alpha = 5$. These results therefore illustrate the performance of the non-overlap bounds in a setting without extreme positivity violations. In this setting, both the non-overlap bounds and the doubly-robust one-step estimator exhibit somewhat better performance at small sample sizes, but as sample size increases the two methods have similar performance in terms of uncertainty interval width, empirical coverage, and power as sample size increases.

\begin{table}
    \centering
    \begin{tabular}{|llrrrrrrrr|}
    \hline
    & & \multicolumn{8}{c|}{95\% UI Width} \\
    & & \multicolumn{2}{c}{Mean} & \multicolumn{2}{c}{Median} & \multicolumn{2}{c}{90\% Quantile} & \multicolumn{2}{c|}{Std. Dev.} \\
    $N$ & $\gamma$ & Bounds & DR & Bounds & DR & Bounds & DR & Bounds & DR \\
    \hline
    100 & $10^{-3}$ & \textbf{0.47} & 0.80 & \textbf{0.45} & 0.47 & \textbf{0.56} & 2.00 & \textbf{0.07} & 0.64\\
     & $10^{-2}$ & \textbf{0.47} & 0.80 & \textbf{0.45} & 0.47 & \textbf{0.56} & 2.00 & \textbf{0.07} & 0.64\\
     & $10^{-1}$ & \textbf{0.46} & 0.80 & \textbf{0.45} & 0.47 & \textbf{0.54} & 2.00 & \textbf{0.06} & 0.64\\
    250 & $10^{-3}$ & \textbf{0.28} & 0.29 & 0.27 & \textbf{0.26} & 0.33 & \textbf{0.31} & \textbf{0.03} & 0.19\\
     & $10^{-2}$ & \textbf{0.28} & 0.29 & 0.27 & \textbf{0.26} & 0.33 & \textbf{0.31} & \textbf{0.04} & 0.19\\
     & $10^{-1}$ & \textbf{0.28} & 0.29 & 0.27 & \textbf{0.26} & 0.32 & \textbf{0.31} & \textbf{0.03} & 0.19\\
    500 & $10^{-3}$ & 0.19 & \textbf{0.18} & 0.19 & \textbf{0.18} & 0.21 & \textbf{0.19} & 0.02 & \textbf{0.01}\\
     & $10^{-2}$ & 0.19 & \textbf{0.18} & 0.19 & \textbf{0.18} & 0.21 & \textbf{0.19} & 0.02 & \textbf{0.01}\\
     & $10^{-1}$ & 0.19 & \textbf{0.18} & 0.19 & \textbf{0.18} & 0.21 & \textbf{0.19} & 0.02 & \textbf{0.01}\\
    1000 & $10^{-3}$ & 0.13 & 0.13 & 0.13 & 0.13 & 0.14 & \textbf{0.13} & 0.01 & \textbf{0.00}\\
     & $10^{-2}$ & 0.13 & 0.13 & 0.13 & 0.13 & 0.14 & \textbf{0.13} & 0.01 & \textbf{0.00}\\
     & $10^{-1}$ & 0.13 & 0.13 & 0.13 & 0.13 & 0.14 & \textbf{0.13} & 0.01 & \textbf{0.00}\\
    \hline
    \end{tabular}
    \caption{Uncertainty width results from Simulation Study 1 (Section~\ref{sec:simulation-study-1}) for simulations generated unconditionally from the simulation DGP with overlap parameter $\alpha = 1$. Results compare the 95\% uncertainty intervals of a doubly-robust one-step estimator (DR) and the tightest uniform 95\% non-overlap bounds (Bounds) across a grid of propensity score thresholds and smoothness parameters in terms of the mean, median, 90\%, quantile, and standard deviation of the  uncertainty interval width.}
    \label{tab:additional-simulation-study-1-results} 
\end{table}

\begin{table}
    \centering
    \begin{tabular}{|llrrrr|}
    \hline
    & & \multicolumn{2}{c}{95\% Coverage} & \multicolumn{2}{c|}{Power} \\
    $N$ & $\gamma$ & Bounds & DR & Bounds & DR \\
    \hline
    100 & $10^{-3}$ & 96.4\% & 96.4\% & \textbf{0.47} & 0.38\\
     & $10^{-2}$ & 96.4\% & 96.4\% & \textbf{0.47} & 0.38\\
     & $10^{-1}$ & 96.7\% & \textbf{96.4\%} & \textbf{0.50} & 0.38\\
    250 & $10^{-3}$ & 95.7\% & \textbf{95.0\%} & \textbf{0.90} & 0.89\\
     & $10^{-2}$ & 95.7\% & \textbf{95.0\%} & \textbf{0.90} & 0.89\\
     & $10^{-1}$ & 95.8\% & \textbf{95.0\%} & \textbf{0.93} & 0.89\\
    500 & $10^{-3}$ & 96.5\% & \textbf{95.6\%} & 1.00 & 1.00\\
     & $10^{-2}$ & 96.5\% & \textbf{95.6\%} & 1.00 & 1.00\\
     & $10^{-1}$ & 96.5\% & \textbf{95.6\%} & 1.00 & 1.00\\
    1000 & $10^{-3}$ & 95.9\% & \textbf{95.0\%} & 1.00 & 1.00\\
     & $10^{-2}$ & 95.9\% & \textbf{95.0\%} & 1.00 & 1.00\\
     & $10^{-1}$ & 95.6\% & \textbf{95.0\%} & 1.00 & 1.00\\
    \hline
    \end{tabular}
    \caption{Empirical coverage and power results from Simulation Study 1 (Section~\ref{sec:simulation-study-1}) for simulations generated unconditionally from the simulation DGP with overlap parameter $\alpha = 1$. Results compare the 95\% uncertainty intervals of a doubly-robust one-step estimator (DR) and the tightest uniform 95\% non-overlap bounds (Bounds) across a grid of propensity score thresholds and smoothness parameters in terms of empirical 95\% coverage and power.}
    \label{tab:additional-simulation-study-1-results-power} 
\end{table}

\subsection{Simulation Study 2}

Figures~\ref{fig:simulation-study-2-results-power} and \ref{fig:simulation-study-2-results-coverage} display the power and empirical 95\% coverage results from Simulation Study 2, mirroring the format of Figure~\ref{fig:simulation-study-2-results} in the main text, which focuses on mean confidence interval width. 

\begin{figure}[ht]
    \centering
    \includegraphics[width=1\linewidth]{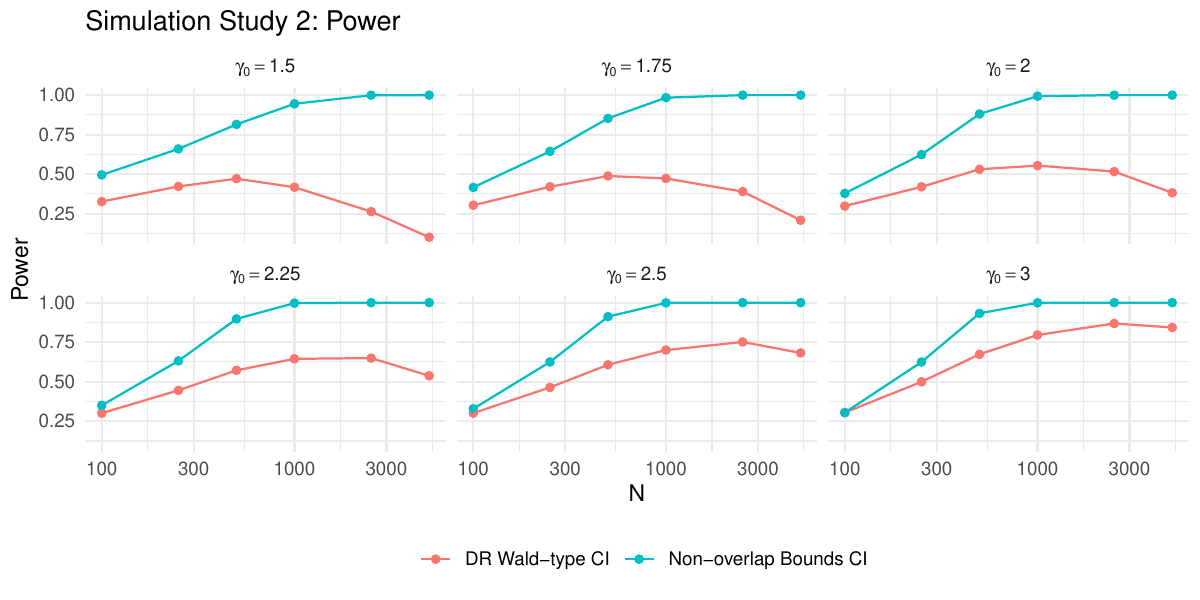}
    \caption{Simulation Study 2. Comparison of uncertainty intervals from a 95\% point estimator confidence interval (CI) centered on a doubly-robust one-step point estimate (red) and the union 95\% confidence interval formed via estimated non-overlap bounds (blue). The $x$-axis is the sample size and the $y$-axis the empirical power (probability of rejecting the null hypothesis that the population ATE is zero). Each panel shows a different setting of $\gamma_0$, the smoothly-varying tails parameter (Assumption~\ref{assumption:slowly-varying-tails}) controlling the overlap regime, with $\gamma_0 = 2$ separating very weak overlap ($\gamma_0 < 2$) from somewhat weak overlap ($\gamma_0 > 2$). }
    \label{fig:simulation-study-2-results-power}
\end{figure}

\begin{figure}[ht]
    \centering
    \includegraphics[width=1\linewidth]{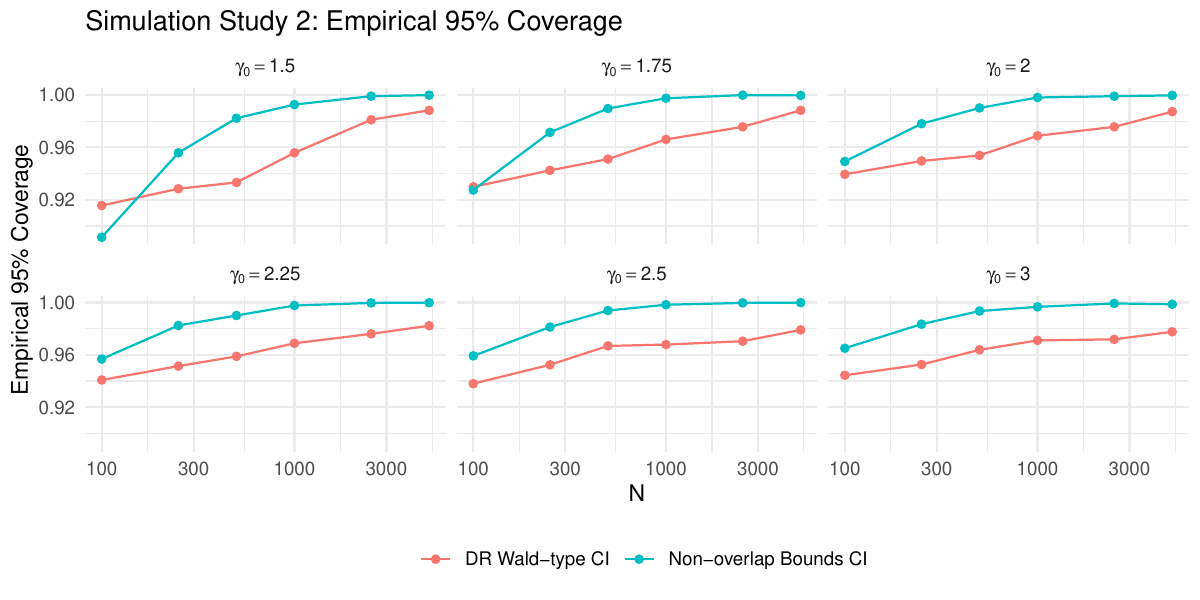}
    \caption{Simulation Study 2. Comparison of uncertainty intervals from a 95\% point estimator confidence interval (CI) centered on a doubly-robust one-step point estimate (red) and the union 95\% confidence interval formed via estimated non-overlap bounds (blue). The $x$-axis is the sample size and the $y$-axis the empirical 95\% coverage. Each panel shows a different setting of $\gamma_0$, the smoothly-varying tails parameter (Assumption~\ref{assumption:slowly-varying-tails}) controlling the overlap regime, with $\gamma_0 = 2$ separating very weak overlap ($\gamma_0 < 2$) from somewhat weak overlap ($\gamma_0 > 2$). }
    \label{fig:simulation-study-2-results-coverage}
\end{figure}

\clearpage

\section{Replication \texttt{R} code for Section~\ref{section:application}}
\label{section:application-code}
The code listing below reproduces the results from the application (Section~\ref{section:application}). Figure~\ref{fig:application-results-alternative-smoothness} shows the non-overlap bounds from the data application with the alternative smoothness parameter $\gamma = 0.001$.

\lstset{frame=tb,
    language=R,
    alsoletter={.},
    basicstyle=\footnotesize\ttfamily,
    breaklines=true,
    showspaces=false,
    showstringspaces=false,
    commentstyle=\color{teal},
    stringstyle=\color{violet}
}

\begin{lstlisting}[language=R]
# Install ATbounds package for data:
# install.packages("ATbounds")
# 
# Install effectbounds package:
# remotes::install_github("herbps10/effectbounds")

library(effectbounds)

set.seed(10016)

data("RHC", package = "ATbounds")

X <- setdiff(colnames(RHC), c("survival", "RHC"))
A <- "RHC"
Y <- "survival"

thresholds <- 10^seq(-5, log10(0.05), 0.2)
smoothness <- c(0.001, 0.01)

bounds <- ate_bounds(
  RHC, X, A, Y, 
  smoothness = smoothness, 
  thresholds = thresholds
)

summary(bounds)

title <- "Population Average Treatment Effect of \nRight Heart Catheterization on 30-day Mortality"
ylim = c(-0.2, 0.10)

pdf("application/rhc_plot.pdf", width = 8, height = 6)
plot(bounds, 
     main = title,
     ylim = ylim,
     bounds_color = "darkred", smoothness = 0.01)
dev.off()

pdf("application/rhc_plot_sensitivity.pdf", width = 8, height = 6)
plot(bounds, 
     main = title,
     ylim = ylim,
     bounds_color = "darkred", smoothness = 0.001)
dev.off()

\end{lstlisting}

\clearpage

\begin{figure}
    \centering
    \includegraphics[width=0.9\linewidth]{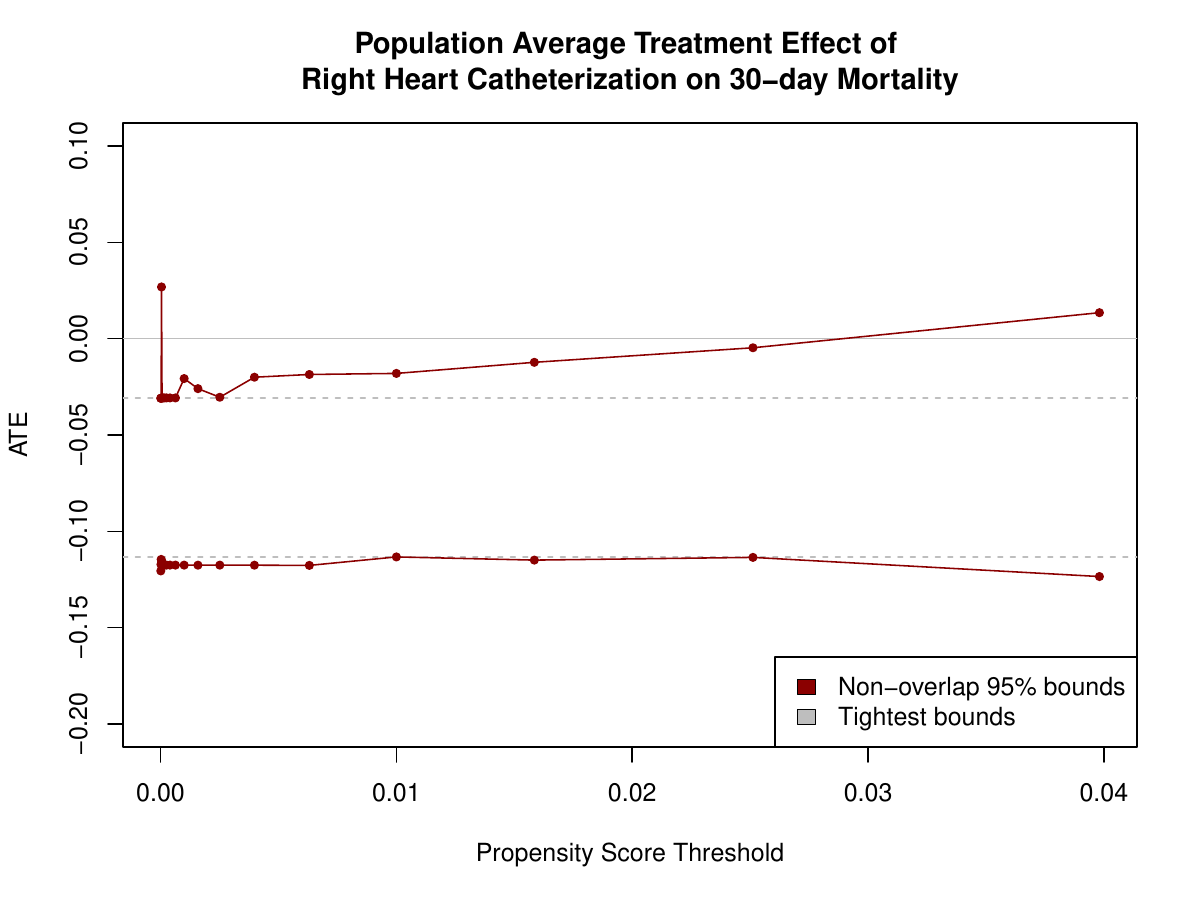}
    \caption{Uniform 95\% non-overlap bounds (for $\gamma = 0.001$) on the average treatment effect (ATE) of right heart catheterization on survival. The points illustrate the lower and upper bounds with respect to a logarithmic grid of propensity score thresholds. The lines between points are solely to guide the eye. The horizontal dotted lines indicate the tightest valid 95\% uncertainty interval that may be formed from the non-overlap bounds.}
    \label{fig:application-results-alternative-smoothness}
\end{figure}

\section{Additional Applications}
\label{appendix:additional-applications}
We applied the proposed non-overlap bounds to six publicly available reference datasets used in causal inference studies. We include the Right Heart Catheterization case study, described in detail in the main text, for comparison. For each dataset, we applied the non-overlap bounds using main-terms generalized linear models for nuisance estimation and with the smoothing parameter $\gamma = 0.001$. As a benchmark, we applied a doubly-robust one-step estimator, again using main-terms generalized linear models for nuisance estimation. Each dataset is described briefly below. All continuous outcomes were scaled to fall in the range $[0, 1]$.
\begin{itemize}
    \item \textbf{LaLonde (Dehejia-Wahba sample).} Treatment: participation in the National Supported Work Demonstration Job Training Program. Outcome: 1978 real earnings (continuous). All other variables in the dataset were treated as covariates. Source: \texttt{nsw\_mixtape} dataset from the \texttt{causaldata} package \citep{causaldata2026}. 
    \item \textbf{National Health and Nutrition Examination Survey: Epidemiologic Follow-up Study (NHEFS)}. Treatment: smoking cessation. Outcome: weight (kg). Covariates: \texttt{sex}, \texttt{age}, \texttt{race}, \texttt{education}, \texttt{ht}, \texttt{smokeintensity}, \texttt{smokeyrs}, \texttt{exercise}, \texttt{active}. Source: \texttt{nhefs} dataset from the \texttt{causaldata} package \citep{causaldata2026}.
    \item \textbf{Student-teacher Achievement Ratio (STAR).} Treatment: small vs. regular class size. Outcome: math examination scores. Covariates: \texttt{gender}, \texttt{ethnicity}, \texttt{birth}, \texttt{lunchk}. Dataset filtered to complete cases. Source: \texttt{STAR} dataset from the \texttt{AER} package \citep{kleiber2008AER}. 
    \item \textbf{Black Politicians.} Treatment: legislator race (black/non-black, \texttt{leg\_black}). Outcome: legislator responded to email. All other variables in the dataset were treated as covariates. Source: \texttt{black\_politicians} dataset from the \texttt{causaldata} package \citep{causaldata2026}.
    \item \textbf{Right Heart Catheterization (RHC).} See description in main text.
    \item \textbf{Crown Court Sentencing.} Treatment: sex (male vs. non-male). Outcome: binary indicator of whether individual was taken into custody. All other variables in the dataset were treated as covariates. Source: \texttt{ccdrug} dataset from the \texttt{causaldata} package \citep{causaldata2026}. 
\end{itemize}
Results of the analyses are shown in Table~\ref{tab:additional-example-applications}, comparing the uncertainty interval width of the non-overlap bounds and point estimator confidence intervals from the one-step estimator. Figure~\ref{fig:application-non-overlap-size} compares the non-overlap subpopulation sizes with respect to the propensity score threshold $c$ in each dataset. The results illustrate the overlap regimes analyzed in Section~\ref{sec:comparison-wald}. In two datasets where overlap is well-satisfied (LaLonde, NHEFS), the non-overlap bounds and point estimator confidence intervals have similar widths, consistent with the asymptotic equivalence result established in Theorem~\ref{thm:equiv-tmle}. In three datasets with severe overlap violations (STAR, RHC, Crown Court Sentencing), the point estimator confidence intervals span the entire parameter space while the non-overlap bounds remain informative, illustrating the regime of Theorem~\ref{thm:very-weak-overlap}. In one dataset (Black Politicians), the non-overlap bound confidence intervals are wider than the point estimator confidence interval, consistent with the finite-sample cost identified in Theorem~\ref{thm:equiv-tmle}, when the non-overlap population is relatively large. 

\begin{table}[ht]
    \centering
    \begin{tabular}{lrrrrrrrr}
             &             &            &               & \multicolumn{2}{c}{Non-overlap CI} & \multicolumn{2}{c}{Doubly-robust CI} \\
        Name & $n$ & $p$ & $\bbP(A = 1)$ & CI & Width & CI & Width \\
        \hline
        LaLonde                        & 445   & 8   & 42\% & $(0.01,0.05)$    & 0.04 & $(0.01, 0.05)$   & 0.05  \\
        NHEFS                          & 1566  & 10  & 26\% & $(0.02,0.05)$    & 0.02 & $(0.03,0.05)$    & 0.02 \\
        STAR       & 3781  & 4   & 46\% & $(0.01,0.03)$    & 0.02 & $(-1.00,1.00)$   & 2.00 \\
        Black Politicians              & 5593  & 12  & 7\%  & $(-0.34,0.40)$   & 0.74 & $(-0.39, 0.27)$ & 0.64 \\
        RHC    & 5735  & 72  & 38\% & $(-0.11, -0.03)$ & 0.08 & $(-1.00,1.00)$   & 2.00 \\
        Crown Court Sentencing         & 16973 & 43  & 93\% & $(0.09-0.15)$    & 0.06 & $(-1.00,1.00)$   & 2.00 \\
        \hline
    \end{tabular}
    \caption{Comparison of confidence interval widths for the ATE from non-overlap bounds and a doubly-robust estimator on six observational datasets. Sample sizes ($n$), number of covariates ($p$), and marginal probability of treatment ($\bbP(A = 1)$) are reported for each dataset. }
    \label{tab:additional-example-applications}
\end{table}

\begin{figure}[ht]
    \centering
    \includegraphics[width=0.9\linewidth]{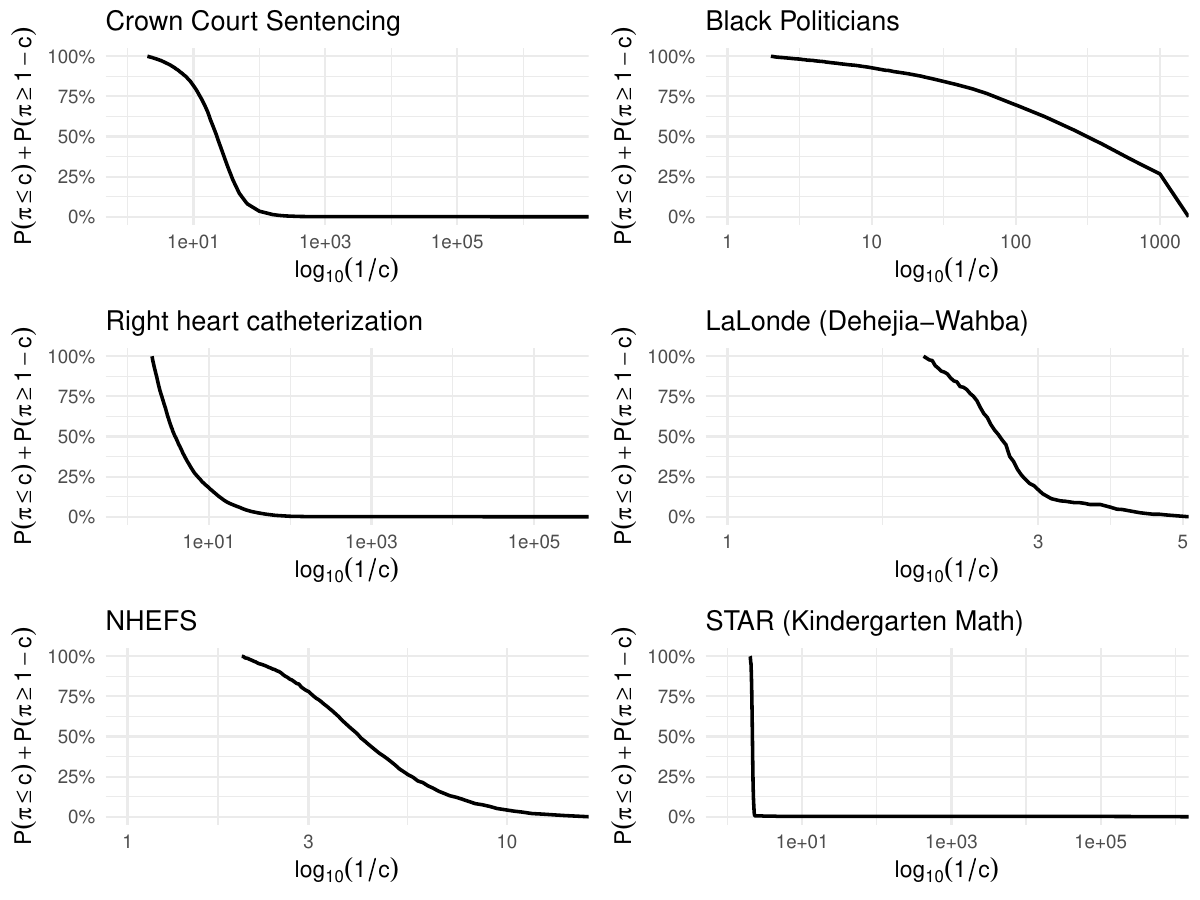}
    \caption{Non-overlap subpopulation size ($\bbP(\pscore \leq c) + \bbP(\pscore \geq 1 - c)$} versus $\log_{10}(1/c)$ for a range of propensity score thresholds $c$ in each application dataset.
    \label{fig:application-non-overlap-size}
\end{figure}

\section{Sensitivity Analysis Extension}
\label{appendix:sensitivity-analysis}

A natural extension to the proposed worst-case bounds (Proposition~\ref{prop:bounds}) is to conduct a sensitivity analysis that considers bounds if the expected counterfactual outcomes in the non-overlap bounds do not take their worst-case values. Specifically, we assume the following sensitivity model:

\medskip

\begin{assumption}[Sensitivity model]
\label{assumption:sensitivity}
Suppose Assumption~\ref{assumption:bounded} holds.
Fix sensitivity parameters $\delta_0 = [\delta_{0,L}, \delta_{0,U}] \subset [0, 1]$ and $\delta_1 = [\delta_{1,L}, \delta_{1,U}] \subset [0, 1]$, and assume that
\begin{align}
    \E ( Y^0 \mid \pi \geq 1 - c )  \in \delta_0 \quad \text { and } \quad \E ( Y^1 \mid \pscore \leq c ) \in \delta_1.
\end{align}
\end{assumption}

Based on the sensitivity model, a similar analysis as in Proposition~\ref{prop:bounds} yields lower and upper ATE bounds given by
\begin{align}
    L(c, \delta_0, \delta_1) &= \psi(c) + \delta_{1,L} \times  \bbP(\pscore \leq c) - \delta_{0,U} \times \bbP(\pscore \geq 1 - c), \\
    U(c, \delta_0, \delta_1) &= \psi(c) + \delta_{1,U} \times  \bbP(\pscore \leq c) - \delta_{0,L} \times \bbP(\pscore \geq 1 - c).
\end{align}
These bounds straightforwardly generalize the worst-case analysis of Proposition~\ref{prop:bounds}, the results of which can be recovered by setting $\delta_0 = \delta_1 = [0, 1]$.
For the lower sensitivity bound to be positive (indicating that the population ATE is positive), the sensitivity parameters must satisfy either
\begin{align}
    \delta_{1,L} > \frac{\delta_{0,U} \times \bbP(\pscore \geq 1 - c) -\psi(c)}{\bbP(\pscore \leq c)} \quad \text{or} \quad \delta_{0,U} < \frac{\delta_{1,L} \times  \bbP(\pscore \leq c) + \psi(c)}{\bbP(\pscore \geq 1 - c)}.
\end{align}
Similarly, for the upper sensitivity bound to be negative (indicating the population ATE is negative), the sensitivity parameters must satisfy either
\begin{align}
    \delta_{1,U} < \frac{\delta_{0,L} \times \bbP(\pscore \geq 1 - c) -\psi(c)}{ \bbP(\pscore \leq c)} \quad \text{or} \quad  \delta_{0,L} > \frac{\delta_{1,U} \times  \bbP(\pscore \leq c) + \psi(c)}{\bbP(\pscore \geq 1 - c)}.
\end{align}
This allows for a fine-grained sensitivity analysis of which conditions on the expected potential outcomes in the non-overlap regions yield non-zero population treatment effects. This approach is in the spirit of \cite{cornfield1959cancer}, who asked how strong an unmeasured confounder would need to be to explain an estimated average treatment effect adjusted for known confounders.

\medskip 

The above conditions can be easily reparameterized as conditions on the \textit{treatment effect} in a non-overlap region. For example, note that $\E\left\{ Y^1 - Y^0 \mid \pscore \leq c \right\} \geq \delta_{1,L} - \E\left\{ Y^0 \mid \pscore \leq c \right\}$ (where the latter expectation satisfies overlap by definition). On this scale, the sensitivity parameter is the size of the treatment effect in the non-overlap subpopulation (that is, the subpopulation satisfying $\pscore \leq c$).

\medskip

Estimators for the sensitivity bounds can be constructed using smooth approximations in a similar manner as for the worst-case bounds. Here, more care is required in constructing the smooth approximations of the bounds such that they maintain the property of containing the true bounds regardless the choice of smoothing parameter $\gamma$. The derivations for the smooth approximations and the TMLE estimators for the sensitivity parameters are presented below.

\subsubsection*{Smooth sensitivity bounds}
The construction of smooth approximations for the sensitivity bounds requires smooth approximations that satisfy a complementary property to Property~\ref{req:smooth-approach}:
\begin{property}
    \label{req:smooth-approach-prime} 
    For smooth approximations $\smooth_\ell$ and $\smooth_g$, it holds that $\smooth_{\ell}(x, c, \gamma) \geq \I(x < c)$ and $\smooth_g(x, c, \gamma) \geq \I(x > c)$.
\end{property}
Intuitively, Property~\ref{req:smooth-approach} requires that the smooth approximations approach the indicator functions ``from below'' (see Figure~\ref{fig:smooth-approximation}), while Property~\ref{req:smooth-approach-prime} requires that the approximations approach ``from above''. 

Pick $s_\ell$ and $s_g$ to satisfy Property~\ref{req:smooth-approach}, and $s_\ell^*$ and $s_g^*$ satisfy Property~\ref{req:smooth-approach-prime}. The lower and upper smooth sensitivity bounds are given by
\begin{align}
    \label{eq:smooth-sensitivity-bound-lower}
    L_s(c, \gamma, \delta_0, \delta_1) &=  \psi_s(c, \gamma) + \delta_{1,L}\times\left[ 1 - \E\left\{ s_g^*(\pscore, c, \gamma) \right\} \right]  - \delta_{0, U} \times \textcolor{purple}{\left[ 1 - \E\left\{ s_\ell(\pscore, 1 - c, \gamma) \right\} \right]} \\
    \label{eq:smooth-sensitivity-bound-upper}
    U_s(c, \gamma, \delta_0, \delta_1) &= \psi_s(c, \gamma) + \delta_{1, U} \times \textcolor{purple}{\left[ 1 - \E\left\{ s_g(\pscore, c, \gamma) \right\} \right]} - \delta_{0,L} \times \left[1 - \E\left\{ s_\ell^*(\pscore, 1 - c, \gamma) \right\} \right],
\end{align}
where $\psi_s(c, \gamma)$ is defined in \eqref{eq:smooth-ate}. Note that if $\delta_0 = \delta_1 = [0, 1]$, the smooth sensitivity bounds reduce to the worst-case smooth non-overlap bounds \eqref{eq:smooth-lower-bound} and \eqref{eq:smooth-upper-bound}. In the above display, the components of the smooth sensitivity bounds that arise in the worst-case smooth non-overlap bounds are highlighted.

The next proposition generalizes Proposition~\ref{prop:smooth-bounds} to the smooth sensitivity bounds.
\begin{proposition}[Smooth sensitivity bounds] \label{prop:smooth-bounds-prime}
    Under the sensitivity model (Assumption~\ref{assumption:sensitivity}) and the conditions of Proposition~\ref{prop:bounds}, suppose $s_{\ell}(x, c, \gamma)$ and $s_g(x, c, \gamma)$ satisfy Property~\ref{req:smooth-approach} and $s_{\ell}^*(x, c, \gamma)$ and $s_g^*(x, c, \gamma)$ satisfy Property~\ref{req:smooth-approach-prime}. Then,
    \begin{align}
        \E \left( Y^1 - Y^0 \right) \in \Big[ \smoothL(c, \gamma, \delta_0, \delta_1), \smoothU(c, \gamma, \delta_0, \delta_1) \Big],
    \end{align}
    where $\smoothL(c, \gamma, \delta_0, \delta_1)$ and $\smoothU(c, \gamma, \delta_0, \delta_1)$ are defined in \eqref{eq:smooth-sensitivity-bound-lower} and \eqref{eq:smooth-sensitivity-bound-upper}, respectively.
\end{proposition}
\begin{proof}
    We show that $L(c, \delta_0, \delta_1) \geq L_s(c, \gamma, \delta_0, \delta_1)$; similar steps can be taken to show that $U(c) \leq U_s(c, \gamma, \delta_0, \delta_1)$. If $\delta_{1,L} = 0$, the analysis in the proof of Proposition~\ref{prop:smooth-bounds} can be applied to show that 
    \begin{align}
        L(c, \delta_0, \delta_1) &= \psi(c) - \delta_{0,U} \times \bbP(\pscore \geq 1 - c) \\
        &\geq \psi_s(c, \gamma) - \delta_{0, U} \times \left[ 1 - \E\left\{ s_\ell(\pscore, 1 - c, \gamma) \right\} \right]. 
    \end{align}
    By Property~\ref{req:smooth-approach-prime}, $\E\lcb s_{g}^*(\pscore, c, \gamma) \rcb \geq \bbP(x > c)$, implying that 
    \begin{align}
        \delta_{1,L} \times \bbP(\pscore \leq c) \geq \delta_{1,L} \times \left[1 - \E\lcb \smooth_{g}^*(\pscore, c, \gamma) \rcb \right].
    \end{align}
    Therefore,
    \begin{align}
        L(c, \delta_0, \delta_1) &= \psi(c) + \delta_{1,L} \times  \bbP(\pscore \leq c) - \delta_{0,U} \times \bbP(\pscore \geq 1 - c) \\
        &\geq \psi_s(c, \gamma) + \delta_{1,L} \times \left[1 - \E\lcb s_{g}^*(\pscore, c, \gamma) \rcb \right] - \delta_{0, U} \times \left[ 1 - \E\left\{ s_\ell(\pscore, 1 - c, \gamma) \right\} \right] \equiv L_s(c, \gamma, \delta_0, \delta_1).
    \end{align}
\end{proof}
For the lower smooth sensitivity bound to be positive (and, by Proposition~\ref{prop:smooth-bounds-prime}, for the lower sensitivity bound to be likewise positive), it must hold that either:
\begin{align}
    \label{conditions:smooth-lower-positive}
    \delta_{1,L} > \frac{\delta_{0,U} \times \left[ 1 - \E\lcb \smooth_\ell(\pscore, 1 - c, \gamma) \rcb \right] -\psi_s(c, \gamma)}{1 - \E\lcb \smooth_g^*(\pscore, c, \gamma) \rcb} \quad \text{or} \quad \delta_{0,U} < \frac{\delta_{1,L} \times \left[ 1 - \E\lcb \smooth_g^*(\pscore, c, \gamma) \rcb \right] + \psi_s(c, \gamma)}{1 - \E\lcb \smooth_\ell(\pscore, 1 - c, \gamma) \rcb}. 
\end{align}
For the upper smooth sensitivity bound to be negative, it must hold that either:
\begin{align}
    \label{conditions:smooth-upper-negative}
    \delta_{1,U} < \frac{\delta_{0,L} \times \left[1 - \E\left\{ s_\ell^*(\pscore, 1 - c, \gamma) \right\} \right] -\psi_s(c, \gamma)}{ 1 - \E\left\{ s_g(\pscore, c, \gamma) \right\} } \quad \text{or} \quad  \delta_{0,L} > \frac{\delta_{1,U} \times  \left[ 1 - \E\left\{ s_g(\pscore, c, \gamma) \right\} \right] + \psi_s(c, \gamma)}{1 - \E\left\{ s_\ell^*(\pscore, 1 - c, \gamma) \right\}}. 
\end{align}

\subsubsection*{Efficient influence functions}
The following lemma establishes the EIFs for the lower and upper smooth sensitivity bounds.
\begin{lemma}[Efficient influence functions smooth sensitivity bounds.]
    \label{lemma:eif-sensitivity-bounds}
    Fix $c \in \left[0, \tfrac{1}{2} \right]$. The parameters $\smoothL(c, \gamma, \delta_0, \delta_1)$ and $\smoothU(c, \gamma, \delta_0, \delta_1)$ are pathwise differentiable with uncentered efficient influence function given by
    \begin{align}
        \eif_{\smoothL}(Z, \delta_0, \delta_1) &= \eif_{\smoothATE}(Z) +\delta_{1,L}\times\left[ 1 - \E\left\{ s_g^*(\pscore, c, \gamma) \right\} \right] - \delta_{0, U} \times \left[1 - \smooth_{\ell}(\pscore, 1 - c, \gamma)\right] \\
        & \quad + \lcb \delta_{0,U} \times \dot{\smooth}_\ell(\pscore, 1 - c, \gamma) - \delta_{1,L} \times \dot{s}_g(\pscore, c, \gamma) \rcb \lcb A - \pscore \rcb \\
        \eif_{\smoothU}(Z, \delta_0, \delta_1) &= \eif_{\smoothATE}(Z) + \delta_{1,U} \times \left[ 1 - \smooth_g(\pscore, c, \gamma) \right] - \delta_{0,L} \times \left[1 - \E\left\{ s_\ell^*(\pscore, 1 - c, \gamma) \right\} \right] \\
        & \quad + \lcb \delta_{0,L} \times \dot{\smooth}_{\ell}(\pscore, 1 - c, \gamma) - \delta_{1,U} \times \dot{\smooth}_g(\pscore, c, \gamma) \rcb \lcb A - \pscore \rcb.
    \end{align}
\end{lemma}
\begin{proof}
    The proof follows using similar techniques as in the proof of Lemma~\ref{lemma:eif-upper}.
\end{proof}

\subsubsection*{TMLE}
The TMLE for the sensitivity bounds is a minor modification of the TMLE for the worst-case non-overlap bounds. The only difference is in the fluctuation model, where we now set
\begin{align}
    M &= \delta_{0,U} \times \dot{\smooth}_\ell(\pscore, 1 - c, \gamma) - \delta_{1,L} \times \dot{s}_g(\pscore, c, \gamma) \\
    \text{ or } M &= \delta_{0,L} \times \dot{\smooth}_{\ell}(\pscore, 1 - c, \gamma) - \delta_{1,U} \times \dot{\smooth}_g(\pscore, c, \gamma) 
\end{align}
when targeting the lower and upper smooth sensitivity bounds, respectively. 

%In addition, it may be of interest to directly estimate the minimal conditions for the lower (upper) smooth sensitivity bounds to to be positive (negative), using the formulas in \eqref{conditions:smooth-lower-positive} and \eqref{conditions:smooth-upper-negative}. To do so, separate targeted estimators can be formed for $\psi_s(c, \gamma)$ and $\E\lcb s(\pscore, c, \gamma) \rcb$ (for generic smooth approximation $s$) using similar techniques as for the lower and upper smooth sensitivity bounds. 
Point estimates can by plugging in estimates to \eqref{conditions:smooth-lower-positive} and  \eqref{conditions:smooth-upper-negative}, with variance estimates derived via the Delta method. 

\end{appendix}

\end{document}